\documentclass[12pt,a4paper,reqno]{amsart}
\usepackage{amsmath,amssymb,graphics,epsfig,color,enumerate,mathrsfs}
\usepackage{dsfont}  \usepackage{verbatim}
\definecolor{refkey}{gray}{.78}
\definecolor{labelkey}{gray}{.78}

\theoremstyle{plain}

\let\so=\o

\newtheorem{Theorem}{Theorem}[section]

\newtheorem{Lemma}[Theorem]{Lemma}
\newtheorem{Proposition}[Theorem]{Proposition}
\newtheorem{Corollary}[Theorem]{Corollary}

\newtheorem{Claim}[Theorem]{Claim}
\newtheorem{Warning}{Warning}[section]

\theoremstyle{remark}
\newtheorem{Definition}[Theorem]{Definition}
\newtheorem{Remark}[Theorem]{Remark}

 \definecolor{darkgreen}{rgb}{0,0.4,0}

\definecolor{light}{gray}{0.9}

\newcommand{\cA}{\ensuremath{\mathcal A}}
\newcommand{\cB}{\ensuremath{\mathcal B}}
\newcommand{\cC}{\ensuremath{\mathcal C}}

\newcommand{\cE}{\ensuremath{\mathcal E}}
\newcommand{\cF}{\ensuremath{\mathcal F}}
\newcommand{\cG}{\ensuremath{\mathcal G}}
\newcommand{\cH}{\ensuremath{\mathcal H}}

\newcommand{\cL}{\ensuremath{\mathcal L}}

\newcommand{\cN}{\ensuremath{\mathcal N}}

\newcommand{\cP}{\ensuremath{\mathcal P}}

\newcommand{\cV}{\ensuremath{\mathcal V}}
\newcommand{\cW}{\ensuremath{\mathcal W}}

\newcommand{\bbB}{{\ensuremath{\mathbb B}} }

\newcommand{\bbE}{{\ensuremath{\mathbb E}} }

\newcommand{\bbG}{{\ensuremath{\mathbb G}} }

\newcommand{\bbI}{{\ensuremath{\mathbb I}} }

\newcommand{\bbL}{{\ensuremath{\mathbb L}} }

\newcommand{\bbN}{{\ensuremath{\mathbb N}} }

\newcommand{\bbR}{{\ensuremath{\mathbb R}} }

\newcommand{\bbZ}{{\ensuremath{\mathbb Z}} }

\let\a=\alpha \let\b=\beta   \let\d=\delta  \let\e=\varepsilon
 \let\g=\gamma       \let\l=\lambda
      \let\o=\omega      
  \let\s=\sigma \let\t=\tau   
  \let\z=\zeta
\let\D=\Delta   \let\G=\Gamma  \let\L=\Lambda 
\let\O=\Omega      
\newcommand{\rosso}{\textcolor{black}} 

\newcommand{\da}{\downarrow}

\newcommand{\toup}{\rightharpoonup}

\newcommand{\be}{\begin{equation}}
\newcommand{\en}{\end{equation}}
\newcommand{\bes}{\begin{equation*}}
\newcommand{\ens}{\end{equation*}}

\thanks{This work has been partially supported by the ERC Starting Grant 680275 MALIG and by the Grant  PRIN 20155PAWZB "Large Scale Random Structures"}

\author[A.~Faggionato]{Alessandra Faggionato}
\address{Alessandra Faggionato.
  Dipartimento di Matematica, Universit\`a di Roma `La Sapienza'
  P.le Aldo Moro 2, 00185 Roma, Italy}
\email{faggiona@mat.uniroma1.it}

\newcommand{\ra}{\rangle}
\newcommand{\la}{\langle}

\title[Conductivity  of random resistor networks]{Scaling limit of the directional conductivity of random resistor networks on simple point processes}

\begin{document}
\maketitle

\begin{abstract}  We consider  random resistor networks with nodes given by a  point process on  $\mathbb{R}^d$ and with random conductances. The length range of the     electrical filaments   can be unbounded.  We assume that the randomness is stationary and ergodic w.r.t.~the action of the group  $\mathbb{G}$,  given by  $\mathbb{R}^d$ or $\mathbb{Z}^d$. This action is covariant w.r.t.~ translations on the Euclidean space. 
Under  minimal assumptions  we prove that a.s. the suitably rescaled directional conductivity of the resistor network   along the principal directions  of the effective homogenized matrix $D$ converges to the corresponding eigenvalue of $D$ times  the intensity of the  point process.   \rosso{More generally, we prove a quenched scaling limit of the directional conductivity along any vector $e\in  {\rm Ker}(D) \cup {\rm Ker}(D)^\perp$}.
  Our results  cover plenty of models including e.g. the standard conductance model on $\mathbb{Z}^d$ \rosso{(also with long filaments)}, the Miller-Abrahams resistor network for conduction in amorphous solids (to which we can now extend the bounds  in agreement with Mott's law   previously  obtained in   \cite{CP1,FM,FSS} for Mott's random walk),    resistor networks on the supercritical  cluster in lattice and continuum percolations,  resistor networks  on crystal lattices and  on Delaunay triangulations.

\smallskip


\smallskip


\end{abstract}



\section{Introduction}
Random resistor networks in $\bbR^d$ are an effective tool to investigate transport in disordered media and have been much investigated both in Physics and Probability (cf.~e.g.~\cite{Bi,Ke,Kirk}).  
Randomness can affect both the conductances of the electric filaments and the location of the nodes. It describes  micro-inhomogeneities which can be of different physical nature. For example, one can consider mixtures of conducting and non conducting materials (cf.~\cite[Section~II]{Kirk}),  thus motivating the study of  resistor networks on percolation clusters. One can also consider amorphous solids as  doped semiconductors  
in the regime of strong Anderson localization. In this case the doping impurities have random positions $x_i$ (described mathematically by a simple\footnote{As in  \cite{DV} the adjective \emph{simple} just means that points have unit multiplicity.}
 point process)   and  the Mott's variable range hopping (v.r.h.) of the  conducting electrons can be modeled by the  Miller-Abrahams (MA) random resistor network. In this network   nodes are given by the impurities and, between any  pair of impurities located at $x_i$ and $x_j$,   there is an electrical filament of conductance 
\be\label{birillo}
c_{x_i,x_j}(\o):=\exp\Big\{ - \frac{2}{\gamma} |x_i-x_j| -\frac{\b}{2} ( |E_i|+ |E_j|+ |E_i-E_j|) \Big\}
\,,
\en
where $\g$ is the localization length, $\b=1/kT$ is the inverse temperature and the $E_i$'s are random energy marks associated to the impurities  (cf.~\cite{AHL,MA,POF,SE} and Section \ref{sec_examples} below). Typically,  the energy marks are taken as i.i.d. random variables with distribution   $p(E) \propto |E|^{\a}dE$ on an interval containing the origin, for some $\a\geq 0$.
Mott's v.r.h. was   introduced by Mott  to explain the anomalous low-temperature conductivity decay, 
  which (for isotropic materials) behaves as 
\be\label{mott-law} 
\s(\b) \asymp A \exp \left\{-c \b^\frac{\a+1}{\a+1+d}   \right\}\,,
\en
where  $A$ has a negligible $\b$-dependence, while  $c>0$ is $\beta $-independent (Mott considered the case $\a=0$, while Efros and Shklovskii introduced  $\a$  to  model a possible Coulomb pseudogap in the density of states).  Eq.~\eqref{mott-law}
  is  usually named Mott's law. We refer to Mott's Nobel Lecture \cite{NL} and  the monographs \cite{POF,SE} for more details.

As in Mott's law, 
a fundamental physical quantity is given by the conductivity of the resistor network along a given direction. One considers a box centered at the origin of $\bbR^d$ with two opposite faces orthogonal to the fixed direction, where the electric potential takes value $0$ and $1$, respectively. Then the conductivity is the total electric current  flowing across any section  orthogonal to the given direction  and equals the total dissipated energy (cf.~\cite[Section~1.3]{Bi}, \cite[Section~11]{Ke} and Section~\ref{MM} below).

In this work  we consider   generic random resistor networks on $\bbR^d$ with random conductances and nodes at random positions, hence described by a  simple point process. The electric filaments can be arbitrarily long. We assume that the randomness is stationary and ergodic w.r.t.~the action of the
group  $\mathbb{G}$,  given by  $\mathbb{R}^d$ or $\mathbb{Z}^d$. This action is covariant w.r.t.~the action 
of $\bbG$ by translations on the Euclidean space.   Under  minimal assumptions
we prove that, as the  box size diverges,  a.s. the suitably rescaled directional conductivity of the resistor network along the principal directions  of the effective homogenized matrix $D$ converges to the corresponding eigenvalue of $D$ times  the intensity of the simple point process \rosso{(cf.~Corollary~\ref{cor_airone})}. \rosso{More generally, we prove a quenched scaling limit of the directional conductivity along any vector $e\in  {\rm Ker}(D) \cup {\rm Ker}(D)^\perp$ (cf.~Theorems~\ref{teo1} and \ref{teo3})}.   $D$ admits a variational characterization which can be used to get upper and lower bounds (cf.~e.g.~\cite{demasi,FM,FSS}). \rosso{Information} on the limiting behavior of the electrical potential is provided in  Proposition~\ref{teo2}. We point out that our finite-moment conditions (A7)  in Section \ref{MM} are the optimal ones, as they are necessary to define the effective homogenized matrix (integrability of $\l_2$ in (A7))
  and the
space  of square integrable forms (integrability of $\l_0$ in (A7)).

Our target has been to achieve a  universal qualitative result, hence  our theorems apply to a large variety of geometric structures.  In particular we obtain the 
scaling limit of the conductivity for the resistor network on the  supercritical percolation cluster on $\bbZ^d$, providing a solution to 
the open Problem~1.18 in \cite{Bi}   and going far beyond Bernoulli bond percolation (cf. also  \cite{Abe,CC,GK,Ke} and references therein). As a byproduct, we get also  a proof of the strictly positivity of $D$ for this model alternative to the original one in \cite{demasi}  (in general, our results can be used to derive the non degeneracy of $D$  from information on  the disjoint crossings in the resistor network as in Section~\ref{ciuffolotto}). 
 Our results cover also the MA resistor network 
 allowing to extend to its  asymptotic directional  conductivities the bounds  in agreement with Mott's law \eqref{mott-law}   previously  obtained in 
 \cite{CP1,FM,FSS} for  Mott's random walk as detailed in Corollary \ref{daie} in Section \ref{sec_examples}  (Theorem \ref{teo1} will be  used also to fully prove Mott's law for several environments in ~\cite{Mott_sinfonia}). For  the standard conductance model on $\bbZ^d$ our results improve  the existing ones (see~\cite{Koz} and the discussion below). The above examples, \rosso{admissible stationary stochastic lattices  and periodic structures are treated in 
  Section \ref{sec_examples},  together with   a brief discussion of} other examples, as \rosso{random} resistor networks    on crystal lattices,  on Delaunay triangulations \cite{FT}, on  supercritical clusters in continuum percolation \rosso{\cite{MR}}.
We point out that  there isn't a   prototype random resistor network  to deal with as a leading example.   For example, the underlying graph of the MA resistor network is the complete graph on an infinite simple point process, which is completely different from  the  supercritical percolation cluster \rosso{on $\bbZ^d$}.  This geometric heterogeneity requires a  geometric  abstract setting to get   the desired universality. 
 
 Our  proof is based on stochastic homogenization via 2-scale convergence \rosso{and the theory of simple point processes.  2-scale convergence has been introduced in \cite{A,Nu} for rapidly oscillating operators (and further extended e.g. also in \cite{Z}). 
 Stochastic 2-scale convergence in mean has been introduced  in \cite{BMW}. Stochastic 2-scale convergence (not in mean) has been   developed   in \cite{ZP} to treat random singular structures as networks and random measures (inspired also by \cite{Z}). The definition given in  \cite{ZP} uses the Palm distribution associated to the random measure and is built  upon the pointwise ergodic theorem  (instead of convergence in mean). 
   In \cite{Fhom3} we have further extended the analysis in \cite{ZP}    to treat reversible long-range random walks on simple point processes (for nearest-neighbor random walks on $\bbZ^d$ see also \cite{MP}). Differently from \cite{ZP} where the gradient of a function is  a   vector-valued function,   in \cite{Fhom3} the variation of a function along all possible jumps is encoded in an  amorphous  gradient (cf.~Section~\ref{MM} below), requiring a separate definition of  2-scale convergence (cf.~Definitions~\ref{priscilla} and \ref{sandalo65} below). Moreover, in \cite{Fhom3} and here as well, we do not restrict to probability spaces which are  compact metric spaces as in \cite{ZP} and  the environment-dependent  test-functions and test-forms in the definition of 2-scale convergence have to be  chosen carefully. In the rest we will give a self-contained discussion of the    stochastic  2-scale convergence used here (see Sections \ref{sec_tipetto} and \ref{anatre12}). }
    We recall that in \cite{ZP} Piatnitski and Zhikov  have  proved   homogenization  for the massive Poisson equation $\l u+ \bbL u=f$ by 2-scale convergence  on bounded domains also with mixed Dirichlet-Neumann b.c., $\bbL$ being the generator of a  diffusion in random environments. 
  In  \cite[Section 7]{ZP}  the above  result  has  been applied   to get that the magnitude of the  effective homogenized  matrix  $D$  equals the limiting rescaled  ``directional  conductivity" 
 for a  diffusion on the skeleton of the supercritical percolation cluster in Bernoulli bond percolation. The proof relies on the 
     a priori check  that $D>0$, based on previous results on  left--right crossings valid in the Bernoulli case.    
We have developed here a direct proof of the  scaling limit of the direction conductivity, which  avoids  the   constraint  $D>0$ (whose check usually requires further assumptions)
    and  previous investigations of the massive Poisson equation (which would require the  cut-off procedures developed in  \cite[Sections~15,17]{Fhom3} in order to deal with  arbitrarily long conductances).
    We refer to \cite[Section~5]{Fhom3} for sufficient conditions assuring that $D>0$ in specific examples \rosso{and to \cite[Appendix~A]{F_SEP} for a model with degenerate but non zero effective homogenized matrix $D$}.
     In general, we avoid any assumption on the left-right crossings of the resistor network (usually of  difficult investigation if the FKG inequality is violated as in the MA resistor network).
 We stress that the existing proofs for random diffusions also with different b.c.~(cf.~\cite{JKO,ZP}) do not  adapt well to our general discrete setting as the presence of arbitrarily long conductances in an amorphous setting  forces to deal with amorphous local gradients, which keep trace of the function variation along any filament exiting from any given point and  which are very irregular objects.

Concerning  previous results, we point out that  the case of i.i.d. random conductances between nearest-neighbor sites of  $\bbZ^d$ with value in a fixed  interval $(\d_0,1-\d_0)$, $\d_0>0$, has been considered by \rosso{Kozlov in \cite{Koz}.}
As stated for example in \cite{BSW},  in 
the case of stationary ergodic random conductances between nearest-neighbor sites of  $\bbZ^d$ having value in $(\d_0,1-\d_0)$ ($\d_0>0$) and with potential at the boundary of the box  given by a fixed linear function, one can prove the scaling limit of the dissipated energy  by adapting  the methods developed for  the continuous case (cf.~\cite{JKO,Koz78,PV} and  the technical results  in \cite{Kunn}).  We also point out that previously, in \cite[p.\,26]{Z0},  Zhikov obtained the scaling limit of the ``directional  conductivity"  of  the standard diffusion  with partial  Dirichlet   b.c.  in   a  perforated domain built  by fattening     the supercritical percolation cluster. We point out that  the b.c. in \cite[Eq.~(1.22)]{Z0} does not correspond to  the  effective one for the resistor network on the supercritical percolation cluster as the Neumann part is missing. \rosso{Other homogenization results for  resistor networks in $\bbZ^d$ can be found in \cite{PR} and for planar polygonal networks in \cite{V} (with different b.c.).}

  \rosso{In the last decades there has been a lot of work on homogenization of  functionals on random  networks by means of $\Gamma$--convergence methods (see e.g. \cite{AC,ACG,B,BC,BF,BP,PR} and references therein). The directional conductivity and the electric potential are indeed respectively  the minimum  and the minimizers of the  energy functional associated to the resistor network (cf.~Lemma~\ref{benposto}). The boundary conditions  assumed in the literature are usually different from the ones considered in the present context and the assumptions are more restrictive in particular for the simple point process (cf. e.g.~\cite[Proposition~2.14]{PR} for the $\G$--convergence of the energy functional associated to  a resistor network on boxes of  $\bbZ^d$  with Neumann b.c. and \cite{ACG} for the $\G$--convergence of energy functionals on more general random networks  with affine b.c. and also with long-range interactions, but where interactions are only internal and not  between inside and outside  of the domain in consideration). In  Section \ref{astl} we consider  the class of  admissible stationary stochastic lattices (which is a class of models with  geometric randomness treated in \cite{BBL,ACG} but without dense clusters or big holes)  and  compare our modeling with the one developed in \cite{ACG}. We point out that our results apply as well to simple point processes with dense clusters and big holes. For  $\G$-convergence  results on energy functionals on  non regular simple point processes as the Poisson one see e.g. \cite{BC,BP}. }

As a further step of investigation we plan to derive quantitative results on the scaling of the directional conductivity at cost of   additional technical assumptions  (cf. e.g. \cite{AD} for some quantitative stochastic homogenization results on  the supercritical percolation cluster on $\bbZ^d$). 
Finally, we point out that the present work is an extension and improvement of our unpublished notes   \cite{1luglio}.\\

\noindent
{\bf Outline of the paper}. In Section \ref{MM} we present our models and main results (i.e.~Theorems~\ref{teo1}, \ref{teo3} \rosso{and} Proposition~\ref{teo2}). In Section \ref{sec_examples} we discuss some relevant examples. The rest of the paper is devoted to proofs.


\section{Models and main results}\label{MM}
We start with a probability space $(\O,\cF,\cP)$ encoding all the randomness of the system. Elements $\o$ of $\O$ are called \emph{environments}. We denote by $\bbE[\cdot]$ the expectation associated to $\cP$.

We 
 denote by  
   $\cN$ the space of  locally finite subsets $\{x_i\} \subset \bbR^d$, $d\geq 1$.
  As common, we will identify the set $\{x_i\}$  with 
 the counting measure $\sum _i  \d_{x_i}$.  In particular, if $\xi=\{x_i\}\in \cN$, then $\int d \xi(x) f(x) =\sum_i f(x_i)$ and 
 $\xi(A)= \sharp( \{x_i\}\cap A)$ for $f: \bbR^d\to \bbR$ and $A\subset \bbR^d$.
 On $\cN$ one defines a special metric $d$ (cf.~\cite[App.~A2.6]{DV})
 such that     a sequence  $(\xi_n)$ converges to  $\xi$ in $\cN$ \rosso{if and only if} 
$\lim _{n\to \infty} \int d\xi_n (x) f(x)  = \int d \xi (x)   f (x) 
$ for any bounded continuous function $f: \bbR^d \to \bbR$ vanishing outside a bounded set.
Then the 
  $\s$--algebra  of  Borel sets of $(\cN,d)$ is generated by the sets $\{\xi\in \cN\,:\, \xi (A)= k\}$ with $A$ and  $k$ varying respectively among the Borel sets of $\bbR^d$  and in $\bbN$  (cf. \cite[\rosso{Section~7.1}]{DV}). In what follows, we think of $\cN$ as    measure space endowed with the  $\s$--algebra  of  Borel sets.

We consider a \emph{simple point process} on $\bbR^d$  defined on $(\O,\cF,\cP)$, i.e.  a measurable map
$\O \ni \o \mapsto \hat \o\in \cN$.
We  also consider the \emph{group} $\bbG$ given by $\bbR^d$ or $\bbZ^d$ acting on the Euclidean space $\bbR^d$ by the translations $\t_g:\bbR^d \to \bbR^d $, where $\t_g x= x+g $. 
\begin{Warning}\label{silente89}  To simplify here the presentation, when $\bbG=\bbZ^d$   we assume that $\hat \o \subset \bbZ^d$ for all $\o \in \O$ (in Section \ref{bingo} we will remove this assumption).
\end{Warning} We assume that $\bbG$ acts also on $\O$ and, with a slight abuse  of notation made non ambiguous by the context, we denote by $(\t_g)_{g\in \bbG}$ also the action of $\bbG$ on $\O$. In particular, this action is given by a family of  $\bbG$--parametrized  maps  $\t_g:  \O\to \O$ such that
  $\t_0=\bbI$, $\t_g\circ \t_{g'}= \t_{g+g'}$ for all $g,g'\in \bbG$, $\bbG\times \O\ni (g, \o)\mapsto \t_g \o\rosso{\in\O}$ is measurable ($\bbR^d$, $\bbZ^d$ are endowed with the  Euclidean metric and the discrete topology, respectively). As common, a  subset  $A\subset \O$ is called  \emph{translation invariant} if  $\t_g A =A$ for all $g\in \bbG$. The name comes from the fact that  the action of $\bbG$ on $\O$ describes how the environment changes when  applying translations on the Euclidean space (cf.~Assumption (A4) below).
We will assume that $\cP$ is stationary and ergodic w.r.t. the action of $\bbG$ on $\O$. 
 We recall that   stationarity means  that $\cP (\t_g A)=\rosso{\cP(A)}$ for any $A\in \cF$ and $g\in \bbG$, while ergodicity means that  $\cP(A)\in\{0,1\}$
 for any  translation invariant set $A\in \cF$.

 Due  to our  assumptions stated below, the simple point process 
 has  finite positive  \emph{intensity} $m$, where 
  \begin{equation}\label{mom_palma0}
m:=\begin{cases}
 \bbE\bigl[  \hat \o (  [0,1]^d)\bigr] &\text{ if }\bbG=\bbR^d\,,\\
 \cP( 0 \in \hat \o) & \text{ if } \bbG=\bbZ^d\,.
 \end{cases}
\end{equation}
   As a consequence,  the \emph{Palm distribution} $\cP_0$ associated to the simple point process is well defined (cf.~\cite[Section~2]{Fhom3} and references therein). We recall that $\cP_0$ is the  probability measure on $(\O,\cF)$ with support in 
        \be \label{pandi}
    \O_0:=\{ \o \in \O\,:\, 0 \in \hat \o\}\,, 
   \en
  such that, for any $A\in \cF$,  
  \be\label{passaparola}
\cP_0(A):=
\begin{cases} \frac{1}{m   }\int _\O d\cP(\o) \int_{[0,1]^d} d\hat \o(x)   \mathds{1}_A(\t_{x} \o) & \text{ if } \bbG=\bbR^d\,,\\
\cP(A |\O_0)& \text{ if }
 \bbG=\bbZ^d\,.
 \end{cases}
 \en
 In the rest, we will denote by $\bbE_0[\cdot]$ the expectation w.r.t. $\cP_0$. 
  \smallskip
 
We  fix  a measurable function (describing the  \emph{random conductance field})
\[
\bbR^d\times \bbR^d \times \O \ni (x,y,\o) \mapsto c_{x,y} (\o)\in [0,+\infty)
\]
such that $c_{x,x}(\o)=0 $ for all $x\in \bbR^d$.
The value of $ c_{x,y} (\o)$ will be relevant only for $x\not =y$ in  $\hat \o$.
For later use we define the function $\l_k:\O_0\to [0,+\infty]$ as
    \begin{equation}\label{organic}
  \l_k(\o):=\int _{\bbR^d} d\hat \o (x) c_{0,x}(\o)|x|^k\,,  \end{equation}
 where   $|x|$ denotes the  Euclidean norm of $x\in \bbR^d$.

\smallskip

Recall \rosso{the} temporary assumption in Warning \ref{silente89}

\smallskip
{\bf Assumptions.}
We make the following assumptions:
\begin{itemize}
\item[(A1)] $\cP$ is stationary and ergodic w.r.t. the action of $\bbG$ on $\O$;
 \item[(A2)] the intensity $m$ given in  \eqref{mom_palma0} is finite and positive;
\item[(A3)]  $\cP ( \o \in \O:  \t_g\o\not = \t_{g'} \o   \; \forall g\not =g' \text{ in }\bbG )=1$;
\item[(A4)]    for all $\o\in \O$, $g\in \bbG$ and $x,y\in \widehat{\t_g\o}$, it holds 
\begin{align}
&\widehat{\t_g\o}  =\t_{-g}( \hat \o )  \,,\label{base}\\
& c_{x,y} (\t_g\o)= c_{\t_g x, \t_g y} (\o) \,;\label{montagna}
 \end{align}
\item[(A5)]  for all $\o\in \O$ the weights $c_{x,y}(\o)$ are symmetric, i.e. $c_{x,y}(\o)=c_{y,x}(\o)$  $\forall x,y\in \hat\o$;
\item[(A6)]    for $\cP$--a.a.~$\o$ the  graph with vertex set $\hat \o$ and edges given by $\{x,y\}$ with $x\not =y$ in $\hat \o$ and $c_{x,y}(\o)>0$ is connected;
\item[(A7)]     $\l_0, \l_2 \in L^1(\cP_0)$;
\item[(A8)]   $L^2(\cP_0)$ is separable.
\end{itemize}

 We point out that the above assumptions are  the same ones presented in \cite[Section~2.4]{Fhom3} when the rates $r_{x,y}(\o)$ there  are symmetric (hence coinciding with our $c_{x,y}(\o)$). To simplify, differently from \cite{Fhom3}, we have required some properties to hold for all $\o$, but their $\cP$-a.s. validity would suffice. 
 We now comment the above assumptions (recalling also some remarks from \cite{Fhom3}).
 
   Due to  \cite[Proposition~10.1.IV]{DV},  (A1) and  (A2), for $\cP$--a.a.~$\o$ the set $\hat \o$ is  infinite.

    \rosso{As observed in \cite[Appendix~A]{Fhom3} 
 the event in  (A3) is measurable. This is trivially true if $\bbG=\bbZ^d$. For $\bbG=\bbR^d$ note that the event equals  $\{\o\in \O\,:\, \t_g\o \not =\o \; \forall g\in \bbG\setminus\{0\}\}$ and that  by \eqref{base} its complement is given by the environments $\o$ such that $\t_g\o=\o$ for some $g\in \bbG  \setminus \{0\}$ with $g\in (\hat\o - \hat \o)$ (as detailed in \cite[Appendix~A]{Fhom3} from the last characterization one easily gets that  the complement is measurable, since points of $\hat\o$ can be enumerated in a measurable way). Trivially the event in (A3) is also translation invariant hence  ergodicity  would imply that its probability is $0$ or $1$. In (A3) we require the probability to be $1$. (A3) is satisfied in plenty of models. In periodic models where (A3)  fails,  for  free one can   add  some randomness  enlarging $\O$ to assure (A3)   (see Section  \ref{app_esempio} for details). }

 
 (A4) describes how the Euclidean translations influence the randomness. 
 
 (A5) is natural due to the interpretation of conductance of $c_{x,y}(\o)$ discussed below. 
 
 (A6) is a technical assumption assuring that a measurable function $u$ on $\O_0$ such that, $\cP_0$--a.s.,  $u(\t_x \o)=u(\o) $ for all $x\in \hat \o$ with  $c_{0,x}(\o)>0$ is  constant $\cP_0$--a.s. (cf.~\cite{Fhom3}[Lemma~8.5]).  (A6) can be weakened: in Section \ref{ciuffolotto} we will discuss a relevant example where (A6) does not hold but anyway the application of our Theorem \ref{teo1} allows to derive the scaling limit of the directional conductivity.
 
By  \cite[Theorem~4.13]{Br} (A8) is fulfilled if $(\O_0,\cF_0,\cP_0)$ is a separable  measure space where $\cF_0:=\{A\cap \O_0\,:\, A\in \cF\}$ (i.e.~there is a countable family $\cG\subset  \cF_0$ such that  the $\s$--algebra  $\cF_0$ is generated by $\cG$). For example, if $\O_0$ is a separable metric space and 
$\cF_0= \cB(\O_0)$ (which is valid if $\O$ is a separable metric space and 
$\cF= \cB(\O)$) then (cf. \cite[p.~98]{Br}) $(\O_0,\cF_0,\cP_0)$ is a separable  measure space  and (A8) is valid.
Note that we are not assuming that $\O$ is a compact metric space as in e.g. \cite{ZP}, hence  (A8) becomes relevant to have countable families  of test functions for the 2-scale convergence (we refer to  \cite{Fhom3} for further comments on this issue).


\begin{Definition}\label{fasma98}  We define the  effective homogenized   matrix $D$ as 
the $d\times d$ nonnegative symmetric matrix such that
 \begin{equation}\label{def_D}
 a \cdot Da =\inf _{ f\in L^\infty(\cP_0) } \frac{1}{2}\int d\cP_0(\o)\int d\hat \o (x) c_{0,x}(\o) \left
 (a\cdot x - \nabla f (\o, x) 
\right)^2\,,
 \end{equation}
 where $\nabla f (\o, x) := f(\t_x \o) - f(\o)$.
\end{Definition}
By (A7) the above definition is well posed. Above $a\cdot b$ denotes  the Euclidean scalar product of the vectors $a$ and $b$.

Given $\ell>0$  we
consider the box, stripe and half-stripes\footnote{The term \emph{stripe} is appropriate for $d=2$. We keep the same terminology for all dimensions $d$.}
  \begin{equation}
  \begin{cases}
  \L_\ell:=(-\ell/2,\ell/2)^d   \,, \;\;  & S_\ell:=\bbR\times (-\ell/2,\ell /2)^{d-1}   \\
  S_\ell^-:=\{ x \in S_\ell \,:\, x_1\leq -\ell/2\}\,,\;\;& S_\ell^+:=\{x \in S_\ell \,:\, x_1\geq \ell/2\}\,.
  \end{cases}
  \en 
  
\begin{Warning}\label{stellina59} 
We denote by $e_1,\dots,e_d$ the canonical basis of $\bbR^d$.  In what follows we focus on the direction determined by $e_1$, \rosso{for simplicity of notation}. 
In the general case, when considering the direction determined by a unit vector $e$,
one has just  to refer our results to the regions $O(\L_\ell)$, 
$O(S_\ell)$, $O(S_\ell^-)$ and $O(S_\ell^+)$, where $O$ is a fixed orthogonal linear  map  such that $O(e_1)=e$.
 \end{Warning}

We define $\O_1$ as the set of $\o\in \O$ satisfying 
the  connectivity  property in (A6)   and the bounds (cf.~\eqref{organic})
\begin{equation}\label{sommetta}
\l_0(\t_x \o)=\sum_{y\in \hat \o } c_{x,y}(\o) <+\infty \;\;\; \forall x \in \hat \o \,.
  \end{equation}
  Note that  $\O_1$ is a translation invariant measurable set with  $\cP(\O_1)=1$  (use (A7) and  Lemma \ref{matteo} below).

\begin{Definition}[Resistor network ${\rm (RN)}^\o_\ell$] \label{def_RN} Given $\o\in \O_1$ we consider the  $\ell$--parametrized resistor network ${\rm (RN)}^\o_\ell$  on $S_\ell$  with  node set  $\hat \o \cap S_\ell$. 
To each unordered pair  of nodes $\{x,y\}$, such that  $\{x,y\}\cap  \L_\ell \not=\emptyset$ and $c_{x,y}(\o)>0$,
we associate an electrical filament of conductance  $c_{x,y}(\o)$ (see Figure \ref{messicano1}-(left)).
\end{Definition}
 We can think  of ${\rm (RN)}^\o_\ell$ as a weighted non-oriented graph  
with vertex set $ \hat \o \cap S_\ell$, edge set 
\be\label{latticini}
\bbB^\o_\ell:=\bigl\{\{x,y\}\subset( \hat \o \cap S_\ell) \,:\, \{x,y\}\cap  \L_\ell \not=\emptyset\,,\;  c_{x,y}(\o)>0 \bigr\}
\en  and weight  of the edge $\{x,y\}$ given by the conductance $c_{x,y}(\o)$.

\begin{figure}
\includegraphics[scale=0.34]{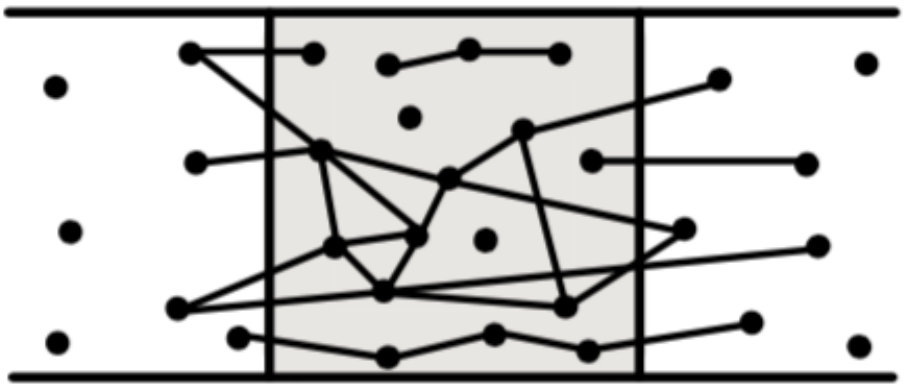}\;\;\qquad\includegraphics[scale=0.34]{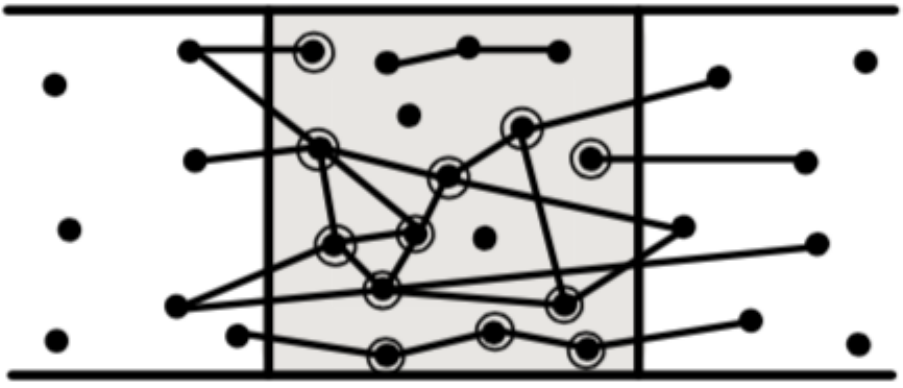}
\caption{(Left)  The resistor network ${\rm (RN)}^\o_\ell$, conductances  are omitted. The box and the stripe correspond to  $\L_\ell$ and   $S_\ell$, respectively. (Right) Functions in $H^{1,\e}_{0,\o}$ take any value on the nodes surrounded by a circle, and have value zero on all other nodes (the box corresponds here to $\L$  and has side length $1$, the underlying graph is $\cG^\e_\o$). } \label{messicano1}
\end{figure}

\begin{Definition}[Electric potential] \label{def_EP} Given $\o\in \O_1$  we denote by   $V_\ell^\o$  the \emph{electric potential} of the resistor network {\rm (RN)$^\o_\ell$} with values $0$ and $1$ on $S_\ell^- $ and $S_\ell^+$, respectively,  taken by convention equal to zero on the connected components  of ${\rm (RN)}^\o_\ell$ included in $\L_\ell$. In particular, $V_\ell^\o$ is the unique function $V_\ell^\o: \hat \o \cap S_\ell\to \bbR$ such that 
\be\label{thenero}
{\sum} _{y\in \hat \o \cap S_\ell} c_{x,y}(\o) \left(  V^\o _\ell (y)-  V^\o _\ell (x)\right) =0 \qquad \forall 
x\in \hat \o \cap \L_\ell\,,
\en
and satisfying
\be\label{thebianco}
\begin{cases}
 V^\o_\ell (x) = 0  & \text{ if } x \in \hat \o \cap \L_\ell, \text{ $x$ is  not connected to  $\hat \o \setminus \L_\ell$ in {\rm (RN)$^\o_\ell$}}\,,\\
 V^\o_\ell (x) = 0  & \text{ if } x \in \hat \o \cap S^-_\ell\,,\\
 V^\o_\ell (x) = 1  & \text{ if } x \in \hat \o \cap S_\ell^+\,.
\end{cases}
\en
\end{Definition}
Note that \eqref{thenero} corresponds to Kirchhoff's law.
As discussed in Section \ref{sec_hilbert}, the above  electric  potential exists and is unique   and has values in $[0,1]$. \rosso{Since the boundary conditions \eqref{thebianco} fix the value of $V_\ell^\o(x)$ apart from a finite set of points $x$, the function $V_\ell^\o$ has to be bounded and therefore, due to \eqref{sommetta}, equation \eqref{thenero} is well posed.}
We recall that, given $(x,y)$ with $\{x,y\}\in \bbB^\o_\ell$ (cf. \eqref{latticini}),  
\be\label{ahahah} i_{x,y}(\o):=c_{x,y}(\o) \bigl( V^\o_\ell(y)- V^\o _\ell (x) \bigr)
\en is the electric current flowing from $x$ to $y$ under the electric potential $V^\o_\ell$, due to  Ohm's law. For simplicity,   $\ell$ is \rosso{omitted} in the notation $i_{x,y}(\o)$.

\begin{Remark}\rosso{Note our convention that current flows uphill w.r.t. $V^\o_\ell$ (as e.g. in \cite{G_book}). The physical electrical potential would be $1-V_\ell^\o$, hence equal to $1$ on $S_\ell^-$ and to $0$ on $S_\ell^+$.}
\end{Remark}

\begin{Definition}[Directional effective conductivity]  \label{def_EC}Given $\o\in \O_1$ we call $\s_\ell(\o)$ the \emph{effective conductivity} of the resistor network  $({\rm RN})^\o_\ell$ along the first direction under the electric potential $V^\o_\ell$.  More precisely,  $\s_\ell (\o)$ is given by
 \[ \label{ide1}
\begin{split}
 \s_\ell(\o)&: =  \sum _{x\in \hat \o \cap S_\ell^-}
   \; \sum_{y \in  \hat \o \cap \L_\ell} i_{x,y}(\o) =  \sum _{x\in \hat \o \cap S_\ell^-}\sum_{y \in  \hat \o \cap \L_\ell} c_{x,y} (\o) \bigl(V_\ell^\o (y)-V_\ell^\o(x)\bigr)\,.
\end{split}
 \]
\end{Definition} 
It is simple to check that, for any $\g \in [-\ell/2, \ell/2)$, 
 $\s_\ell(\o)$ equals the current flowing through the hyperplane $\{x\in \bbR^d\,:\, x_1=\g\}$:
\begin{equation}\label{eq_ailo0}
\s_\ell(\o)=\sum _{\substack{x\in \hat \o \cap S_\ell:\\
 x_1\leq \g }} \; \sum_{\substack{ y \in  \hat \o \cap S_\ell:\\
 \{x,y\} \in \bbB^\o_\ell, \, y_1>\g
 }} i_{x,y}(\o) \,.
 \end{equation}
 $\s_\ell(\o)$   also equals  the total dissipated energy:
\begin{equation}\label{eq_ailo}
\s_\ell(\o)= \sum_{ \{x,y\} \in \bbB^\o_\ell } c_{x,y}(\o) \bigl( V^\o _\ell (x) - V^\o _\ell (y) \bigr)^2\,.
\end{equation}
Indeed, by collapsing all nodes in $\hat \o \cap S_\ell^-$ into a single node and similarly for $\hat \o \cap S_\ell^+$, one reduces to the same setting of  \cite[Section~1.3]{DS} where \eqref{eq_ailo} is proved.

\smallskip
\rosso{To state our main results we give a quick and  rough definition of the  gradient  $\nabla_*$ and the function $\psi$ when  $e_1 \in {\rm Ker}(D)^\perp={\rm Ran}(D)$, 
 referring to Section \ref{eff_equation} for a precise treatment. 
$\nabla_* $  denotes the  weak gradient along ${\rm Ker}(D)^\perp$. If $\varphi$ is a  regular function, then  $\nabla_* \varphi$ coincides with the orthogonal projection of $\nabla \varphi $ on ${\rm Ker}(D)^\perp$.   $\psi$ denotes the unique weak solution on $\L:=(-1/2,1/2)^d$  of the equation $ \nabla_* \cdot ( D \nabla_* \psi) =0$ with the following mixed Dirichlet-Neumann boundary conditions: $\psi$ equals  zero on $\{x\in \bar \L: x_1=-1/2\}$ and one on $\{x\in \bar  \L: x_1=1/2\}$, while $ D \nabla_* \psi (x) \cdot  \mathbf{n}(x) =0 $  on the other faces of $\L$, where $n(x)$ is the outward normal vector. 
}


We can now state our first main result concerning the infinite volume asymptotics of $\s_\ell (\o)$ (the proof is given in Sections \ref{renato_zero} and \ref{limitone}):
\begin{Theorem}\label{teo1} 
 \rosso{There  exists a translation invariant measurable set $\O_{\rm typ}\subset \O_1$  with $\cP(\O_{\rm typ})=1$ such that for all $\o\in \O_{\rm typ}$ the following holds:
\begin{itemize}
\item[(i)] if $e_1 \in {\rm Ker}(D)$, then  $\lim _{\ell \to +\infty} \ell^{2-d} \s_\ell(\o)= 0$;
\item[(ii)] if $e_1 \in {\rm Ker}(D)^\perp$, then  $\lim _{\ell \to +\infty} \ell^{2-d} \s_\ell(\o)= m \int _\L D \nabla_*\psi(x)\cdot \nabla_*\psi(x)  dx$.
\end{itemize}}
\end{Theorem}
The notation $\O_{\rm typ}$ refers to the fact that elements of $\O_{\rm typ}$ are typical environments, as $\cP(\O_{\rm typ})=1$. 

\begin{Corollary}\label{cor_airone}
\rosso{Suppose that $e_1$ is  an eigenvector of $D$. Then  for all $\o\in \O_{\rm typ}$ it holds $\lim _{\ell \to +\infty} \ell^{2-d} \s_\ell(\o)= m D_{1,1}$.}
\end{Corollary}
\begin{proof}
\rosso{If $e_1$  has zero eigenvalue, then  $e_1 \in {\rm Ker}(D)$,  $D_{1,1}=0$ and therefore the claim is equivalent to Theorem \ref{teo1}-(i). If $e_1$ has positive eigenvalue, then $e_1 \in {\rm Ker}(D)^\perp$ and $\psi$ is given by the function $\L \in x\mapsto x_1+1/2\in[0,1]$ as stated in Corollary \ref{melone28}. As  $\nabla_* \psi=e_1$, the claim follows from 
Theorem \ref{teo1}-(ii).}
\end{proof}

To clarify the link with homogenization and state our further results, it is convenient 
to rescale space in order to deal with fixed stripe and box. More precisely, we set $\e:= 1/\ell$.
Then $\e>0$ is our scaling parameter. We set \rosso{(recalling the definition of $\L$)}
\be\label{serpente}
\begin{cases}
 \L:=(-1/2,1/2)^d\,,
 & S:=\bbR\times (-1/2,1/2)^{d-1}\,,    \\
S^-:= \{x\in S: x_1 \leq -1/2\} \,, &S^+:=\{x\in S: x_1 \geq 1/2\}\,.
\end{cases}
\en
Note that $\L_\ell=\ell \L$, $S_\ell=\ell S$, $S_\ell^\pm=\ell S^\pm$.
Here and below, $\o\in \O_1$.
We write  $V _\e: \e \hat \o \cap S \to [0,1] $ for the function given by $V_\e (\e x):= V_\ell ^\o (x)$ (note that the dependence on $\o$ in $V_\e$ is understood, as for other objects below). 

We introduce  the atomic measures 
\be\label{atomiche}
\mu^\e_{\o,\L}:= \e^d \sum _{x \in \e\hat \o  \cap \L} \d_x\,,\qquad 
 \nu ^\e_{\o,\L } : = \e^d
  \sum_{(x,y)\in \cE_\e(\o)}   c_{x/\e,y/\e }(\o) \d_{( x ,(y-x)/\e)}\,,
\en
where  
\be\label{caprino67}
 \cE_\e(\o) :=\{ (x,y)\,:\, x, y \in \e \hat \o \cap S, \, \rosso{c_{x/\e,y/\e}(\o)}>0 \text{ and }\{x,y\}\cap  \L\not=\emptyset\}\,.
 \en 
Note that $\mu^\e_{\o,\L}$ and $ \nu ^\e_{\o,\L } $ have finite total mass (for the latter use that $\o\in \O_1$).

Given a function $f: \e \hat \o \cap S \to \bbR$, we define the \emph{amorphous gradient} $\nabla_\e f$ on pairs $(x,z)$ with  $x\in \e \hat \o\cap S$ 
and $x+ \e z\in \e \hat\o\cap S$ as 
\begin{equation}\label{ricola}
 \nabla_\e f(x,z)= \frac{ f(x+\e z)- f(x)  }{\e}\,.
 \end{equation}
Moreover,  we define the operator
\be\label{degregori} \bbL^\e_\o f( x):= \e^{-2} \sum _{y  \in \e\hat \o \cap S }  c_{x/\e,y/\e}\left[ f( y) - f( x)\right]\,, \qquad x\in \e\hat \o  \cap \L\,,
\en
whenever  the series in the r.h.s. is  absolutely convergent. 
Let  $f: \e \hat \o \cap S\to \bbR$  be a bounded function. 
Then 
  $\bbL^\e_\o f(x) $ is well defined for all $x\in 
  \e \hat \o\cap \L$  as $\o\in \O_1$. 
   As the amorphous gradient $\nabla_\e f$ is bounded too, we have that $\nabla_\e f\in 
 L^2(\nu^\e_{\o,\L})$. Moreover,  if in addition $f$ is zero outside $\L$,  it holds (cf. Lemma \ref{spiaggia})
\be\label{dir_form}
\la f, -\bbL^\e_\o f \ra _{L^2(\mu ^\e_{\o,\L})}
=
\frac{1}{2}\la \nabla _\e f, \nabla_\e f\ra _{L^2(\nu^\e_{\o,\L})}<+\infty \,.
\en

\begin{Definition}[Graph $\cG_\o^\e$ and sets $\cC_{\o,\L}^\e$,  $\cC_{\o}^\e$] \label{pinco52}
Given $\e>0$ and $\o\in \O_1$, we consider the non-oriented  graph $\cG_\o^\e$ with vertex set $\e \hat \o \cap S$ and edges given by the unordered pairs  $\{x,y\}$ such that $\rosso{c_{x/\e,y/\e}(\o)}>0$  and $\{x,y\}$ intersects $\L$.
We write  $\cC_{\o,\L}^\e$   and $\cC_{\o}^\e$
  for the union of the connected components in $\cG_\o^\e$   included in $\L$ and, respectively, intersecting $S^-\cup S^+$.
\end{Definition}

 $\cG_\o^\e$
equals the $\e$-rescaling of the graph obtained by disregarding the weights in the weighted graph (RN)$^\o_\ell$, where $\ell=1/\e$.
Note that $\cG_\o^\e$  has no edge between $\e \hat \o \cap S^-$ and $\e \hat \o \cap S^+$ and that $\e \hat \o \cap S=
\cC_{\o,\L}^\e \rosso{\sqcup}\,\cC_{\o}^\e$.  The edges of $\cG_\o^\e$ coincide with the edges obtained from $\cE_\e(\o)$ when disregarding the orientation.  \rosso{$\cC_{\o,\L}^\e$ is the family of points $\e x$ with $x\in \hat \o \cap \L_\ell$  not connected to $\hat \o \setminus \L_\ell$ in {\rm (RN)$^\o_\ell$}, while $\cC^\e_\o \setminus \L$
is the family of points $\e x$ with  $x\in \hat \o \cap S_\ell^-$ or $x\in \hat \o \cap S_\ell^+$ (cf.~\eqref{thebianco}).}

\begin{Definition}[Functional spaces  $H^{1,\e}_{0,\o}$, $K^{\e}_{\o}$]  \label{def_hilbert} Given $\o \in \O_1$
we define  the Hilbert space 
\[
H^{1,\e}_{0,\o} :=\Big \{ f :\e \hat \o \cap  S \to \bbR \text{ such that  } 
f(x)=0\;\;   \forall x \in \cC^\e_{\o,\L} \cup (\cC^\e_\o \setminus \L) 
\Big\}
\]
 endowed 
with the squared  norm
$\|f\|^2_{H^{1,\e}_{0,\o}}:= \|f \| ^2_{ L^2( \mu ^\e_{\o,\L})}+\| \nabla_\e  f \|^2_{L^2(  \nu ^\e_{\o,\L})}
$. In addition, we define  $K^{\e}_{\o}$ as the set of functions $ f:\e \hat \o \cap S \to \bbR$ such that 
 $f(x)=0$  for all $ x  \in   (\e \hat \o \cap S^-)\cup \cC_{\o,\L}^\e$ and $f(x)=1$ 
 for all $ x  \in  \e \hat \o \cap S^+$.  
\end{Definition}
We refer to  Figure \ref{messicano1}-(right) for a graphical comment on the graph $\cG^\e_\o$ and  functions in $H^{1,\e}_{0,\o}$.
Given    $\o \in \O_1$, in  Section \ref{sec_hilbert} we will derive that,  due to \eqref{thenero} and \eqref{thebianco},   $V_\e$ is the  unique function in $ K^\e _\o$ such that $\bbL^\e_\o V_\e (x)=0$ for all $x \in \e\hat \o \cap \L$ (cf. Lemma \ref{benposto}). We point out that, by \eqref{eq_ailo},  the rescaled conductivity $\ell^{2-d} \s_\ell(\o)$  equals the flow energy associated to $V_\e$:
\be\label{chiave}
\ell^{2-d} \s_\ell(\o)= 
 \frac{1}{2} \la   \nabla _\e V_\e, \nabla_\e V_\e \ra _{L^2(\nu ^\e_{\o,\L})} \,.
\en
\rosso{The limits in Theorem \ref{teo1} can therefore  be restated as}
 \begin{equation}\label{mortisia} 
 \rosso{
 \lim_{\e \da 0} \frac{1}{2} \la   \nabla _\e V_\e, \nabla_\e V_\e \ra _{L^2(\nu ^\e_{\o,\L})}
 =
 \begin{cases}
0 & \text{ if } e_1 \in {\rm Ker}(D) \\
 m \int _\L D \nabla_*\psi(x)\cdot \nabla_*\psi(x)  dx &\text{ if } e_1 \in {\rm Ker}(D)^\perp 
 \end{cases}\,,
  \;\forall \o \in \O_{\rm typ}\,.}
 \end{equation}

%

To prove Theorem \ref{teo1} we will distinguish the cases  \rosso{$e_1 \in {\rm Ker}(D)$ and $e_1 \in {\rm Ker}(D)^\perp$}. The proof for \rosso{$e_1 \in {\rm Ker}(D)$}  (which is simpler) is given in Section \ref{renato_zero}, while  the proof for \rosso{$e_1 \in {\rm Ker}(D)^\perp$} will take the rest of our investigation and will be concluded in Section \ref{limitone}.
In the case  \rosso{$e_1 \in {\rm Ker}(D)^\perp$}
 we can say more on the behavior of $V_\e$. \rosso{Indeed, as an intermediate step for the proof of Theorem \ref{teo1}-(ii) we will prove the following:}
 
\begin{Proposition}\label{teo2} Suppose  that \rosso{$e_1 \in {\rm Ker}(D)^\perp$}.  Then,  for all $\rosso{\tilde\o} \in \O_{\rm typ}$, \rosso{it holds}
\begin{itemize}
\item[(i)] $V_\e\in L^2(\mu^\e_{\rosso{\tilde\o},\L}) $  converges   weakly and 2-scale converges weakly  to $\rosso{\psi} \in L^2(\L, dx)$;
\item[(ii)] \rosso{$\nabla_\e V_\e\in L^2(\nu^\e_{\rosso{\tilde\o},\L})$ 2-scale converges to $\nabla_* \psi + v_1 (x,\o,z)\in   L^2(\L \times \O,  m dx \times \nu)$, where   $v_1\in L^2\bigl( \L,\rosso{dx;}\, L^2_{\rm pot} (\nu)\bigr )$  and, given $x$,   the \emph{correcting form} $(\o,z)\mapsto v_1(x,\o, z)$ is   the opposite (in sign) of  the orthogonal projection of the form $(\o,z)\mapsto \nabla_* \psi(x)\cdot z$ along the subspace $ L^2_{\rm pot} (\nu)$ of potential forms.}
\end{itemize}
\end{Proposition}
\rosso{The definitions of the above   concepts (forms, potential forms, weak and 2-scale convergence)  can be found in Sections \ref{sec_review} and \ref{anatre12}}.
\rosso{Due to the above proposition, when $e_1\in {\rm Ker}(D)^\perp$,  the equation $\nabla_* \cdot ( D \nabla_*\psi) =0$ with boundary conditions $\psi\equiv 0$  on $\{x\in \bar \L: x_1=-1/2\}$, $\psi\equiv 0$ on $\{x\in \bar  \L: x_1=1/2\}$ and $ D \nabla_* \psi (x) \cdot  \mathbf{n}(x) =0 $  on the other faces of $\L$,  represents the effective macroscopic equation of the electric potential $V_\e$  in the limit $\e\da 0$.}

\begin{Remark}\label{luminaria} From the description of  $\O_{\rm typ}$ in Sections \ref{renato_zero} and \ref{sec_tipetto} and from the proof, it is simple to check that the same set $\O_{\rm typ}$ leads to the conclusion of Theorem \ref{teo1} (and, similarly, of Proposition \ref{teo2} and Theorem  \ref{teo3}  below)  for any \rosso{direction  (see Warning \ref{stellina59}). To facilitate  the verification of this  issue we have  given further comments in  Remark \ref{baguette} when dealing with  directions in ${\rm Ker}(D)$ and we have given a detailed description of $\O_{\rm typ}$ when dealing with  directions  in ${\rm Ker}(D)^\perp$.}
 \end{Remark}

\subsection{Extended setting}\label{bingo} We detail here the extended setting where our Theorem \ref{teo1} and Proposition \ref{teo2} still hold, with suitable modifications.

First we point out that the action of the group $\bbG$ on $\bbR^d$ can be a generic proper action, i.e. given by translations  $\t_g:\bbR^d\to\bbR^d$ of the form $\t_g x:= x+ g_1 v_1+\cdots+ g_d v_d$ for a fixed basis $v_1,v_2,\dots, v_d$.  Indeed, one can always reduce to the previous  case $\t_g x= x+g$ at cost of a linear isomorphism of the Euclidean space. \rosso{Hence our results remain valid for a generic proper action} (we refer to \cite[Section 2]{Fhom3} for the form of $D$ when applying a generic proper action). 
\rosso{To illustrate with an example when such an action becomes natural, suppose to deal  with  resistor networks corresponding to random weighted graphs in the planar hexagonal lattice  (see Figure~\ref{esa}-(left), where conductances have been omitted). In this case it is natural to consider stationarity w.r.t. the group $\bbG=\bbZ^2$ acting on $\bbR^2$ by translations of the form 
$\t_g x:= x+ g_1 v_1+ g_2 v_2$ with $v_1$ and $v_2$ as in Figure~\ref{esa}-(right). The same considerations hold in general for graphs in crystal lattices (see  \cite[Section~5.2]{F_SEP} and \cite[Section~5.6]{Fhom3}).}

\begin{figure}[!ht]
   \begin{center}
                \includegraphics[width=8cm]{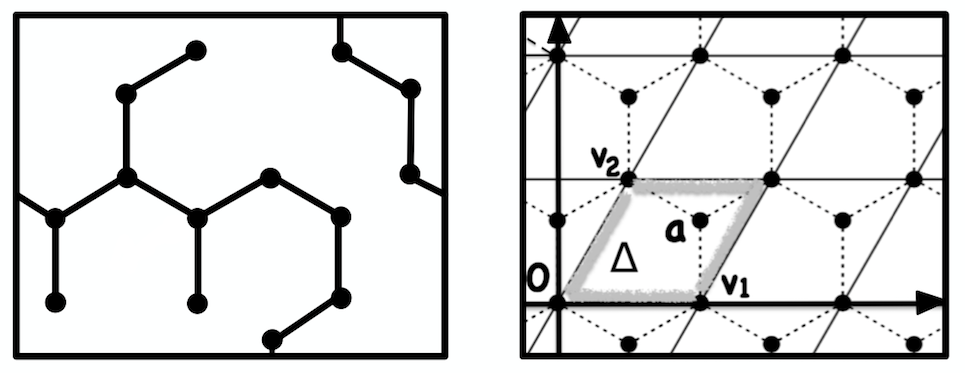}
           \end{center}
           \caption{\rosso{Graphs in the planar hexagonal lattice. }}\label{esa}
  \end{figure}

Let  us keep in what follows  $\t_g x:= x+g$ for $x\in \bbR^d$ and $g\in \bbG$.
Before (see Warning \ref{silente89}), to simplify the presentation,  for  $\bbG=\bbZ^d$ we have restricted to the case $\hat \o \subset \bbZ^d$ for all $\o\in \O$. This is a real restriction, since it does not allow to cover interesting examples given  e.g. by  crystal lattices. \rosso{Indeed, consider the example of Figure~\ref{esa} discussed above. As commented above, we can reduce to the trivial action   $\t_g x:= x+g$ by applying to our pictures  an orthogonal transformation mapping $v_1$ and $v_2$ to $e_1=(1,0)$ and $e_2=(0,1)$, respectively. Then the fundamental cell $\D=\{t_1v_1+t_2 v_2\,: t_1,t_2\in [0,1)\}$ in Figure~\ref{esa}-(right) is mapped to $[0,1)^2$, but not all vertexes  of the hexagonal lattice are mapped into $\bbG=\bbZ^2$ (consider e.g. the point $a$ in Figure~\ref{esa}-(right)).}

For the general case we have the following result:
\begin{Theorem}\label{teo3} When $\bbG=\bbZ^d$ and not necessarily $\hat \o \subset \bbZ^d$, the content of Theorem \ref{teo1} and Proposition \ref{teo2} remain valid under Assumptions (A1),..,(A8)  with $m$, $\cP_0$, $\l_k$ and $D$ defined respectively as  in  \cite[Eq.~(10)]{Fhom3},\cite[Eq.~(11)]{Fhom3},  \cite[Eq.~(15)]{Fhom3} and   \cite[Definition~3.6]{Fhom3}.
\end{Theorem}
For the reader's convenience, we point out that in the above formulas from  \cite{Fhom3}  one has to replace $\theta_g$ by  $\t_g$,   $r_{x,y}(\o)$ by $c_{x,y}(\o)$, $\mu_\o$ by $\hat\o$, $\D$ by $[0,1)^d$.

 \subsection{Reduction to the case $\bbG=\bbR^d$} \label{sec_riduzione}
 \rosso{Once proven  Theorem \ref{teo1} and Proposition \ref{teo2} for $\bbG=\bbR^d$, by a standard procedure in homogenization theory  (cf.~\cite[Section~6]{Fhom3} for the present geometrical context)  one gets Theorem \ref{teo3}  (and therefore also Theorem \ref{teo1}  and Proposition \ref{teo2} for $\bbG=\bbZ^d$).
  For completeness, we give some more details in Appendix \ref{app_porto}.}


 \begin{Warning} \rosso{Due to the above observation,  starting from Section \ref{GS},   we  consider only the case  $\bbG=\bbR^d$ and focus  on the proof Theorem \ref{teo1} and Proposition \ref{teo2}. }  \end{Warning}

\section{Examples}\label{sec_examples} Our results can be applied to plenty of random resistor networks. Our assumptions (A1),...,(A8) (with the modifications indicated in Theorem \ref{teo3} for the general case with $\bbG=\bbZ^d$) equal the ones in \cite{Fhom3} when $n_x(\o)\equiv 1$ there (corresponding to the fact that  $r_{x,y}(\o)$ in \cite{Fhom3} is symmetric and equals our $c_{x,y}(\o)$). Hence, to all the examples discussed in \cite[Section 5]{Fhom3} with symmetric rates  one can apply   the present Theorems \ref{teo1}, \ref{teo3} \rosso{and} Proposition \ref{teo2}.

\rosso{The simple point process $\hat\o$ and the conductance field  can be encoded in the weighted undirected graph $\cG(\o)$ with vertex set $\hat \o$,  edges given by the pairs $\{x,y\}$ with $x\not = y$ in $\hat \o$ and $c_{x,y}(\o)>0$ and weight of the edge $\{x,y\}$ given by $c_{x,y}(\o)$. Hence, below we will equivalently refer to the resistor network built on   $\cG(\o)$ and the resistor network built from $\hat \o$ and the conductance field. }

\begin{figure}[!ht]
\includegraphics[width=0.28\textwidth]{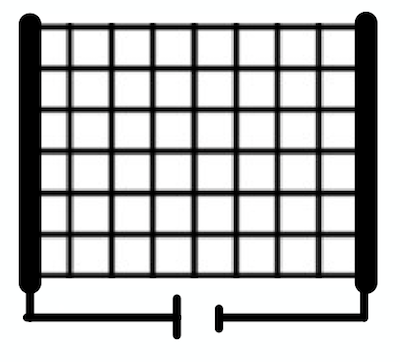} 
      \includegraphics[width=0.5\textwidth]{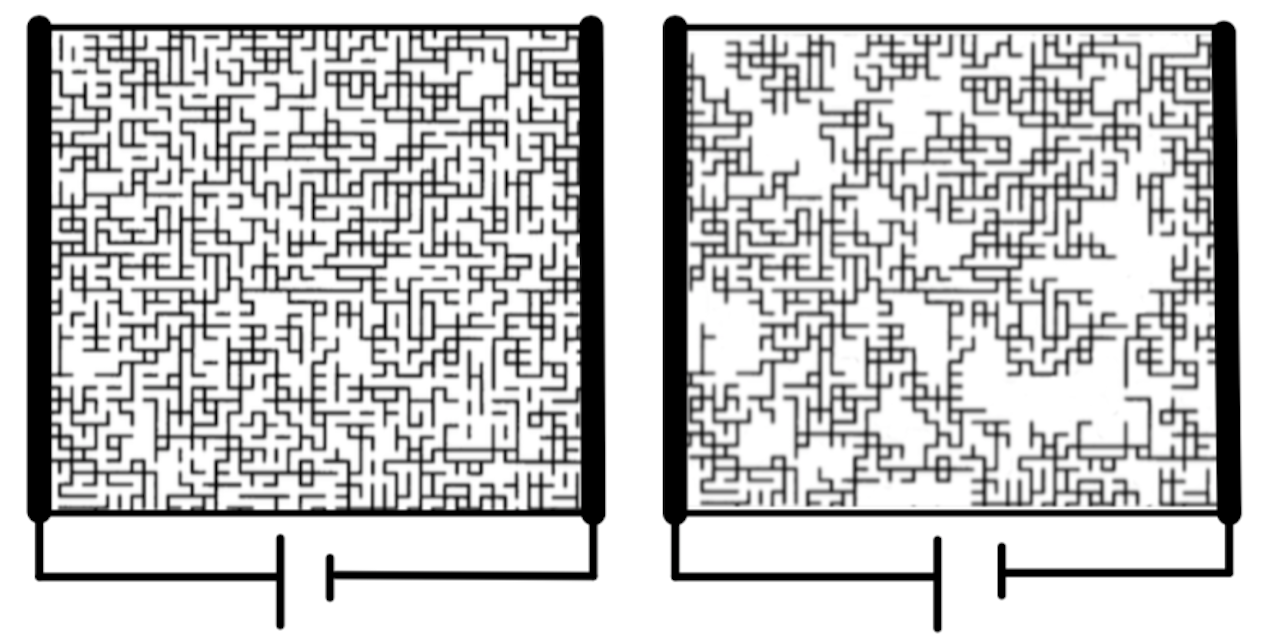} 
      \caption{\rosso{Resistor networks built on $\bbZ^d$, the supercritical bond percolation on $\bbZ^d$ and the supercritical bond percolation cluster on $\bbZ^d$  (from left to right), conductances are omitted.}}\label{fig_reticoli} 
         \end{figure}

\begin{figure}[!ht]
 \includegraphics[width=0.9\textwidth]{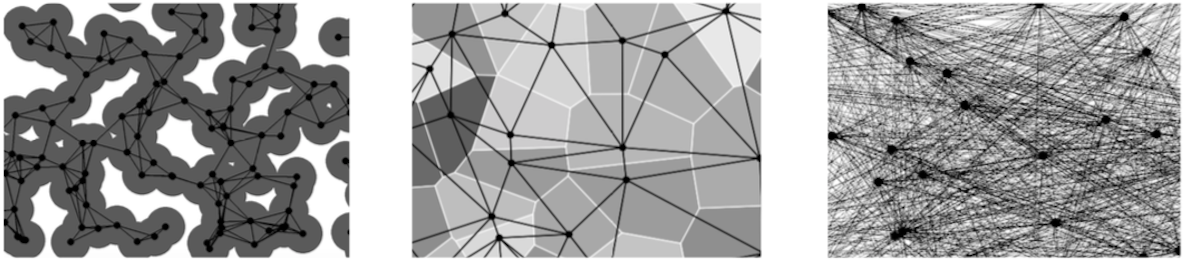} 
 \caption{\rosso{Some random graphs $\cG(\o)$ (conductances are omitted) on which our resistor networks can be built: supercritical Boolean cluster, Delaunay triangulation, complete graphs (from left to right). For the complete graph the drawing is with a finite set $\hat \o$ for practical reasons.}}\label{fig_infinito}
\end{figure}

\begin{figure}[!ht]
       \includegraphics[width=0.45\textwidth]{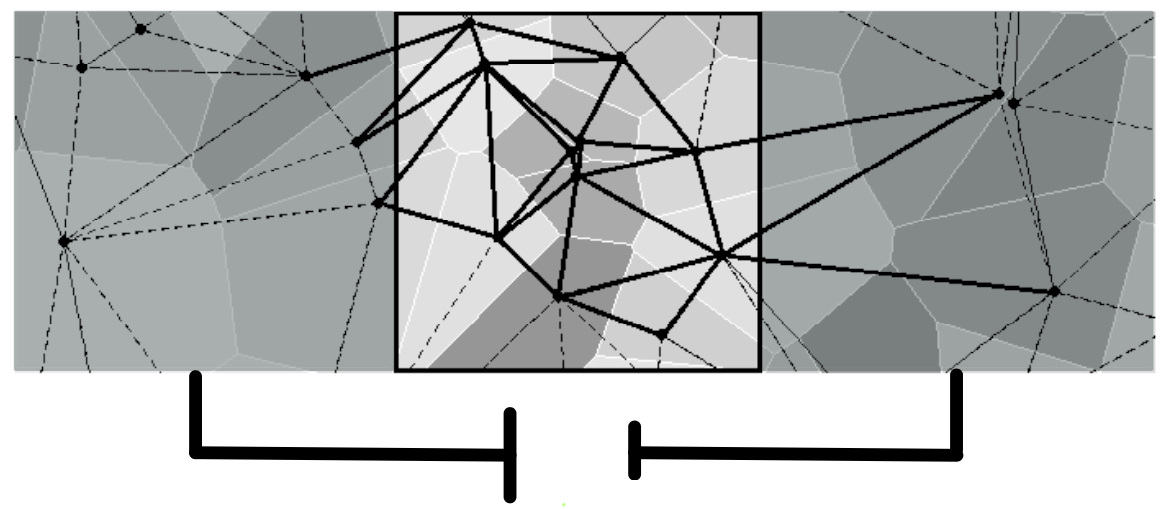} 
       \includegraphics[width=0.44\textwidth]{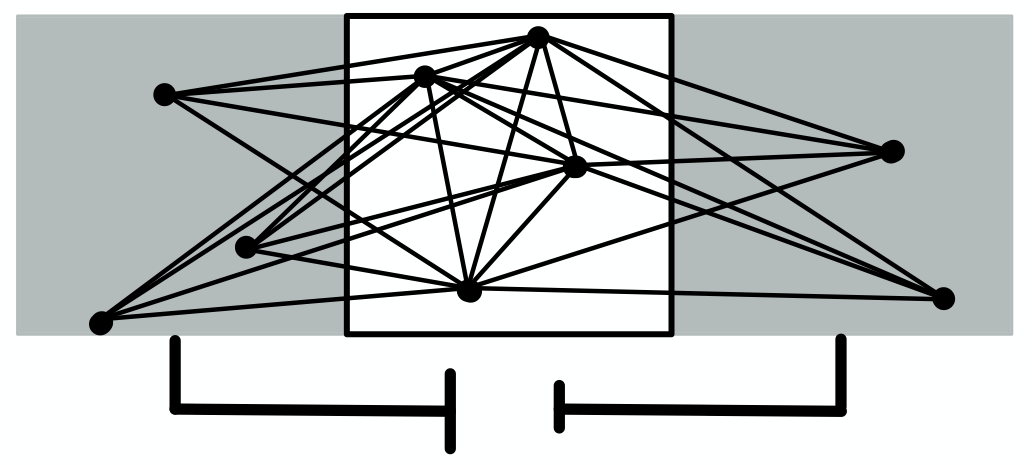} 
       \caption{\rosso{Resistor networks built from the  Delaunay triangulation (left) and a  complete graph  (right). For the complete graph the drawing is with a finite set $\hat \o$ for pratical reasons. }}\label{fig_sfinito}
         \end{figure}

\subsection{Nearest-neighbor random conductance model on $\bbZ^d$}\label{primo}  We take $\O:= (0,+\infty)^{\bbE_d}$, where $\bbE_d$ is the set of non-oriented edges of the standard lattice $\bbZ^d$. 
We take $\hat \o := \bbZ^d$ for all $\o\in \O$ and set $c_{x,y}(\o):=\o_{\{x,y\}}$ if $\{x,y\}\in \bbE_d$ and $c_{x,y}(\o)=0$ otherwise. The group $\bbG=\bbZ^d$ acts in a natural way on $\bbR^d$   by translations and on $\O$  by the maps $\t_g \o=( \o_{x-g,y-g})_{\{x,y\}\in \bbE_d}$ if $\o=(\o_{\{x,y\}})_{\{x,y\}\in \bbE_d}$. \rosso{The resistor network is equivalent to the one in  Figure~\ref{fig_reticoli}-(left) (trivially, with bounded filaments one can avoid the half-stripes $S^\pm_\ell$).}
We recall that Assumption (A3) is \rosso{rather} fictitious \rosso{(see also Section \ref{app_esempio} below)}.  \rosso{All} the other Assumptions (A1), (A2) and (A4),..,(A8)  are satisfied whenever $\cP$ is stationary and ergodic w.r.t.  $\bbZ^d$--translations and $\bbE[ c_{0,e}]<+\infty$ for all $e$ in the canonical basis of $\bbZ^d$.
 \rosso{A  previous  result for the directional conductivity in the random conductance model on $\bbZ^d$ was  provided by Kozlov in \cite[Section~3]{Koz} for i.i.d.~conductances with value in a strictly positive interval $(\d_0,1-\d_0)$, $\d_0>0$. 
} 

%
%

 \subsection{Resistor network on  infinite percolation clusters \rosso{on $\bbZ^d$}}\label{ciuffolotto}
We take  $d\geq 2 $ and  $\O=\{0,1\}^{\bbE_d}$. $\bbE_d$ and   the action of  $\bbG:=\bbZ^d$ are  as in the previous example.  Let $\cP$ be a probability measure on $\O$ stationary,  ergodic and fulfilling (A3). We assume that for $\cP$--a.a. $\o$ there exists a unique infinite connected component $\cC(\o)\subset \bbZ^d$ in the graph given by the edges $\{x,y\}$ in $\bbE_d$  with  $\o _{x,y}=1$. 
   Given $\o \in \O$ we  set $\hat\o:=\cC(\o)$ when $\cC(\o)$
 exists and $\hat \o= \emptyset$ otherwise.
 We  set $c_{x,y} (\o):=1$ for all $x,y\in \hat \o$ such that $\{x,y\}\in \bbE_d$ and $\o_{\{x,y\}}=1$, and   $c_{x,y} (\o):=0$ otherwise. Then all assumptions (A1),...,(A8) are satisfied.
 
 To simplify the notation, we suppose that $\cP$ is invariant by coordinate permutations. Then 
  by symmetry $D=c\mathbb{I}$ and   Theorem \ref{teo1} 
 implies that, along any direction of $\bbR^d$, $\cP$--a.s. $\lim _{\ell \to +\infty} \ell^{2-d} \s_\ell(\o)= c\, m$, where 
  $m$ is the probability that the origin belongs to the infinite cluster, i.e.~$m=\cP(0\in \hat \o)$.  
  Since  $\cP_0:=\cP(\cdot | 0 \in \hat\o)$, we can rewrite the variational characterization of  limiting value $c\, m$ as 
     \begin{equation}
 c\, m =\inf _{ f  }
  \frac{1}{2}\int d\cP (\o)\sum_{\substack{ e: |e|=1\,,\\ \o_{\{0,e\}}=1}}  \mathds{1}_{0\in \hat \o} \left( e\cdot e_1 - \nabla f (\o, e)  \right)^2\,,
 \end{equation}
 where $f$ varies among all bounded Borel  functions on $\O_0=\{0\in \hat \o\}$ and $\nabla f (\o, e) := f(\t_e \o) - f(\o)$. This result solves the open problem stated in  \cite[Problem~1.18]{Bi}, even without restricting to Bernoulli bond percolation. 

From now on, here and in the next paragraph, we focus on the case that $\cP$ is 
 a Bernoulli bond percolation on $\bbZ^d$   with supercritical parameter $p$ \rosso{as in Figure~\ref{fig_reticoli}-(right)} (which fulfills all the above properties).
It  is known that  $c>0$ (cf.~\cite{BB,demasi,MP}). The proof originally provided in \cite[Section~4.2]{demasi}  can now be simplified due to our result and due to \cite{GM}. Indeed, calling $N_\ell(\o)$ the maximal number of vertex-disjoint left-right crossings of $\L_\ell$ and calling $s_j$ the number of bonds in the $j$--th path, as pointed out in \cite[Eq.~(3.3),(3.4)]{CC}, it holds
\be
\s_\ell(\o) \geq \sum_{j=1}^{N_\ell (\o)} \frac{1}{s_j}  \geq N_\ell (\o)^2 \bigl( \sum _{j=1}^{N_\ell(\o)} s_j\bigr)^{-1}\,.
\en
In fact, the first bound follows by computing exactly the conductivity of $N_\ell(\o)$ linear resistor chains, while the second bound follows from Jensen's inequality. Since, for some positive constant $C$,   $\cP$--a.s.  $N_\ell(\o)\geq C \,\ell^{d-1}$ for $\ell$ large  (combine \cite[Remark (d) in Section 5]{GM} with Borel-Cantelli lemma) and since $\sum _{j=1}^{N_\ell(\o)} s_j $ is upper bounded by the total number of bonds in $\L_\ell$,  we conclude that $c= m^{-1}\lim _{\ell \to +\infty} \ell^{2-d} \s_\ell(\o)>0 $.

Differently from  \cite[Problem~1.18]{Bi}, where  the resistor network  in $\L_\ell$ is made only by  the bonds in the supercritical percolation cluster, in e.g.~\cite{CC}   the resistor network   in $\L_\ell$ is made  by all    bonds $\{x,y\}\in\bbE_d$ with $\o_{\{x,y\}}=1$  \rosso{as in Figure~\ref{fig_reticoli}-(middle)}. On the other hand, we claim that  $\cP$--a.s.   the two resistor networks have the same conductivity along a given direction for $\ell$ large enough (hence the scaling limit  is the same). Indeed, by combining Borel-Cantelli lemma with Theorems (8.18) and (8.21) in \cite{G}, we get  that   the following holds for some $C>0$  $\cP$--a.s. eventually in $\ell$: for all $x\in \L_\ell\cap \bbZ^d$, if  the cluster  in the bond percolation $\o$ containing $x$ is finite, then it has radius at most $C \ln \ell$. As a consequence of this result, fixed a given direction, $\cP$--a.s. the bonds not belonging  to the infinite percolation cluster do not contribute to the  directional conductivity if the box side length $\ell$ is large enough. As known, by similar arguments (apply \cite[Theorem~(5.4)]{G}) one gets that $\cP$--a.s.~the directional conductivity $\s_\ell(\o)$ is zero for $\ell$ large when considering the   subcritical bond percolation  and the resistor network on $\L_\ell$ given by all bonds  $\{x,y\}\in\bbE_d$ with $\o_{\{x,y\}}=1$.


Results similar to the above ones can be derived for  site percolation or for  positive random conductances on the supercritical cluster as in \cite[Section~5]{Fhom3}.


\subsection{Miller-Abrahams (MA) random resistor network}\label{sec_MA}
The MA  random resistor network plays a fundamental role in the study of 
electron transport in amorphous media as doped semiconductors in the regime of strong Anderson localization  (cf. \cite{AHL,MA,POF,SE}). With our notation, it is obtained by taking as $\O$ the space of   configurations of marked simple point processes \cite{DV}. In particular, an element of $\O$ is a countable set $\o=\{(x_i,E_i)\}$ with $\{x_i\}\in \cN$, such that if $x_i=x_j$ then $E_i=E_j$.
$E_i\in \bbR$ is a real number, called \emph{energy mark} of $x_i$.  The simple point process $\O\ni \o \mapsto \hat \o \in \cN$ is given by $\hat \o:=\{x_i\}$ if $\o=\{(x_i,E_i)\}$. The group $\bbG:= \bbR^d$ acts on $\O$ as  $\t_x \o:=\{(x_i-x,E_i)\}$ if $\o= \{(x_i,E_i)\}$.  Moreover, between any pair of nodes $x_i$ and $x_j$, there is a filament with  conductance  \eqref{birillo}.
 In the physical applications (cf. \cite{FSS,SE} and references therein), $\cP$ is obtained by starting with a simple point process $\{x_i\}$  with law $\hat \cP$ and marking points with i.i.d. random variables $\{E_i\}$ (independently from the rest) with common distribution of the form $\nu(dE)= c\, |E|^{\alpha} dE$ with  finite support  $[-A,A]$,  for some exponent $\a\geq 0$ and some $A>0$. Then (A1) and (A2) are satisfied if $\hat \cP$ is stationary  and ergodic   (see \cite{DV}) with positive finite intensity. (A3), (A4), (A5), (A6) are  automatically satisfied.  As discussed in \cite[Section~5.4]{Fhom3} (A7) is satisfied when $\bbE[\hat \o([0,1]^d)^2]<+\infty$, while (A8) is always satisfied \rosso{(for (A8) we used in \cite{Fhom3} that $\O_0$ is a separable metric space and the comments before Definition \ref{fasma98})}.

\rosso{We point out that the graph $\cG(\o)$ is given here by the complete graph with vertex set $\hat \o$ and suitable weights. Figure~\ref{fig_infinito}-(right) and Figure~\ref{fig_sfinito}-(right) give 
some flavor of the MA resistor network but have been plotted with a finite number of vertexes (portions of the MA resistor networks for typical $\hat \o$ cannot be plotted).}
 As by  Theorem \ref{teo1}   the asymptotic rescaled directional conductivities can be expressed in terms of $D=D(\b)$, which coincides with the diffusion matrix of Mott's random walk, we get the following:
 \begin{Corollary}\label{daie}
 \rosso{The asymptotic rescaled directional conductivities  of the MA resistor network
  satisfy} for   $\b\uparrow +\infty$ the   bounds in agreement with Mott's law obtained for the diffusion matrix of Mott's random walk  stated  in   \cite[Thm.~1]{FM} and \cite[Thm.~1]{FSS}  for $d\geq 2$ and in  \cite[Thm.~1.2]{CP1} for $d=1$.
  \end{Corollary}
   For other rigorous results on the  MA  resistor network we refer to \rosso{\cite{Mott_sinfonia,FAMH1,FAMH2}}.

 \subsection{Resistor network on admissible stationary stochastic lattices} \label{astl}  \rosso{We now deal  with a class of models with  geometric randomness,  studied in variational calculus and material science (cf.~\cite{ACG,BBL} and references therein). }
 \rosso{As in \cite{ACG,BBL} we define a \emph{stochastic lattice} as a measurable map $\cL:\O\to 
 (\bbR^d)^{\bbZ^d}$, where $(\O,\cF,\cP)$ is the probability space  and  $(\bbR^d)^{\bbZ^d}$ is endowed with the product topology and the Borel $\s$--field. The idea is to think of $\cL(\o)$ as a random deformation of the lattice $\bbZ^d$, the point $i\in \bbZ^d$ is moved to   $\cL(\o)_i$. To avoid pathological cases, we assume that the points $\cL(\o)_i $ are all distinct as $i$ varies among $\bbZ^d$.}
  
\rosso{The stochastic lattice is called \emph{stationary}   if the group $\bbG=\bbZ^d$ acts on $\O$ with an action $(\t_z)_{z\in \bbZ^d}$, $\cP$ is stationary w.r.t. this action and $\cP$--a.s. $\cL(\t_z\o)=\cL(\o)-z$ in the sense that $\cL(\t_z\o)_i = \o_i-z$ for all $i \in \bbZ^d$ (we have chosen the sign as in \cite{BBL}). 
 Given $\o\in \O$ we define $\hat \o:=\{ \cL(\o)_i\,:\, i \in \bbZ^d\}$, then the map $\O \in \o \mapsto \hat \o\in \cN$ defines a simple point process.  By writing again $(\t_z)_{z\in \bbZ^d}$ for the trivial action of $\bbZ^d$ on $\bbR^d$ by translations (i.e.~$\t_z x:= x+z$), we then have for $\cP$--a.a. $\o$ that $\widehat{\t_z \o}= \t_{-z} \hat \o$ as in \eqref{base}.  
 As in \cite{ACG},  the stochastic lattice is called \emph{admissible}
 if there exist $R,r$ such that, for $\cP$--a.a.~$\o$,  (i) every ball in $\bbR^d$ with radius $R$ contains at least  a point of $\hat \o$ and  (ii) any two points of $\hat \o$ have distance at least $r$.   The above conditions imply that $\cP$--a.s. $\hat \o$  has neither dense regions nor big holes. Apart from ergodicity, the above admissible stochastic lattices fulfill assumptions (A1) and (A2). To apply our results, we need to require ergodicity.}
 
\rosso{Suppose now   to  have a conductance field satisfying the covariant relation \eqref{montagna} and the symmetry relation in (A5) such that $c_{x,y}(\o)\leq \Theta (|x|)$ with $\Theta$ bounded,  decreasing and measurable. Then, due to property (ii),  for an admissible stochastic lattice we have the uniform bound 
 $\l_k (\o) \leq  C \int_0 ^\infty s^{d-1+k} \Theta(s)ds$ $\cP$--a.s. Hence, the moment bounds in (A7) are satisfied 
 whenever 
 \be\label{caverna}
  \int_0 ^\infty s^{d+1} \Theta(s)ds<+\infty\,,
  \en as we assume.  We point out that the above observations do not use property (i) in the definition of admissible  stochastic lattice.  }

\rosso{In order to split the energy into a short-range term and a long-range term, as in \cite{ACG} we consider the Voronoi tessellation associated to $\hat \o$ (cf.~\cite{M}). We  say that $x,y\in \hat \o$ are nearest-neighbor (n.n.) if their Voronoi cells share a $(d-1)$-dimensional face.  If the stochastic lattice is admissible then  nearest-neighbor points  have distance in $[r, 8R]$  $\cP$--a.s. (cf.~\cite[Lemma~1]{ACG}).   To make a connection with the theory of $\G$--convergence of stochastic integrals and in particular to \cite{ACG}, 
 let us further specify our model and take  conductances  of the form
 \[
 c_{x,y}(\o)=\begin{cases}
  \a_{x,y}(\o) g_{\rm nn} (y-x)  &\text{ if } x,y \text{ are n.n.}\,,\\
  \a_{x,y}(\o) g_{\rm lr} (y-x)  &\text{  otherwise}\,,
  \end{cases}
  \] 
  with $\a_{x,y}(\o) \leq C$ for some absolute constant $C$ (for simplicity). The subscripts nn and lr refer  to the nearest-neighbor and long-range case, respectively.  We  set $f_{\rm nn}(v, \z):=g_{\rm nn}(v)|v|^2 |\z|^2$ and $f_{\rm lr}(v, \z):=g_{\rm lr}(v)|v|^2 |\z|^2$, where $v\in \bbR^d$and $\z\in \bbR$. 
  Since n.n. points have distance in  $[r,8R]$, if $g_{\rm nn}$ is locally bounded,
  then our moment condition \eqref{caverna} on $\Theta$ is assured by the upper bound concerning $f_{\rm lr}$ in  \cite[Eq. (13), (15)]{ACG}  (essentially, they are the same condition). Note also that, if $g_{\rm nn}(v)>0$ whenever $|v|\in [r, 8R]$, then our assumption (A6) is automatically satisfied.}

  \rosso{By \eqref{chiave}  and Lemma \ref{benposto} below we have
\[ \ell^{2-d} \s_\ell(\o)= 
 \frac{1}{2} \la   \nabla _\e V_\e, \nabla_\e V_\e \ra _{L^2(\nu ^\e_{\o,\L})}
 =\inf \Big\{ \frac{1}{2} \la \nabla_\e u, \nabla _\e u\ra _{L^2(\nu^\e_{\o,\L})}\,:\, u \in K^\e_\o\Big\}\,,
 \]
 where $K^\e_\o$ has been introduced in Definition \ref{def_hilbert}.
Note that 
\be\label{occhi_bianchi}
\begin{split}
\frac{1}{2}\la \nabla_\e u, \nabla _\e u\ra _{L^2(\nu^\e_{\o,\L})} 
&=  \frac{1}{2}
\e^{d} \sum _{ (\e x,\e y) \in \cE_\e(\o): x,y \text{\;n.n.}  }\a_{x ,y}(\o) f_{\rm nn}\left(y-x, \frac{u(\e y)-u(\e x)}{\e |y-x| }\right) \\
& +\frac{1}{2}
\e^{d} \sum _{ (\e x,\e y) \in \cE_\e(\o): x,y \text{\,non n.n.}  }\a_{x ,y}(\o) f_{\rm lr}\left(y-x, \frac{u(\e y)-u(\e x)}{\e |y-x| }\right)
\,.
\end{split}
\en
The r.h.s. of \eqref{occhi_bianchi} is the energy functional, split into a short-range term and a long-range term.
  Since $\ell^{2-d}\s_\ell(\o)$ is the minimum of the energy functional  \eqref{occhi_bianchi} as $u$ varies in $K^\e_\o$, the variational problem under consideration fits with the one in \cite[Section~3.1]{ACG} apart from the boundary conditions, the stochastic interaction terms $\a_{x,y}(\o)$ and  that there are  interactions between sites inside and outside the domain. With the same exceptions, we refer to \cite{ACG} (and in particular to Theorem~2 there) for $\G$--convergence results of the energy functional for stationary admissible stochastic lattices (see also \cite[Section~3.3]{ACG} where stochastic interaction terms are considered when $\hat\o\equiv \bbZ^d$).}

  
%
%

 \subsection{Periodic models} 
\label{app_esempio}
\rosso{Let us fix a basis $v_1,v_2,\dots, v_d$ of $\bbR^d$ and consider the parallelogram  $\D:=\{\sum_{i=1}^d t_i v_i: t_1,\dots, t_d\in [0,1)\}$. The group $\bbG=\bbZ^d$ acts on $\bbR^d$ by the translations $\t_g x:= x+ g_1 v_1+\cdots+g_d v_d$ where $g=(g_1,\dots, g_d)$.
We fix in $\D$ a  finite nonempty  set $\cA$ of points  and consider the infinite set $\cV:=\cup_{g\in \bbG} \t_g \cA$. We suppose  to have symmetric weights $\a_{x,y}=\a_{y,x}\in [0,+\infty)$ for $x,y\in \cV$ which are periodic in the sense that $\a_{\t_g x,\t_g y}=\a_{x,y}$ for all $x,y\in \cV$ and $g\in \bbG$. We are interested to the resistor network built on the undirected weighted graph $\cG$  with vertex set $\cV$, edges given by the pairs $\{x,y\}$ with $x\not =y$ in $\cV$  and with $\a_{x,y}>0$, whose weight is given by $\a_{x,y}$. For example one can consider the standard lattice $\bbZ^d$ or the standard hexagonal lattice with unit weight on the edges. One can consider as well  $\cV:=\bbZ^d$ with weights such that $\a_{x,y}=\a_{x+Ng, y+Ng}$ for some fix $N\geq 1$  and all $x,y,g\in\bbZ^d$ as  in Figure \ref{cheronte} (in this case one has to take $v_i=N e_i$).}

\rosso{Our general results can be applied to the above setting. Some care is necessary to fulfill (A3). To this aim  we mark   points of $\cV$ by i.i.d. non degenerate random variables with common distribution $\rho$, i.e. we take the product probability space $\O:=\bbR^{\cV}$, $\cP:= \otimes _{x\in \cV} \rho$. $\bbG$ acts on $\O$ as $(\t_g \o)_x:= \o _{\t_{g}x}$.  The simple point process is given by $\hat \o:=\cV$ and the conductances are given by $c_{x,y}(\o)=\a_{x,y}$ for $x,y\in \cV$. Then all our assumptions (A1),...,(A5), (A8) are satisfied. (A6) means that the graph $\cG$ is connected. Due to the form of the Palm distribution when $\bbG=\bbZ^d$ (and not necessarily $\hat \o \subset \bbZ^d$) discussed in \cite[Section~2]{Fhom3} and \cite[Example~5.6]{Fhom3}, we have that (A7) is satisfied whenever $\sum_{y\in \cV} \a_{x,y}|x-y|^2<+\infty$ for all $x\in \cA$.}

\begin{figure}[h!]
\includegraphics[scale=0.25]{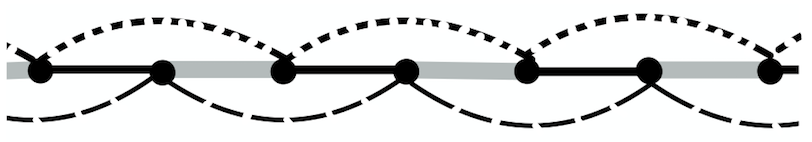}
\caption{Periodic resistor network with period $2$. Similar  filaments have the same conductance. } \label{cheronte}
\end{figure}

 \subsection{Other examples} \rosso{For the discussion of the validity of Assumptions (A1),...,(A8) for resistor networks built on $\bbZ^d$ with long conductances we refer to \cite[Section~5.2]{Fhom3}, for resistor networks built on crystal lattices we refer to \cite[Section~5.2]{F_SEP}  and \cite[Section~5.6]{Fhom3}, for resistor networks built on the Delaunay triangulation (see Figure~\ref{fig_infinito}-(middle) and Figure~\ref{fig_sfinito}-(left)) we refer to \cite[Section~5.5]{Fhom3} when the set of vertexes is a Poisson point process and to \cite{FT} for more general simple point processes, for 
 resistor networks built on supercritical clusters of the Boolean model (see Figure~\ref{fig_infinito}-(left) in case of a deterministic radius) and   the random-connection model we refer to   \cite{For_Francis}. All the above graphs have to be thought as weighted with random conductances.
 We refer to \cite{MR} for the definition of the Boolean model and the random-connection model. 
By  Delaunay triangulation associated to the set $\hat \o$ we mean the graph with vertex set $\hat \o$ and edges given by the  pairs $\{x,y\}$ with $x\not=y$ in $\hat \o$ such that the Voronoi cells of $x$ and $y$ share a $(d-1)$--dimensional face. Note that, differently from Section \ref{astl}, dealing e.g. with the Poisson point process this graph has both big holes and arbitrarily dense regions.
 }

\section{Preliminary facts on the Palm distribution $\cP_0$}\label{GS}
In this short section we  recall some basic  facts  on the Palm distribution $\cP_0$  used in what follows.  Recall \eqref{pandi}.

\begin{Lemma}\label{matteo}  
  \cite[Lemma~7.1]{Fhom3}
Given a measurable subset $A\subset \O_0$, the following facts are equivalent: (i) 
 $\cP_0(A)=1$; (ii)
$\cP\left(\o \in \O\,:\, \t_x \o \in A \; \forall x \in \hat\o\right)=1$;
(iii)  $\cP_0\left(\o \in \O_0\,:\, \t_x \o \in A \; \forall x \in \hat\o\right)=1$.
\end{Lemma}

By ergodicity (cf.~(A1)),  given a  measurable function $g: \O_0\to [0,\infty)$  it holds
 \begin{equation}\label{lattone1}
 \lim_{s \to \infty} \frac{1}{(2s)^d } \int_{(-s,s)^d  }  d\hat{\o}(x) \, g(\t_x \o) = m \,\bbE_0[ g ]\qquad \cP\text{--a.s.}
 \end{equation}
 One can indeed refine the above result. 
To this aim 
we define $\mu^\e_\o$ as the atomic measure on $\bbR^d$ given by  $\mu^\e_\o := \e^d \sum _{x\in \hat \o} \d_{\e x}$. Then it holds:
\begin{Proposition}\label{prop_ergodico}\cite[Prop.~3.1]{Fhom3} Let  $g: \O_0\to \bbR$ be a  measurable function with $\|g\|_{L^1(\cP_0)}<+\infty$. Then there exists a translation invariant   measurable  subset $\cA[g]\subset \O$  such that $\cP(\cA[g])=1$ and such that,  for any $\o\in \cA[g]$ and any  $\varphi \in C_c (\bbR^d)$, it holds
\begin{equation}\label{eq-limitone}
\lim_{\e\da 0} \int  d  \mu_\o^\e  (x)  \varphi (x ) g(\t_{x/\e} \o )=
\int  dx\,m\varphi (x)  \bbE_0[g]\,.
\end{equation}
\end{Proposition}
The above proposition (which is the analogous e.g. of \cite[Theorem~1.1]{ZP}) is at the core of 2-scale convergence. It follows again from  ergodicity. The variable $x$ appears in the l.h.s. of \eqref{eq-limitone} at  the macroscopic scale in $\varphi (x)$ and at the microscopic scale in $g(\t_{x/\e} \o )$.

\begin{Definition}\label{sibilla}  Given a function  $g: \O_0 \to [0,+\infty]$ such that  $\bbE_0[ g] <+\infty$,  we define $\cA[g]$ as $\cA[ g_0]$ (cf. Proposition \ref{prop_ergodico}), where  $g_0: \O_0 \to \bbR$ is defined as $g$ on $\{ g <+\infty\}$ and as $0$ on $\{ g =+\infty\}$. 
\end{Definition}

\begin{Remark}\label{dono} By playing with suitable test functions $\varphi$, one derives from Proposition \ref{prop_ergodico} that the limit \eqref{lattone1} holds indeed for all $\o$ varying in a suitable  translation invariant   measurable  subset  of $\O$ with $\cP$--probability $1$.
\end{Remark}

\section{The Hilbert space $H^{1,\e}_{0,\o}$ and  the amorphous gradient $\nabla_\e f$} \label{sec_hilbert}
Let $\o \in \O_1$.  Recall   Definition \ref{pinco52} of $\cG_\o^\e$, $\cC_{\o,\L}^\e$ and  $\cC_{\o}^\e$ and Definition \ref{def_hilbert} of $H^{1,\e} _{0,\o} $  and $K^\e_\o$. 
For later use, we point out that 
 $K^{\e}_{\o}=H^{1,\e}_{0,\o}+ \z_\e $, where 
$\z_\e  :\e \hat \o \cap S\to [0,1]$ is the function (depending also on $\o$) 
\be\label{def_psi_e}
\z_\e (x):=
\begin{cases}
x_1+\frac{1}{2} & \text{ if } x\in \cC_{\o}^\e \cap  \L\,,\\
0 & \text{ if }x\in (\e \hat \o \cap S^-)\cup  \cC_{\o,\L}^\e\,,\\
1 & \text{ if } x\in  \e \hat \o \cap S^+\,.
\end{cases}
\en

As discussed in Section \ref{MM}, if $f:\e \hat \o \cap S\to \bbR$ is bounded, then $f\in L^2(\mu^\e_{\o,\L} )$, $\nabla_\e f \in L^2(\nu^\e_{\o,\L})$ and $\bbL^\e _\o f\in   L^2(\mu^\e_{\o,\L})$. By definition of $\nu^\e_{\o,\L}$, given bounded functions $f,g:\e \hat \o\cap S \to \bbR$, we have 
\be\label{violino5}
\la \nabla_\e f, \nabla_\e g \ra _{L^2(\nu^\e_{\o,\L} )}= \e^{d-2} \sum _{ (x,y) \in \cE_\e(\o)  }c_{x/\e,y/\e}(\o) \bigl( f(y)-f(x) \bigr)  \bigl( g(y)-g(x) \bigr)\,.
\en

\begin{Lemma}\label{spiaggia} Let $\o \in \O_1$. Given $f,g: \e \hat \o\cap S \to\bbR$ with 
$f\equiv 0$   on $\e \hat \o \setminus \L$ and  $g$ bounded, it holds
\be\label{gallina}
\la f, -\bbL^\e_\o g \ra_{L^2(\mu^\e_{\o,\L} )}=
 \frac{1}{2} \la   \nabla _\e f, \nabla_\e g  \ra _{L^2(\nu ^\e_{\o,\L})} \,.
 \en
\end{Lemma}
\begin{proof}
Since $f\equiv 0$ outside $\L$ 
we have 
\be\label{consegna}
\begin{split}
\la f, -\bbL^\e_\o g \ra_{L^2(\mu^\e_{\o,\L} )}=
- \sum _{x\in \e \hat \o \cap S} \e^{d-2} f(x) \sum _{y \in \e \hat \o \cap S} c_{x/\e,y/\e}(\o) \bigl( g(y)- g(x) \bigr)\,.
\end{split} 
\en
The r.h.s. is an absolutely convergent series  as $\o \in \O_1$, $f$ and $g$ are bounded and $f\equiv 0$ on $S\setminus \L$, hence we can freely permute the addenda.
Due to the symmetry of the conductances, the r.h.s. of \eqref{consegna} equals
\[-\sum _{y\in \e \hat \o \cap S} \e^{d-2} f(y) \sum _{x \in \e \hat \o \cap S} c_{x/\e,y/\e}(\o) \bigl( g(x)- g(y) \bigr)\,.
\] By summing  the above expression with  the r.h.s. of \eqref{consegna}, we get
\[ \la f, -\bbL^\e_\o g \ra_{L^2(\mu^\e_{\o,\L} )}=\frac{
  \e^{d-2}}{2} \sum _{x\in \e \hat \o \cap S }
  \sum_{y\in \e \hat \o \cap S} c_{x/\e,y/\e}(\o) \bigl( f(y)-f(x) \bigr)  \bigl( g(y)-g(x) \bigr)\,.
\]
As the generic addendum in the r.h.s. is zero if $(x,y)\not \in \cE_\e (\o)$ since $f\equiv 0$ on $S\setminus \L$,  by \eqref{violino5} we get \eqref{gallina}.
\end{proof}
%
%
%
%
%
\rosso{Due to the following lemma,  Definition \ref{def_EP} of $V^\o_\ell$  (equivalently, the definition of $V_\e$) is well posed.}
\begin{Lemma}\label{benposto} Given   $\o \in \O_1$, 
 the following holds:
\begin{itemize}
\item[(i)] \rosso{there is} a unique function $v\in K_\o^\e$ such that $\bbL^\e_\o v (x)=0$ for all $x \in \e\hat \o \cap \L$;
\item[(ii)] \rosso{there is} a  unique  function $v\in K^\e_\o  $ such that 
$
\la  \nabla_\e u,\nabla_\e v\ra _{L^2(\nu ^\e_{\o,\L} ) }=0$ for all $u \in H^{1,\e}_{0,\o}$;
\item[(iii)] \rosso{there is} a   unique minimizer $v$  of  the following variational problem:
\be\label{avengers}
\inf \bigl\{  \la \nabla_\e \rosso{u}, \nabla _\e \rosso{u}\ra _{L^2(\nu^\e_{\o,\L})}\,:\, \rosso{u} \in K^\e_\o\bigr\}\,. 
\en
\end{itemize}
\rosso{Moreover, the above functions $v$ coincide and therefore equal  $V_\e$ (which was introduced  via  $V^\o_\ell$ as the function $v$ in Item (i))}.
\end{Lemma}
\begin{proof}  On the finite dimensional Hilbert space $ H^{1,\e}_{0,\o}$
we consider the bilinear form $a(f,g):= \frac{1}{2}\la \nabla _\e f, \nabla _\e g \ra_{L^2(\nu^\e_{\o,\L})}$. Trivially, $a(\cdot, \cdot)$ is a continuous symmetric bilinear form.   If  $a(f,f)=0$, then  $f$ is constant on the connected components of the graph $\cG_\o^\e$. This fact and the definition of  $ H^{1,\e}_{0,\o}$ imply that $f\equiv 0$. As $ H^{1,\e}_{0,\o}$ is finite dimensional, we conclude that $a(\cdot, \cdot)$ is also coercive.
By writing
$v= f_\e +\z_\e$, the function $v\in K^\e_\o$ in Item (i) is the  one such that   $f_\e \in H^{1,\e}_{0,\o}$ and 
\be\label{sibillino}
\bbL^\e_\o f_\e( x) = - \bbL^\e _\o  \z_\e (x)
\qquad \forall   x \in \e \hat \o \cap \L\,.\en
As $f_\e \equiv 0 \equiv \z_\e$ on $\cC^\e_{\o,\L}$, which is a union of connected components of $\cG_\o^\e$,   \eqref{sibillino} is automatically satisfied for $x\in \cC^\e_{\o,\L}$. Hence in \eqref{sibillino} we can restrict to $x\in \cC^\e_\o\cap \L$. On the other hand, functions in $H^{1,\e}_{0,\o}$ are free on $\cC^\e_\o\cap \L$, and zero elsewhere.
Due to the above observations,  \eqref{sibillino} is equivalent to requiring that $\la u, -\bbL^\e_\o f_\e\ra_{L^2(\mu^\e_{\o,\L})}= \la u, \bbL^\e _\o  \z_\e \ra_{L^2(\mu^\e_{\o,\L})}$ for any $u\in H^{1,\e}_{0,\o}$. Hence, 
by     Lemma \ref{spiaggia},   $f_\e\in H^{1,\e}_{0,\o}$ satisfying \eqref{sibillino}  can be characterized also as the solution  in  $H^{1,\e}_{0,\o}$  of the problem 
\be \label{marino}
a( f_\e, u)= -
 \frac{1}{2}\la \nabla _\e \z_\e , \nabla _\e u \ra_{L^2(\nu^\e_{\o,\L})}
 \qquad \forall u \in  H^{1,\e}_{0,\o}\,.
\en
By the Lax--Milgram theorem we conclude that there exists a unique function $f_\e$ satisfying \eqref{marino}, thus implying Item (i).  
Since $a( f_\e, u)= \frac{1}{2}\la \nabla_\e f_\e , \nabla_\e u \ra_{L^2(\nu^\e_{\o,\L})} 
$, the uniqueness of the solution $f_\e$ of \eqref{marino} leads to Item (ii).
Moreover, always by  the  Lax--Milgram theorem,  $f_\e$ is the unique minimizer of  the functional 
$H^{1,\e}_{0,\o}    \ni \rosso{u} \mapsto \frac{1}{2} a(\rosso{u},\rosso{u}) + \frac{1}{2}\la \nabla _\e\z_\e, \nabla _\e \rosso{u} \ra_{L^2(\nu^\e_{\o,\L})}\rosso{\in \bbR}$, and therefore of the functional 
$H^{1,\e}_{0,\o}    \ni\rosso{u} \mapsto\frac{1}{4} \la  
  \nabla_\e (\rosso{u}+\z_\e) , \nabla_\e (\rosso{u}+\z_\e) \ra _{L^2(\nu^\e_{\o,\L})}\rosso{\in\bbR}$.
 This proves Item (iii). \rosso{The conclusion then follows from the above observations.}
 \end{proof}

\begin{Remark}\label{fiorellino3} As  $V_\e\rosso{\in K^\e_\o}$ is   harmonic on $\e \hat \o \cap \L$ (cf. Lemma \ref{benposto}-(i)),  $V_\e$ has values in $[0,1]$.
\end{Remark}

\begin{Remark}\label{sirenetta} Given $(x,y)\in \cE_\e(\o)$,  $\{x,y\}$ is 
included either in $\cC^\e_{\o,\L}$ or in $\cC^\e_\o$.    In the first  case we have $\nabla_\e V_\e (x, (y-x)/\e)=0$. This observation and Lemma \ref{benposto}-(ii) imply that 
$\la \nabla_\e f, \nabla_\e V_\e\ra_{L^2(\nu^\e_{\o,\L})} =0 $ for all $ f: \e \hat \o \cap S \to \bbR$  such that $ f(x)=0 $ $\forall  x \in \e\hat \o \cap (S\setminus \L)=\cC_\o^\e\setminus \L $.
\end{Remark}


\rosso{For} the next result, recall \rosso{the definition \eqref{def_psi_e} of $\z_\e$ and define $\z:S\to [0,1]$ as}
\be\label{def_psi}
\z(x):= x_1+\frac{1}{2}  \;\text{ if }\; x\in \L\,,\;\;\;\;\;  \z(x):=0 \; \text{ if }\; x\in S^-
\,,\;\;\; \;\;\z(x):=1 \; \text{ if }\; x\in S^+\,.
\en

\begin{Lemma}\label{paletta}  There exists a  translation invariant measurable  subset $\O_2\subset \O_1$ such that $\cP(\O_2)=1$ and,   for all $\o\in \O_2$,
\begin{align} 
& \varlimsup_{\e \da 0}   \| \z\|_{L^2(  \mu^\e_{\o, \L})} <+\infty \,,\, \qquad\varlimsup_{\e \da 0}   \|\nabla_\e \z\|_{L^2(\nu^\e_{\o,\L})} < +\infty\,,\label{parasole0}
\\
&  \varlimsup_{\e \da 0}   \| \z_\e\|_{L^2(  \mu^\e_{\o, \L})} <+\infty \,, \qquad\varlimsup_{\e \da 0}   \|\nabla_\e \z_\e\|_{L^2(\nu^\e_{\o,\L})} < +\infty\,,\label{parasole1}
\\
&  \varlimsup_{\e \da 0}   \| V_\e \|_{L^2(\mu^\e_{\o,\L})} <+\infty \,, \qquad\varlimsup_{\e \da 0}   \|\nabla_\e V_\e\|_{L^2(\nu^\e_{\o,\L})} < +\infty\,.\label{parasole}
\end{align}
\end{Lemma} 
\begin{proof}
By Proposition \ref{prop_ergodico}  and Remark \ref{dono}  there exists a  translation invariant  measurable
  set $\O_2\subset \O_1$ such that,  $\forall \o \in \O_2$,
 $\lim_{\e  \da 0}\mu^\e _\o (\L) = m $ and $\lim_{\e \da 0} \int_\L d\mu^\e_\o (x) \l_2 (\t _{x/\e}\o) = m \bbE_0[\l_2] $.

Let us take $\o \in \O_2$. Since  $\z,\z_\e, V_\e$ have value in $[0,1]$ and $ \mu^\e_{\o,\L}$ has mass $\mu^\e_\o (\L)\to m$, we get the first bounds in \eqref{parasole0},  \eqref{parasole1} and  \eqref{parasole}.

 Let us prove that $ \varlimsup_{\e \da 0}   \|\nabla_\e \z\|_{L^2(\nu^\e_{\o,\L})} < +\infty$.
 We have  (recall \eqref{violino5})
\be\label{carretto}
\begin{split}
 \|  \nabla_\e \z \|^2  _{L^2(\nu^\e_{\o,\L})} &=\e^{d-2} \sum _{ (x,y) \in \cE_\e(\o) } c_{x/\e,y/\e}(\o) \left( \z( y) -\z (x) \right)^2 \\ & \leq 
 \e^{d-2} \sum _{ (x,y)  \in \cE_\e(\o) } c_{x/\e,y/\e}(\o) \left( y_1-x_1\right)^2 
 \\
 &\leq 2 \e^{d-2} \sum _{x\in \e\hat \o  \cap \L} \sum _{ y \in \e \hat \o \cap S} c_{x/\e,y/\e}(\o)  \left(y_1 -x_1 \right)^2\,.
  \end{split}
\en
We can rewrite the last expression as $2 \e^d  \sum _{x\in \hat \o  \cap (\e^{-1} \L) } \sum _{ y \in  \hat \o \cap(\e^{-1} S)} c_{x ,y}(\o)  (y_1-x_1)^2$,
which is   bounded by $ 2\e^d  \sum _{x\in \hat \o  \cap (\e^{-1} \L) } \l_2(\t_x \o)=2\int_\L d\mu^\e_\o (x) \l_2(\t _{x/\e} \o) $. The last integral  converges to $2 m\bbE_0 [\l_2]<+\infty $ as  $\e\da 0$ since $\o \in \O_2$. This concludes the proof that 
$ \varlimsup_{\e \da 0}   \|\nabla_\e \z\|_{L^2(\nu^\e_{\o,\L})} < +\infty$.

To prove that $ \varlimsup_{\e \da 0}   \|\nabla_\e \z_\e\|_{L^2(\nu^\e_{\o,\L})} < +\infty$ it is enough to use the above result and note  that $c_{\rosso{x/\e,y/\e}}(\o) | \z _\e(y)-\z_\e (x) | \leq c_{\rosso{x/\e,y/\e}}(\o) | \z  (y)-\z(x)|$  $\rosso{\forall (x,y)\in \cE_\e(\o)}$.

 Since $V_\e $ minimizes \eqref{avengers} and $\z_\e \in K^\e_\o$, 
 we have  $  \|\nabla_\e V_\e\|_{L^2(\nu^\e_{\o,\L})}  \leq   \|\nabla_\e \z _\e\|_{L^2(\nu^\e_{\o,\L})} $. Hence 
   $\varlimsup_{\e \da 0} \|\nabla_\e V_\e\|_{L^2(\nu^\e_{\o,\L}) }  <+\infty $ 
   by the second bound in  \eqref{parasole1}.
\end{proof}

\subsection{Some properties of the amorphous gradient $\nabla_\e$}\label{gnocchi} In Section \ref{MM} we have defined $\nabla _\e f $ for functions $f: \e \hat \o\cap S\to \bbR$. Definition \eqref{ricola} can by extended by replacing $S$ with any set $A\subset \bbR^d$ by requiring that $x, x+\e z\in \e \hat \o \cap A$.
Given $f,g : \e \hat \o \to \bbR$, it is simple to check the following Leibniz rule:
 \begin{equation}\label{leibniz} \nabla _\e  (fg)(x,z)
   =\nabla _\e  f (x, z ) g (x )+ f (x+\e z ) \nabla _\e g  ( x, z)\,.
\end{equation}
Given $\varphi \in C_c(\bbR^d)$, 
 take $\ell$  such that $\varphi (x)=0$ if $|x| \geq \ell$ and fix   $\phi \in C_c(\bbR^d)$ with values in $[0,1]$, such that $ \phi(x)=1$ for $|x| \leq \ell$ and $\phi(x)=0$ for $|x| \geq \ell+1$. 
 Then, by applying respectively the mean value theorem and Taylor expansion, we get for some positive constant $C(\varphi)$  that 
 \begin{align}
& \bigl | \nabla _\e \varphi(x,z) \bigr | \leq \| \nabla \varphi \|_\infty |z| \bigl( \phi(x)+ \phi(x+\e z) \bigr)   \text{ if } \varphi \in C^1_c(\bbR^d)\,,\label{paradiso}\\
& \bigl| \nabla _\e \varphi (x,z) - \nabla \varphi (x) \cdot z\bigr | \leq \e C(\varphi) |z|^2 \bigl( \phi(x) + \phi(x+\e z) \bigr) 
   \text{ if } \varphi \in C^2_c(\bbR^d)\,.\label{mirra}
 \end{align}

\section{Proof of Theorem \ref{teo1} \rosso{when  $e_1\in {\rm Ker}(D)$}}\label{renato_zero} 
Due to \eqref{mortisia} we need to prove  that   
$\lim_{\e \da 0}  \la \nabla_\e V_\e, \nabla_\e V_\e\ra_{L^2(\nu^\e_{\o,\L})}=0$ for all $\o$ varying in a good set,   where  \emph{good set} means   a translation invariant measurable subset of $\O_1$ with  $\cP$-probability equal to $1$. 
Trivially it is enough to prove the following claim: given $\d>0$, $\varlimsup_{\e \da 0}\frac{1}{2} \la \nabla_\e V_\e, \nabla_\e V_\e\ra_{L^2(\nu^\e_{\o,\L})} \leq \d
$  for all $\o$ in a good set.  Let us prove our claim.

As $D_{1,1}=0$ \rosso{(since $e_1\in {\rm Ker}(D)$)}  and by \eqref{def_D},
given $\d>0$ we can  fix $f\in L^\infty(\cP_0)$ such that 
\be\label{moonstar}
m\bbE_0\Big[\int d\hat \o (x) c_{0,x}(\o) \left (x_1 - \nabla f (\o, x) 
\right)^2\Big]\leq \d\,.
\en
Given $\e>0$ we define the functions $v_\e, \tilde v_\e: \e \hat \o \cap S\to \bbR$ as
$v_\e(x):=\z(x) - \e f (\t_{x/\e}\o )$ if  $ x\in \L$; $v_\e(x):=0$ if $x\in S^-$; $v_\e(x):=1$ if $x\in S^+$ and  $\tilde v_\e(x):= v_\e(x) \mathds{1}_{S\setminus \cC^\e_{\o,\L}}(x)$ for all $x$.
As $\tilde v_\e(x) \in K^\e_\o$, by Lemma \ref{benposto}-(iii)  it is enough to prove that   $\varlimsup_{\e \da 0} \rosso{\frac{1}{2}} \la \nabla_\e \tilde v_\e, \nabla_\e \tilde v_\e\ra_{L^2(\nu^\e_{\o,\L})}\leq \d$ for all $\o$ in a good set. To this aim we
 bound
 \[  \frac{1}{2}\la \nabla_\e \tilde v_\e, \nabla_\e  \tilde v_\e\ra_{L^2(\nu^\e_{\o,\L})}
 \leq \e^{d-2}
\sum _{x\in  \hat \o \cap  \e^{-1}\L}  \sum _{y \in  \hat \o \cap  \e^{-1} S} c_{x, y}(\o)
\bigl( v_\e (\e y) - v_\e(\e x) \bigr)^2\,.
\]
We split the sum in the r.h.s.  into three contributions $C(\e)$, $ C_-(\e)$ and $C_+(\e)$, corresponding respectively to the cases $ y \in \hat \o \cap \e^{-1} \L$,  $ y \in \hat \o \cap \e^{-1} S^-$ and
$ y \in \hat \o \cap \e^{-1} S^+$, while in all the above contributions $x$ varies among $\hat \o \cap  \e^{-1}\L$.

If $x,y\in \hat \o \cap \e^{-1} \L$, then $v_\e(\e y)-v_\e( \e x)= \e( y_1-x_1- \nabla f( \t_x \o, y-x)) $. Hence, we can bound
\be\label{134}
C(\e) \leq \e^d \sum _{x\in  \hat \o \cap  \e^{-1}\L} \; \sum _{y \in  \hat \o } c_{x, y}(\o)( y_1-x_1- \nabla f( \t_x \o, y-x)) ^2 \,.
\en
By ergodicity  (cf.~\eqref{lattone1} and Remark \ref{dono})  for all $\o$ varying in a suitable good set 
the r.h.s.  of \eqref{134} converges   to  the l.h.s of \eqref{moonstar}, and therefore it is bounded by $\d$. Hence, $\varlimsup _{\e \da 0} C(\e) \leq \d$.

We now consider $C_-(\e)$ and prove that $\lim_{\e \da 0} C_-(\e)=0$ \rosso{for $\o$ varying in a suitable good set}.
If $x\in  \hat \o \cap \e^{-1} \L$ and $y \in \hat \o \cap \e^{-1} S^-$, then $(v_\e(\e x)-v_\e(\e y))^2 = \e^2(  x_1 - f (\t_x \o))^2 \leq 2 \e^2 x_1^2 + 2 \e^2 \|f\|_\infty^2 \leq 2 \e^2 (x_1-y_1)^2 + 2 \e^2 \|f\|_\infty^2$. Hence it remains to show  that 
\be\label{nike}
\e^d \sum _{x\in  \hat \o \cap  \e^{-1}\L} \; \sum _{y \in  \hat \o \cap  \e^{-1} S^-} c_{x, y}(\o)[( x_1-y_1)^2+ 1]
\en
goes to zero as $\e\da 0$.   
Given  $\rho\in (0,1)$ we set  $\L_\rho:=\rho \L$.  We denote by 
$A_1(\rho, \e)$    the sum in \eqref{nike} restricted to $x\in \hat \o \cap \e^{-1}\L_\rho $ and  $y \in  \hat \o \cap  \e^{-1} S^-$. We denote by  $A_2(\rho, \e)$ the sum coming from the remaining addenda.
 We observe that  $
A_1(\rho, \e)\leq \e^d \sum_{x\in \hat \o \cap \e^{-1}\L_\rho} g_{(1-\rho)/2\e}  (\t_x \o)$,
where $g_\ell (\o) := \sum _{z\in \hat \o} c_{0,z}(\o) [ z_1^2 +1] \mathds{1}(|z| _\infty\geq\ell)$. By \eqref{lattone1} and Remark \ref{dono}, for all $\o$ varying in a good set  $\O'$ it holds  $\lim_{\e \da 0} \e^d \sum_{x\in \hat \o \cap \e^{-1}\L_\rho} g_{ \ell }  (\t_x \o)= m  \rho^d \bbE_0[g_\ell]$  for all $\ell\in \bbN$.  As $g_{ (1-\rho)/2\e}\leq g_\ell$ for $\e$ small and by (A7) and dominated convergence, we obtain for all $\o\in \O'$ that
$\varlimsup _{\e\da 0} 
A_1(\rho, \e)  \leq \varlimsup_{\ell \uparrow \infty} m  \rho^d \bbE_0[g_\ell]=0$.
We move to $A_2(\rho, \e)$.  We can bound $A_2(\rho, \e)$ by 
\be\label{nike7}
\e^d \sum _{x\in  \hat \o \cap  \e^{-1}(\L\setminus \L_\rho)} \sum _{y \in  \hat \o } c_{x, y}(\o)[(x_1-y_1)^2+ 1]\,.
\en
By Proposition \ref{prop_ergodico} with suitable test functions,  for all $\o$ varying in a suitable  good set we get that \eqref{nike7} converges as $\e \da 0$ to $m \bbE_0 [ \l_2+ \l_0]  \ell (\L\setminus \L_\rho)$, where here $\ell(\cdot)$ denotes the Lebesgue measure. To conclude the proof that $\lim_{\e\da 0} C_-(\e)=0$,  it is therefore enough to take the limit $\rho \uparrow 1$ \rosso{along a sequence (to deal with countable good sets)}.

By the same arguments used for $C_-(\e)$, one proves that $\lim_{\e\da 0} C_+(\e)=0$. 

\begin{Remark}\label{baguette}\rosso{We point out that the same good set of environments can be used to prove the analogous of Theorem \ref{teo1} for any direction in ${\rm Ker}(D)$ (see Warning \ref{stellina59}).  For simplicity of notation suppose that  ${\rm Ker}(D)$ is generated by $e_1,e_2, \dots, e_k$. Take a unit vector $v\in {\rm Ker}(D)$. If $f_i$ satisfies 
$m\bbE_0\bigl[\int d\hat \o (x) c_{0,x}(\o) \left (x_i - \nabla f_i (\o, x) 
\right)^2\bigr]\leq \d/k$ for any $i=1,\dots, k$, setting $f:=\sum_{i=1}^k v_i f_i$ by Schwarz inequality we have 
$m \bbE_0\bigl[\int d\hat \o (x) c_{0,x}(\o) \left (v\cdot x- \nabla f (\o, x) 
\right)^2\bigr]\leq \d $. Similarly the analogous of $C(\e)$ in \eqref{134}   can be bounded by $ \sum_{i=1}^k   \e^d \sum _{x\in  \hat \o \cap  \e^{-1}\L} \; \sum _{y \in  \hat \o } c_{x, y}(\o)( y_i-x_i- \nabla f_i( \t_x \o, y-x)) ^2$ (now the box $\L$ is along the direction $v$). This shows that $C(\e)$ for the unit vector $v$ vanishes as $\e\da 0$ for the same good set of environments for which the same holds along the directions $e_1,...,e_k$. 
The simultaneous control of the terms $C_-(\e)$ and $C_+(\e)$ for all unit vectors $v\in {\rm Ker}(D)$ is even simpler.}
\end{Remark}
\section{Effective equation with mixed  boundary conditions}\label{eff_equation}
In this section we assume that $d_*:={\rm dim} ({\rm Ker}(D)^\perp)\geq 1$ and, given a domain  $A\subset \bbR^d$, $L^2(A)$ and $H^1(A)$ refer to the Lebesgue measure $dx$, which will be omitted from the notation. Here, and in the other sections, given a unit vector $v$   we write $\partial_v f $ for the weak derivative of $f$ along the direction  $v$ (if $v=e_i$, then $\partial_v f $ is simply the standard weak derivative $\partial_i f$).


We are interested in elliptic operators with mixed (Dirichlet and   Neumann) boundary conditions.
By denoting $\bar \L$ the closure of $\L$, we set 
\[  F_-:=\{ x\in \bar \L\,:\, x_1 =-1/2\}\,,\;  \,  F_+:=\{ x\in \bar \L\,:\, x_1 =1/2\}\,, \,
\; F:= F_- \cup F_+
\,.
\]

\begin{Definition}\label{divido} 
We  fix an orthonormal basis  $\mathfrak{e}_1, \mathfrak{e}_2,\dots, \mathfrak{e}_{d_*}$   of  ${\rm Ker}(D)^\perp\rosso{={\rm Ran}(D)}$ and an orthonormal basis $\mathfrak{e}_{d_*+1},\dots, \mathfrak{e}_{d}$  of  ${\rm Ker}(D)$,  where $d_*:={\rm dim}\bigl( {\rm Ker}(D)^\perp\bigr)$ (when  $D$ is non degenerate,  one can simply take $\mathfrak{e}_k:=e_k$).
  \end{Definition}

%
%

\begin{Definition}\label{vettorino} We introduce the following three functional spaces:
\begin{itemize}
\item We define $H^1(\L, d_*)$ as the  Hilbert space given by  functions $f\in L^2(\L)$ with weak derivative 
$\partial_{\mathfrak{e}_i} f$ in $L^2(\L)$  for any $i=1,\dots, d_*$,  endowed    with the 
squared norm  $\|f\| ^2_{1,*}:= \|f\| ^2_{L^2( \L)} +\sum _{i=1}^{d_*}   \| \partial_{\mathfrak{e}_i} f \| ^2_{L^2( \L)}$. Moreover, given $f\in H^1(\L,d_*)$, we define
 \be\label{aria7}
 \nabla_* f:= \sum_{i=1}^{d_*}  \left( \partial _{\mathfrak{e}_i} f \right) \mathfrak{e}_i\,.
 \en
\item We define $H^1_0( \L,F,d_*)$ as the closure in  $H^1(\L,d_*)$ of
\[ 
\Big \{ \varphi _{|\L} \,:\, \varphi \in C^\infty_c (\bbR^d \setminus F) \Big\}\,.\]
 \item  We define the functional set $K$ as (cf. 
  \eqref{def_psi})
  \be\label{kafka}
K:= \{\z_{|\L}+ f\,:\, f \in H^1_0( \L,F,d_*)\}\,.
\en
\end{itemize}
  \end{Definition}
We stress  that, if $\mathfrak{e}_i = e_i$ for all $i=1,\dots, d_*$ (as in many applications), then 
\be
 \nabla_* f:= (\partial_1 f, \dots, \partial _{d_*} f , 0,\dots, 0)\in L^2(\L)^d\,.
 \en 

The definition of \rosso{the space} $H^1(\L,d_*)$ is indeed intrinsic, i.e.~it does not depend on the particular choice of  $\mathfrak{e}_1$,$ \mathfrak{e}_2$,...,$ \mathfrak{e}_{d_*}$ as orthonormal basis of ${\rm Ker}(D)^\perp$. \rosso{Indeed,}   by linearity  $f\in H^1(\L,d_*)$ if and only if $f\in L^2(\L)$ and $\partial _v f \in L^2(\L)$ for any \rosso{vector} $v\in  {\rm Ker}(D)^\perp$.  \rosso{Moreover, since  for $v\in  {\rm Ker}(D)^\perp$ and $f\in H^1(\L,d_*)$ we have  $\partial_v f= \sum _{i=1}^{d_*} (v\cdot \mathfrak{e}_i) \partial_{\mathfrak{e}_i} f= v\cdot \nabla_* f$, both 
$\nabla_* f$ and $\|f\| ^2_{1,*}:=\|f\|^2_{L^2(\L)}+\int _\L  |\nabla_* f|^2 dx$
are  basis-independent. When} $f$ is smooth,  $\nabla_* f$ is simply the orthogonal projection of $\nabla f $ on ${\rm Ker}(D)^\perp$.

  \begin{Remark}\label{memorandum} Suppose that  $\mathfrak{e}_i = e_i$ for all $1\leq i\leq  d_*$. Let $f\in H^1(\L,d_*)$. Given $1\leq i \leq d_*$,  by integrating $\partial_i f $ times $\varphi (x_1, \dots, x_{d_*}) \phi (x_{\rosso{d_*}+1}, \dots, x_d)$ with 
  $\varphi \in C_c^\infty (\bbR^{d_*})$ (varying in a suitable countable subset) and $\phi \in C_c^\infty( \bbR^{d-d_*})$, one gets  that  the function $f( \cdot, y_1, \dots, y_{\rosso{d}-d_*})$ belongs to $H^1\left( (-1/2,1/2)^{d_*}\right)$ for a.e. $(y_1, \dots, y_{\rosso{d}-d_*}) \in (-1/2,1/2)^{\rosso{d}-d_*}$.
  \end{Remark}
   Being a closed subspace of the Hilbert space
 $H^1(\L,d_*)$,  $H^1_0( \L ,F,d_*)$ is a Hilbert space.
We also point out that in the definition of $K$ one could replace $\z_{|\L}$ by any other function $\phi\in H^1(\L,d_*) \cap C(  \bar\L)$  such that $\phi\equiv 0$ on $F_-$ and $\phi\equiv 1$ on $F_+$ 
 as follows from the next lemma: 
\begin{Lemma}\label{12anni} Let $u \in H^1(\L,d_*) \cap C(  \bar \L )$   satisfy $u\equiv 0$ on $F$. Then $u\in H^1_0( \L ,F,d_*)$.
\end{Lemma}
\begin{proof} For simplicity of notation we suppose that   $\mathfrak{e}_i = e_i$ for all $1\leq i\leq  d_*$ (in the general case, replace $\partial_i $ by $\partial_{\mathfrak{e}_i}$ below). We use some idea from the proof of \cite[Theorem~9.17]{Br}.
We set  $u_n(x):= G( n u(x) )/n$, where $G\in C^1(\bbR)$ satisfies: $|G(t) | \leq |t| $ for all $t\geq 0$, $G(t)=0$ for $|t|\leq 1$ and $G(t)=t $ for $|t|\geq 2$.  Note that   $\partial _i u_n(x) = G' (n u(x) ) \partial _i u(x) $ for $1\leq i \leq d_*$ (cf. \cite[Prop.~9.5]{Br}). Hence, $u_n\to u$ a.e. and $\partial _i u_n \to   \partial _i u$ a.e. In the last identity, we have used that $\partial _i u=0$ a.e. on $\{u=0\}$ which follows as a byproduct of  Remark \ref{memorandum} and  Stampacchia's theorem (see Thereom~3 and Remark~(ii) to Theorem 4 in \cite[Section~6.1.3]{EG}).  By dominated convergence one obtains that $u_n \to u$ in $H^1( \L, d_*)$. Since $H^1_0( \L, F,d_*)$ is a closed subspace of $H^1(\L,d_*)$, it is enough to prove that $u_n\in H^1_0( \L, F,d_*)$. Due to our hypothesis on $u$ and the definition of $G$, $u_n\equiv 0$ in a neighborhood of $F$ inside $ \bar \L$. Hence the conclusion follows by applying the implication (iii)$\Rightarrow$ (i) in Proposition \ref{monti} below. Equivalently, it is enough to observe that, by adapting  \cite[Cor.~9.8]{Br} or \cite[Thm.~1, Sec.~4.4]{EG}, there exists a sequence of functions   $\varphi_k \in C_c^\infty (\bbR^d)$ such that ${\varphi_k }_{|\L} \to u_n$ in $H^1(\L, d_*)$. Since $u_n\equiv 0$ in a neighborhood of $F$, it is easy to  find   $\phi \in C_c^\infty (\bbR^d\setminus F)$   such that  ${(\phi \varphi_k)} _{|\L} \to u_n$ in $H^1(\L, d_*)$. Hence  $u_n\in H^1_0( \L, F,d_*)$. 
\end{proof}

One can  adapt the proof of  \cite[Prop.~9.18]{Br} to get the following criterion assuring that a function belongs to $H^1_0(\L,F,d_*)$:
\begin{Proposition}\label{monti} 
Given  a function $u \in L^2(\L)$, the following properties are equivalent:
\begin{itemize}
\item[(i)] $ u \in H^1_0(\L, F,d_*)$;
\item[(ii)]
 there exists $C>0$  such that 
\be\label{tennis73}
\Big|
 \int _{\L} u\,\partial _{\mathfrak{e}_i} \varphi dx
 \Big| 
 \leq C\|\varphi \|_{L^2(\L)} \quad \forall \varphi \in C_c^\infty (S) \,,\;\forall i:1 \leq i \leq  d_*\,;
 \en 
 \item[(iii)] the function 
 \begin{equation}
 \bar u (x) := 
 \begin{cases}
 u(x) & \text{ if } x \in \L\,,\\
 0 & \text{ if }  x \in S \setminus \L\,,
 \end{cases}
 \end{equation}
 belongs to $H^1(S,d_*)$ (which is defined similarly  to $H^1(\L, d_*)$ by replacing $\L$ with $S$). Moreover, in this case it holds $\partial _{\mathfrak{e}_i}
 \bar u = \overline{ \partial _{\mathfrak{e}_i} u}$ for $1\leq i \leq d_*$, where $ \overline{\partial _{\mathfrak{e}_i} u}$ is defined by extending $\partial _{\mathfrak{e}_i} u$ as zero on $S\setminus\L$.
 \end{itemize}
 \end{Proposition}

\begin{Lemma}[Poincar\'e inequality] \label{poincare} \rosso{If $e_1 \in {\rm Ker}(D)^\perp$, then}
 $\| f \|_{L^2(\L )}\leq \| \partial_1 f \|_{L^2(\L )}$ for any $f\in H^1_0 (\L, F,d_*)$.
\end{Lemma}
\begin{proof}
Given  $f \in \cC^\infty _c (\bbR^d\setminus F) $, by Schwarz inequality, for any $(x_1,x') \in \L$  with $x_1\in (-\frac{1}{2},\frac{1}{2})$ we have 
$ f (x_1,x')^2=\bigl(
 \int_{-1/2} ^{x_1} \partial_1 f (s, x')ds\bigr)^2 \leq \int_{-1/2} ^{1/2} \partial_{1} f (s, x')^2 ds$. By integrating over $\L$ we get the desired estimate for   $f \in \cC^\infty _c (\bbR^d\setminus F) $.   Since $\cC^\infty _c (\bbR^d\setminus F )$ is dense in $H^1_0(\L, F,d_*)$ \rosso{and $e_1 \in {\rm Ker}(D)^\perp$}, we get the thesis.
\end{proof}

\begin{Definition}\label{fete}
We say that  $v $ is a weak solution of the equation
\be\label{salvateci} 
\nabla_* \cdot ( D \nabla_* v ) =0
\en
on $\L$ with boundary conditions 
\be\label{mbc}
\begin{cases}
v(x) =0 & \text{ if } \;\;x\in F_-\,,\\
v(x)= 1 & \text{ if }\;\;  x \in F_+\,,\\
 D \nabla_* v (x) \cdot  \mathbf{n}(x) =0 & \text{ if } \;\;x \in \partial \L \setminus F\,,
\end{cases}
\en
if $v \in K$ (cf. \eqref{kafka})  and if $\int_{\L} \nabla_* u \cdot D \nabla_* v\, dx =0$  for all  $u\in  H^1_0(\L, F,d_*)$.
\end{Definition}
Above $\mathbf{n}$ denotes the  outward  unit normal vector   to  the boundary in $\partial \L$.
%
\begin{Remark}\label{piccoletto}
In the above definition it would be enough to 
require that $\int _\L \nabla_* u \cdot D \nabla_*  v dx =0$  for all   $u\in C^\infty_c(\bbR^d\setminus F) $ since the functional $H^1_0(\L , F,d_*)\ni u   \mapsto \int _\L\nabla_* u \cdot D \nabla_* v dx \in \bbR $ is continuous.
\end{Remark}
We shortly motivate the above Definition \ref{fete}. Just to simplify the notation we suppose that  $\mathfrak{e}_i=e_i$ for $1\leq i \leq d_*$.
 By  Green's formula  we have 
 \be\label{lunare}
\int _\L  ( \partial_i f) g\, dx=-\int_\L  f (\partial_i g)\, dx  + \int _{\partial \L} fg  ( \mathbf{n}\cdot e_i ) d S\,, \qquad \forall f,g \in C^1(\bar{\L})\,,
\en
where   $dS$ is the surface measure on $\partial \L$.
 By taking $1\leq i , j \leq d_*$, $f= \partial _{j} v$ and $g=u$ in \eqref{lunare},  we get 
\be \label{solare}
\int_\L u \nabla_* \cdot ( D \nabla_* v ) dx  =-
\int_\L \nabla_* u  \cdot ( D \nabla_* v) dx + \int _{\partial \L} u (\nabla_* v \cdot (D\mathbf{n}) ) d S \,,
\en
for all $v \in C^2(\bar{\L})$ and $ u \in C^1(\bar{\L})$.  Hence,  $v\in C^2(\bar \L)$ satisfies 
$ \nabla_* \cdot ( D \nabla_* v )=0$ on $\L$ and  $\nabla_* v \cdot (D\mathbf{n})\equiv 0$ on $\partial \L \setminus F$ if and only if $\int_{\L} \nabla _* u  \cdot ( D \nabla_* v) dx=0 $ for any $u \in C^1(\bar \L)$ with $u\equiv 0 $ on $F$.
Such a set $\cC$ of  functions $u$ is  dense in $H^1_0(\L,F,d)$. Indeed $\cC\subset H^1_0(\L,F,d)$ by Lemma \ref{12anni}, while  $C^\infty_c(\bbR^d\setminus F)\subset \cC$. Hence, we conclude that $v\in C^2(\bar \L )$ satisfies 
$ \nabla_* \cdot ( D \nabla_* v )=0$  on $\L$ and  $\nabla_* v \cdot (D\mathbf{n})\equiv 0$ on $\partial \L \setminus F$ if and only if $\int_{\L} \nabla_* u  \cdot ( D \nabla_* v) dx=0 $ for any  $u\in  H^1_0(\L,F,d)$.  We have therefore proved that    $v\in C^2(\bar \L)$  is a classical solution of \eqref{salvateci} and \eqref{mbc} if and only if it is a weak solution in the sense of Definition \ref{fete}.

%
%
%
%
%
%
%
%
%
%
%
%
%


\begin{Lemma}\label{unico}  Suppose that $e_1 \in {\rm Ker}(D)^\perp$. Then 
there exists a unique weak solution $v\in K$ of the equation $\nabla_* \cdot ( D \nabla_* v ) =0$ with boundary conditions  \eqref{mbc}. Furthermore, $v$ is the unique minimizer of 
$\inf _{h \in K}   \int \nabla_* h \cdot D \nabla_* h \,dx$.
\end{Lemma}
\begin{proof} To simplify the notation, in what follows we write $\z$ instead of $\z_{|\L}$.
We define the bilinear form $a(f,g):= \int _{\L} \nabla_* f \cdot D \nabla_* g dx $ on  the Hilbert space $H^1_0(\L,F,d_*)$. The bilinear form $a(\cdot,\cdot)$ is symmetric and  continuous.  Due to the  Poincar\'e inequality (cf.~Lemma~\ref{poincare}) and since  $e_1\in {\rm Ker}(D)^\perp$, 
$a(\cdot,\cdot)$ is also coercive. \rosso{Indeed,}  $\rosso{\|f\|_{L^2(\L)}^2\leq}\int _{\rosso{\L}} |\partial _1 f|^2 dx =\int_{\rosso{\L}}|\nabla_* f \cdot e_1|  dx \leq \int _{\rosso{\L}} |\nabla_* f|^2 dx\leq  a(f,f)/C$, where $C$ is the minimal positive eigenvalue of $D$.


By definition 
we have that $v\in K$ is a weak solution of  equation $\nabla_* \cdot ( D \nabla_* v ) =0$ with b.c. \eqref{mbc} if and only if, setting $f:=v-\z$, $f \in H^1_0(\L, F,d_*)$ and $f$ satisfies 
\be\label{baia}
\int _{\rosso{\L}}\nabla_* f\cdot D \nabla_* u dx = - \int _{\rosso{\L}}\nabla_* \z  \cdot D \nabla_* u\, dx \qquad \forall u \in H^1_0(\L, F,d_*)\,. 
\en
Note that the r.h.s. is a continuous functional in $u\in H^1_0(\L, F,d_*)$.
Due to the above observations and by  Lax--Milgram theorem,  we conclude that there exists a unique such function  $f$, hence there is a unique weak solution $v$ of  equation $\nabla_* \cdot ( D \nabla_* v ) =0$ with b.c. \eqref{mbc}. Moreover $f$ satisfies  
\[
\frac{1}{2} a(f,f)+ \int _{\rosso{\L}}\nabla_* \z \cdot D \nabla_* f \, dx= \inf _{g\in H^1_0(\L, F,d_*)} \Big\{\frac{1}{2} a(g,g)+ \int _{\rosso{\L}}\nabla_* \z  \cdot D \nabla_* g\,dx
\Big\}\,.
\]
By adding to both sides $\frac{1}{2} \int_{\rosso{\L}} \nabla_* \z \cdot D \nabla_* \z  \, dx$,  we get that  $\frac{1}{2}\int _{\rosso{\L}}\nabla_* v \cdot D\nabla_* v= \inf _{h\in K} \frac{1}{2}\int _{\rosso{\L}}\nabla_* h \cdot D\nabla_* h \,dx$.
\end{proof}
\begin{Remark}
\rosso{If  $e_1\not \in {\rm Ker}(D)^\perp$, then  the conclusion  of Lemma \ref{unico} can fail.}  Take for example  $d=2$, $d_*=1$.  Then, if $|\mathfrak{e}_1\cdot e_1 |\ll 1$, 
functions of the form 
  $f(x):= \varphi (x \cdot \mathfrak{e}_2)$, with $\varphi \in C^1 (\bbR)$, $\varphi(u)\equiv 0 $ for $u<-\e$, $\varphi(u)\equiv 1 $ for $u>\e$ and $\e>0$ small, satisfy $f\in K$ and $\nabla_* f\equiv 0$ (and therefore  $\nabla_* \cdot ( D \nabla_* f ) =0$). 
\end{Remark}
\rosso{Lemma  \ref{unico} justifies the following definition:}
\begin{Definition}\label{rondini25} \rosso{When $e_1 \in {\rm Ker}(D)^\perp$, we denote by $\psi$ 
the unique weak solution $v\in K$ of the equation $\nabla_* \cdot ( D \nabla_* v ) =0$ with boundary conditions  \eqref{mbc}. }
\end{Definition}
\rosso{As $\nabla_*\z_{|\L} (x)=e_1$ and $De_1$ is proportional to $e_1$ if $e_1$ is a $D$--eigenvector,  from the above lemma we immediately get:}
\begin{Corollary}\label{melone28}\label{h2o} \rosso{If $e_1$ is an eigenvector of $D$ with eigenvalue $D_{1,1}>0$, then  $\z_{|\L}=\psi$.}
\end{Corollary}

%


\section{Square integrable forms and effective homogenized matrix}\label{sec_review}

As common in homogenization theory, the variational formula \eqref{def_D} defining the effective homogenized matrix $D$  admits 
a geometrical interpretation  in the Hilbert space of square integrable forms.
For later use we  recall here this interpretation. We also   collect  some  facts taken from \cite{Fhom3}. 

\smallskip
We define $\nu$ as the \rosso{finite} measure on $\O \times \bbR^d$ such that 
\begin{equation}\label{labirinto}
 \int d \nu (\o, z) g (\o, z) = \int  d \cP_0 (\o)\int d  \hat \o (z) c_{0,z}(\o) g( \o, z) 
 \end{equation} 
 for any nonnegative measurable function $g(\o,z)$. We point out that $\nu$ has finite total mass since 
$ \nu(\O\times \bbR^d )=\bbE_0[\l_0]<+\infty$.
  Elements of  $L^2 ( \nu)$
 are called \emph{square integrable forms}.

 \smallskip
 
 Given a  function $u:\O_0\to \bbR$,  its gradient  $\nabla u: \O \times \bbR^d \to \bbR$ is defined  as 
 \begin{equation}\label{cantone}
 \nabla u (\o, z):= u (\t_z \o)-u (\o)\,.
 \end{equation}
 If $u$ is defined  $\cP_0$--a.s., then $\nabla u$ is well defined $\nu$--a.s. by Lemma \ref{matteo}.
 If $u $ is bounded and measurable, then $\nabla u \in L^2(\nu)$.
The subspace of \emph{potential forms} $L^2_{\rm pot} (\nu)$ is defined as   the following closure in $L^2(\nu)$:
 \[ L^2_{\rm pot} (\nu) :=\overline{ \{ \nabla u\,:\, u \text{ is  bounded and measurable} \}}\,.
 \]
 The subspace of \emph{solenoidal forms} $L^2_{\rm sol} (\nu)$ is defined as the orthogonal complement of $L^2_{\rm pot} (\nu)$ in $L^2(\nu)$.

\begin{Definition}\label{def_div}
Given a square integrable form $v\in L^2(\nu)$  its divergence  ${\rm div} \, v \in L^1(\cP_0)$ is defined as 
${\rm div}\, v(\o)= \int d \hat{\o} (z) c_{0,z}(\o) ( v(\o,z)-  v(\t_z \o, -z) )
$.
\end{Definition}
The r.h.s. in the above identity is well defined since it corresponds  $\cP_0$--a.s. to an  absolutely convergent series (cf.~\cite[Section~8]{Fhom3} where (A3) is used).
 For any  $v \in L^2(\nu)$ and any  bounded and measurable function $u:\O_0 \to \bbR$, it holds (cf. \cite[Section~8]{Fhom3})
$\int  d \cP_0(\o)   {\rm div} \,v(\o)  u (\o)= - \int d \nu(\o, z) v( \o, z) \nabla u (\o, z)$.
As a consequence we have that, given    $v\in L^2(\nu)$, $v\in L^2_{\rm sol}(\nu)$ if and only if ${\rm div}\, v=0$ $\cP_0$--a.s.
By using (A6) one also gets  (cf. \cite[Section~8]{Fhom3}):
\begin{Lemma}\label{santanna} 
 The  functions $g\in L^2(\cP_0)$ of the form $g ={\rm div}\, v $ with $v\in L^2(\nu)$ are dense in $\{w \in L^2(\cP_0)\,:\, \bbE_0[w]=0\}$.
\end{Lemma}

 As $\l_2\in L^1(\cP_0)$, 
given $a\in \bbR^d$   the form 
$u_a(\o,z):= a\cdot z$
 belongs to $L^2(\nu)$. 
 We note  that the effective homogenized  matrix $D$  defined in \eqref{def_D} satisfies, for any $a\in \bbR^d$,
 \begin{equation}\label{giallo}
 \begin{split}
q(a):= a \cdot Da  
&=   \inf_{ v\in L^2 _{\rm pot}(\nu) } \frac{1}{2}\| u_a+v \|^2_{L^2(\nu)}=\frac{1}{2} \| u_a+v ^a \|^2_{L^2(\nu)}\,,
\end{split}
 \end{equation}
 where  $v^a=-\Pi u_a$ and 
 $\Pi: L^2(\nu) \to L^2_{\rm pot}(\nu)$ denotes the orthogonal projection of $L^2(\nu)$ on $L^2_{\rm pot}(\nu)$.     Hence   $v^a$ is characterized by the properties
\begin{equation}\label{jung}
v^a \in  L^2_{\rm pot}(\nu)\,, \qquad v^a+u_a\in L^2_{\rm sol}(\nu)\,.
\end{equation}
Moreover it holds (cf. \cite[Section 9]{Fhom3}):
 \begin{equation}\label{solare789}
 Da =\frac{1}{2}  \int d\nu (\o, z)  z \bigl(a\cdot z + v^a( \o, z) \bigr)\qquad \forall a \in \bbR^d\,.
\end{equation}

 By \eqref{giallo}
 the   kernel  ${\rm Ker}(q)$ of the quadratic form $q$  is given by 
$ {\rm Ker}(q):=\{a\in \bbR^d\,:\, q(a)=0\}=\{ a\in \bbR^d\,:\, u_a \in L^2_{\rm pot}(\nu)\}$. Note that ${\rm Ker}(q)={\rm Ker}(D)$.
The following result is the analogous of \cite[Lemma 5.1]{ZP} (cf. \cite[Section~9]{Fhom3}):
\begin{Lemma}\label{rock}
$
{\rm Ker}(q)^\perp= \big\{  \int  d \nu (\o,z) b(\o, z) z  \,:\, b\in L^2_{\rm sol} (\nu) \big\}$.
\end{Lemma}
It is simple to check that  Definition \ref{divido}   and Lemma \ref{rock} imply the following:
 \begin{Corollary}\label{rock_bis} 
 ${\rm Span}\big \{ \mathfrak{e}_1, \mathfrak{e}_2, \dots, \mathfrak{e}_{d_*}\big\}= \big\{  \int  d \nu (\o,z) b(\o, z) z  \,:\, b\in L^2_{\rm sol} (\nu) \big\}$.
  \end{Corollary}


\begin{Definition}\label{artic}  Let $b(\o,z): \O _0 \times \bbR^d\to \bbR$ be a measurable function  with $\|b\|_{L^1(\nu)}<+\infty$. We define the measurable function  $c_b:\O_0 \to [0,+\infty]$ as
\begin{equation}\label{magno}
c_b (\o):=  \int d \hat{\o}(z) c_{0,z}(\o)  |b(\o,z)|\,, \end{equation}
 the measurable function $\hat b: \O_0 \to \bbR$ as
\begin{equation}\label{zuppa}
\hat b(\o):= 
\begin{cases}
\int d \hat{\o}(z) c_{0,z}(\o) b(\o,z) & \text{ if } c_b (\o) <+\infty\,,\\
0 & \text{ if } c_b (\o) = +\infty\,,
\end{cases}
\end{equation}
and the  measurable set $\cA_1[b]:= \{ \o\in \O\,:\, c_{b } ( \t _z \o) <+\infty\; \forall z\in \hat \o\}$.
\end{Definition}
 $ \cA_1[b]$ is a measurable translation invariant set and   $\cP( \cA_1[b])=\cP_0( \cA_1[b])=1$ if $\|b\|_{L^1(\nu)}<+\infty$  (cf.~\cite[Section~10]{Fhom3}). Trivially $\| \hat b \|_{L^1(\nu)}\leq \| b\|_{L^1(\nu)}$.

%
%
%
%
%

\begin{Definition}\label{ometto}
Given  a measurable  function  $b :\O_0\times \bbR^d\to \bbR$ we define 
 $\tilde b :\O_0\times \bbR^d\to \bbR$
as
\begin{equation}
\tilde b (\o, z):=
\begin{cases}  b (\t_{z} \o, -z)  & \text{ if } z\in \hat \o\,,\\
0 & \text{ otherwise}\,.
\end{cases}
\end{equation}
\end{Definition}
%
%

\begin{Definition}\label{naviglio} Let $b: \O_0\times \bbR^d \to \bbR$ be a  measurable function with $\| b\|_{L^1(\nu)}<+\infty$.
 If $\o \in   \cA_1[b]\cap \cA_1[\tilde b]\cap \O_0$,  we set ${\rm div}_* b  (\o) := \hat b (\o) - \hat{\tilde{b}}(\o) \in \bbR$.
\end{Definition}
One can check (cf. \cite[Section~11]{Fhom3}) that 
given a measurable function 
 $b: \O_0\times \bbR^d \to \bbR$  with $\|b\|_{ L^2(\nu)}<+\infty$, then
 $\|b\|_{ L^2(\nu)}=\|\tilde b\|_{ L^2(\nu)}$, $\cP_0(  \cA_1[b]\cap \cA_1[\tilde b])=1$, 
${\rm div}_* b  = {\rm div }\,b=- {\rm div }\,\tilde b $ in $L^1(\cP_0)$. In particular, for $\|b\|_{ L^2(\nu)}<+\infty$, $b\in L^2_{\rm sol}(\nu)$ if and only if $\tilde b\in L^2_{\rm sol}(\nu)$.

\begin{Definition}\label{tav}
Let  $b: \O_0\times \bbR^d \to \bbR$ be a measurable function with  $\|b\|_{ L^2(\nu)}<+\infty$ and such that  its class of equivalence in $L^2(\nu)$ belongs to $L^2_{\rm sol}(\nu)$. Then we set 
\be\label{eq_tav}
\cA_d [b]:= \{\o \in \cA_1[b] \cap \cA_1[ \tilde b]\,:\, {\rm div}_* b (\t_z \o) =0 \; \forall z \in \hat \o\}\,.
\en
\end{Definition}

 $\cA_d [b]$ is a translation invariant measurable set and  $\cP( \cA_d [b])=1$ (cf. \cite[Section~11]{Fhom3}).

\smallskip

Recall the set $\cA[g]$ introduced in Proposition \ref{prop_ergodico} and Definition \ref{sibilla}.
The following lemma corresponds to \cite[Lemma 19.2]{Fhom3}. Here we have isolated the properties on $\o$ used in the proof there. \rosso{We set  $\nu ^\e_{\o } : = \e^d \sum _{x\in \e\hat \o    }  \sum _{y\in  \e\hat \o   }   c_{\frac{x}{\e},\frac{y}{\e} }(\o) \d_{( x ,\frac{y-x}{\e})}$.}

\begin{Lemma}\label{blocco}  Suppose that $\o$ belongs to the sets\\
 $\cA_1[1]$, $ \cA[\l_0]$, $ \cA_1[ |z|^2 ]$ and $ \cA[  \int d \hat \o (z) c_{0,z}(\o) |z|^2  \mathds{1}( |z| \geq  \ell) ]$ for all $\ell \in \bbN$.
 Then $\forall \varphi \in C_c^2(\bbR^d)$ we have
 \begin{equation}\label{football}
\lim  _{\e \da 0}  \int d \nu^\e_{\o}(x,z)  \bigl[\nabla_\e \varphi (x,z) - \nabla  \varphi (x) \cdot z \bigr]^2   =0\,.
\end{equation}
\end{Lemma}
\section{The set $\O_{\rm typ}$ of typical environments \rosso{when $e_1\in{\rm Ker}(D)^\perp$}}\label{sec_tipetto}
In this section we describe the set $\O_{\rm typ}$ of typical environments for which Theorem \ref{teo1} and Proposition \ref{teo2} hold  when \rosso{when $e_1\in{\rm Ker}(D)^\perp$ (strictly speaking, the set $\O_{\rm typ}$ there  is the intersection of  the set of typical environments  described in Section \ref{renato_zero} and the set $\O_{\rm typ}$ presented in this section)}. To define $\O_{\rm typ}$ and also  2-scale convergence in the next section, we need to isolate suitable countable dense sets of $L^2(\cP_0)$ and $L^2(\nu)$.

 Recall that $L^2(\cP_0)$ is separable due to (A8). Then one can prove that also  $L^2(\nu)$ is separable (cf.~\cite[Section~12]{Fhom3}). The separability of $L^2(\cP_0)$ and $L^2(\nu)$ will be used below to construct suitable countable dense subsets.   The functional sets $\cG_1,\cH_1,\cG_2,\cH_2,\cW$ presented below are as in \cite[Section~12]{Fhom3}
(we recall their definition for the reader's convenience).

\smallskip
\noindent
$\bullet$ {\bf The functional sets $\cG_1,\cH_1$}. 
We fix a countable set $\cH_1$ of measurable functions $b: \O_0\times \bbR^d\to \bbR$ such that 
$\|b \|_{L^2(\nu)}<+\infty$ for any $b \in \cH_1$ and such that   $\{ {\rm div} \,b\,:\, b \in \cH_1\}$ is a dense subset of $\{ w \in L^2(\cP_0)\,:\, \bbE_0[ w]=0\}$ when thought of as set of $L^2$--functions (recall Lemma \ref{santanna}). 
For each $b \in \cH_1$ we define  the measurable function $g_b : \O_0 \to \bbR$ as  (cf. Definition~\ref{naviglio})
\be \label{lupetto}
g_b (\o):= 
\begin{cases}
{\rm div}_* b (\o) & \text{ if } \o \in \cA_1[b]\cap \cA_1[\tilde b]\,,\\
0 & \text{ otherwise}\,.
\end{cases}
\en
Note that  $g_b={\rm div}\, b$, $\cP_0$--a.s. (see Section \ref{sec_review}).
We set $\cG_1:=\{ g_b \,:\, b \in \cH_1\}$. 

\smallskip

\noindent
$\bullet$ {\bf The functional sets $\cG_2,\cH_2 $}.  We fix a countable set $\cG_2$  of  bounded measurable functions 
$g: \O_0\to \bbR$ such that  the set $\{ \nabla g\,:\, g \in \cG_2\}$, thought in $L^2(\nu)$,  is  dense in $L^2_{\rm pot}(\nu)$.  We  define $\cH_2$ as the set of measurable functions $h:\O_0 \times \bbR^d\to \bbR$ such that $h=\nabla g$ for some $g\in \cG_2$.

\smallskip

\noindent
$\bullet$ {\bf The functional set $\cW$}. We fix a  countable set $\cW$ of measurable functions $b:\O_0\times \bbR^d\to \bbR$ such that, thought of as subset of $L^2(\nu)$, $\cW$  is dense in $L^2_{\rm sol} (\nu)$.  Since  $\tilde b \in L^2_{\rm sol}(\nu)$ for any $b \in L^2_{\rm sol}(\nu)$, at cost to enlarge $\cW$ we assume that $\tilde b \in \cW$ for any $b \in \cW$.

\smallskip

We introduce the following functions on $\O_0$ (note that they are $\cP_0$--integrable),  where $n,i,b$ vary respectively in $\bbN$, $\{1,\dots, d\}$, $\cW$:
\begin{align}
& f_n(\o):= \int d \hat \o (z) c_{0,z} (\o) |z|^2 \mathds{1}(|z|\geq n)\,,\label{hesse1}\\
& u _{  b,i }(\o):= \int d\hat \o (z)
c_{0,z}(\o) z_i  {b} (\o, z)\,,\label{hesse2}\\
& u _{ b,i,n}(\o):= | u_{ b ,i}(\o)| \mathds{1}\bigl(| u_{ b ,i}(\o) |\geq n \bigr)  \,. \label{hesse3}
\end{align}

\begin{Definition}[Functional  set $\cG$] 
We \rosso{take as $\cG$ a countable set of measurable functions on $\O_0$, which is dense in $L^2(\cP_0)$ and  contains the following subsets:}  $\{1\}$,  $\cG_1$, $\cG_2$ and $\{ u_{b,i} \mathds{1}( |u_{b,i}|\leq n)\}$ with  $ b \in \cW$, $i \in \{1,\dots, d\}$, $n\in \bbN $.
\end{Definition}


\begin{Definition}[Functional  set $\cH$] We \rosso{take as  $\cH$  a  countable set  of measurable functions on $\O_0\times \bbR^d$,   which is dense in $L^2(\nu)$ and contains the following subsets:}  $\cH_1$, $\cH_2$, $\cW$, $\{ (\o,z) \mapsto z_i \,:\, 1\leq i \leq d\}$,  $\{ b \mathds{1}(|b|\leq n)\,:\, b \in \cW, \; n \in \bbN\}$.
\end{Definition}

\rosso{The above sets $\cG$ and $\cH$ will enter  also in  Definitions \ref{priscilla} and \ref{sandalo65} of 2-scale weak convergence. We use them also to define the set $\O_{\rm typ}$ of typical environments:}

\begin{Definition}[Set $\O_{\rm typ}$ of typical environments] \label{amen} 
The set $\O_{\rm typ}\subset \O$  of typical environments is  the intersection of the following sets:
\begin{itemize}
\item  $\cA[g g'] $ for all $g,g'\in \cG$ (recall that $1\in \cG$); 
\item $\cA_1[b b'] \cap \cA[ \widehat{ b b'}]$ as $b,b'\in \cH$;
\item    $\O_2$ (cf.~ Lemma \ref{paletta});
\item  $\cA_1[ |z|^k]\cap \cA[\l_k]  $ for $k=0,2$;
\item $\cA[ f_n]$ for all $n \in \bbN$;
\item $\cA_1[b]\cap \cA_1[\tilde b]\cap \cA_1[ b^2] \cap \cA_1[ \tilde{b}^2] $ for all $b\in \cH$;
\item $
  \cA[ \widehat{\;b^2\;}]\cap  \cA[ \widehat{\;\tilde{b}^2\;}] \cap  \cA  [\widehat{\,| b|\,} ]\cap  \cA [\widehat{ \, |\tilde b|\,}]$ for all $b\in \cH$;
\item $\cA_1[ \tilde b(\o,z) z_i]$ for $1\leq i \leq d$ for all $b\in \cW$;
\item $\cA[ u _{ b, i, n}]$ for all  $b \in \cW$, $1\leq i \leq d$ and $n\in \bbN$;
\item $\cA_d [b]$ for all $b\in \cW$ (recall \eqref{eq_tav}).
\end{itemize}
\end{Definition}
As $\l_0, \l_2\in L^1(\cP_0)$,  due to Proposition \ref{prop_ergodico} and  the discussion in Section \ref{sec_review}, and since we are dealing with a countable family of constraints, 
$\O_{\rm typ}$ is a measurable subset  of $\O_2\subset \O_1$ with $\cP(\O_{\rm typ})=1$. Since moreover $\O_2$, $\cA[\cdot ]$, $\cA_1[\cdot]$ and $\cA_d[\cdot]$ are translation invariant sets as already pointed out, we conclude that $\O_{\rm typ}$ is also translation invariant.

\begin{Remark}
Above we have listed the properties characterizing $\O_{\rm typ}$ as they will be used along the proof (anyway, we will point out the specific property used in each step under consideration). \rosso{The list in Definition \ref{amen} should also help the interested reader in checking that the same set of typical environments works well for the conductivity along any direction, and not only along $e_1$. }
\end{Remark}

%


\section{Weak convergence and  2-scale convergence}\label{anatre12}
Recall   $\mu^\e_{\o,\L} $ and $\nu^\e_{\o,\L}$  given in \eqref{atomiche}.
We also set (recalling the definition of $\mu^\e_\o$ \rosso{and $\nu^\e_\o$} given before Proposition~\ref{prop_ergodico} \rosso{and Lemma~\ref{blocco}, respectively})
 \begin{align}
 &  \mu^\e_\o:= \e^d \sum _{x \in \e\hat\o } \d_x\,, \qquad\qquad \qquad   \nu ^\e_{\o } : = \e^d \sum _{x\in \e\hat \o    }  \sum _{y\in  \e\hat \o   }   c_{\frac{x}{\e},\frac{y}{\e} }(\o) \d_{( x ,\frac{y-x}{\e})}\,, \label{melanzane}\\
&  \mu^\e_{\o,S}:= \e^d \sum _{x \in \e\hat \o  \cap S} \d_x\,,\qquad \qquad  \nu ^\e_{\o,S} : = \e^d \sum _{x\in \e\hat \o  \cap S }  \sum _{y\in  \e\hat \o  \cap S } 
 c_{\frac{x}{\e},\frac{y}{\e} }(\o) \d_{( x ,\frac{y-x}{\e})}\,.\label{acqua}
  \end{align}
In this section   $\D$ equals $S$ or $\L$.

\begin{Definition}\label{debole_forte} Fix $\o\in \O$ and a 
 family of $\e$--parametrized  functions $v_\e \in L^2( \mu^\e_{\o,\D})$. 
   We say  that the family $\{v_\e\}$  \emph{converges weakly} to the function  $v\in L^2( \D, m dx)$, and write  $v_\e \rightharpoonup v$, if the family $\{v_\e\}$ is bounded (in the sense: $\varlimsup  _{\e \da 0} \| v_\e\| _{L^2(\mu^\e_{\o,\D}) } <+\infty$)
and 
$\lim _{\e \da 0} \int d  \mu^\e _{\o,\D}  (x)   v_\e (x) \varphi (x)= 
\int_\D dx\, m   v(x) \varphi(x)$
for all  $\varphi \in C_c(\D)$.
%
\end{Definition}


%
%


  
\begin{Definition}\label{priscilla}
Fix   $\tilde \o\in \O_{\rm typ}$, an $\e$--parametrized  family of functions  $v_\e \in L^2( \mu^\e_{\tilde \o,\D })$ and  a function $v \in L^2 \bigl(\D \times \O, m dx \times \cP_0\bigr)$.
 We say that \emph{$v_\e$ is weakly 2-scale convergent to $v$}, and write 
$v_\e \stackrel{2}{\rightharpoonup} v$, 
if the family $\{v_\e\}$ is bounded, i.e.
$
 \varlimsup_{\e\downarrow 0}  \|v_\e\|_{L^2(\mu^\e_{\tilde \o,\D})}<+\infty$, 
 and 
\begin{equation}\label{rabarbaro}
\lim _{\e\downarrow 0} \int d \mu _{\tilde \o, \D }^\e (x)  v_\e (x) \varphi (x) g ( \t _{x/\e} \tilde{\o} ) =\int d\cP_0(\o)\int _\D dx\,  m v(x, \o) \varphi (x) g (\o)  \,,
\end{equation}
for any $\varphi \in  C_c (\D)$ and any $g \in\cG$.
\end{Definition}
As  $\tilde \o \in \O_{\rm typ}\subset  \cA[g]$ for all $g\in \cG$, by Proposition \ref{prop_ergodico} one  gets  that $v_\e \stackrel{2}{\rightharpoonup} v$ where $v_\e:= \varphi \in L^2(\mu^\e_{\tilde \o,\D})$ and $v:=\varphi \in L^2(\D, m dx)$ for any $\varphi \in C_c(\D)$.
%
%
 \smallskip
 
One can  prove the following fact  by using the first item in Definition~\ref{amen} and by adapting the proof of  \cite[Lemma~13.5]{Fhom3}:

\begin{Lemma}\label{compatto1}
 Let $\tilde \o\in \O_{\rm typ}$.    Then, given a bounded family of functions $v_\e\in L^2 ( \mu^\e_{\tilde{\o},\D})$,  there exists a subsequence $\{v_{\e_k}\}$ such that
 $ v_{\e_k}
  \stackrel{2}{\rightharpoonup}   v $ for some 
 $  v \in L^2(\D \times \O, m dx \times \cP_0 )$ with $\|  v\|_{   L^2(\D\times \O, m dx \times \cP_0 )}\leq \varlimsup_{\e\da 0} \|v_\e\|_{L^2(\mu^\e_{\tilde{\o},\D})}$.
 \end{Lemma}


Recall the definition of the measure $\nu$ given in \eqref{labirinto}.
 
 \begin{Definition}\label{sandalo65}
Given  $\tilde \o\in \O_{\rm typ}$, an $\e$--parametrized  family of functions $w_\e \in L^2( \nu ^\e_{\tilde \o, \D })$ and  a function  $w \in L^2 \bigl( \D\times \O\times \bbR^d\,,m  dx \times d\nu\bigr)$, we say that \emph{$w_\e$ is weakly 2-scale convergent to $w$}, and write $  w_\e \stackrel{2}{\rightharpoonup}w   $,  if   $\{w_\e\}$ is bounded
 in $L^2( \nu ^\e_{\tilde \o,\D})$, i.e. $
 \varlimsup_{\e \da 0} \|  w_\e\|_{L^2( \nu ^\e_{\tilde \o,\D})}<+\infty$,
   and
\be\label{yelena}
 \lim _{\e\downarrow 0} \int   d \nu_{\tilde \o,\D}^\e (x,z) w_\e (x,z ) \varphi (x) b ( \t _{x/\e} \tilde{\o},z )
=\int_\D  dx \,m\int d \nu (\o, z) w(x, \o,z ) \varphi (x) b (\o,z )  \,,
\en
for any    $\varphi \in C_c(\D)$  and  any  $b \in \cH $. 
\end{Definition}

 One can  prove the following fact  by using the second item in Definition~\ref{amen} and by adapting the proof of  \cite[Lemma~13.7]{Fhom3}:

\begin{Lemma}\label{compatto2} Let $\tilde \o\in \O_{\rm typ}$.  Then, given a bounded family of functions $w_\e\in L^2 ( \nu^\e_{\tilde{\o},\D })$,  there exists a subsequence $\{w_{\e_k}\}$   such that  $w_{\e_k} \stackrel{2}{\rightharpoonup} w$ for some 
 $ w \in L^2(  \D\times \O\times \bbR^d\,,\, m dx \times \nu )$ with  $\|  w\|_{   L^2(\D\times \O\times \bbR ^d\,,\, m dx \times \nu )}\leq \varlimsup_{\e\da 0} \|w_\e\|_{L^2( \nu^\e_{\tilde{\o},\D})}$.
\end{Lemma}

\section{$2$-scale limits of uniformly bounded functions} 
\label{sec_bike}
We fix $\tilde \o \in \O_{\rm typ}$ and  {we assume that $d_*={\rm dim} ({\rm Ker}(D)^\perp)\geq 1$. The domain $\D$ below  can be $\L,S$. We  consider a family of functions $\{f_\e\}$  with $f_\e: \e \widehat{\tilde \o} \cap S  \to \bbR $
such that 
\begin{align}
& \varlimsup_{\e\da 0} \|  f_\e\|_\infty <+\infty\,,\label{re1}\\
&\varlimsup_{\e \da 0}   \| f_\e \|_{L^2(  \mu^\e_{\tilde \o, \D})} <+\infty \,,\label{re2} \\
& \varlimsup_{\e \da 0}   \|\nabla_\e  f_\e \|_{L^2(\nu^\e_{\tilde \o,\D})} < +\infty\,.\label{re3}
\end{align}

Due to Lemmas \ref{compatto1} and \ref{compatto2}, along a subsequence $\{\e_k\}$ we have 
\begin{align}
&L^2(  \mu^\e_{\tilde \o, \D})\ni  f_\e \stackrel{2}{\toup} v  \in L^2(\D \times \O,  m dx \times \cP_0 ) \,,\label{totani1}\\
& L^2(\nu^\e_{\tilde \o,\D})\ni \nabla_\e{f}_\e \stackrel{2}{\toup} w\in   L^2(\D\times \O\times \bbR^d  \,,\, m dx \times \nu ) \,, \label{totani2}
\end{align}
for suitable functions $v,w$. 
\begin{Warning}
 In this section (with exception of Lemma \ref{lunghissimo} and Claim \ref{mengoni}), when taking the limit $\e\da 0$, we understand that $\e$ varies along the subsequence  $\{\e_k\}$ satisfying \eqref{totani1} and \eqref{totani2}.  Moreover, we set 
 \be  \bar f_\e(x):=0 \text{ for } x\in \e \widehat{\tilde  \o} \setminus S \,.
 \en
\end{Warning}

The structural results presented below (cf. Propositions \ref{prova2} and \ref{oro})  correspond to  a general strategy in homogenization by 2-scale convergence (see Propositions 16.1 and 18.1 in \cite{Fhom3}, Lemmas 5.3 and  5.4 in \cite{ZP}, Theorems 4.1 and 4.2 in \cite{Z}).
 Condition \eqref{re1} would not be strictly necessary, but it allows important technical simplifications, and in particular it allows to avoid the cut-off procedures developed in  \cite[Sections~15,17]{Fhom3} in order to deal with the long jumps in the Markov generator  \eqref{degregori}. Differently from \cite{Fhom3} here we have to control also several boundary contributions.
 We will apply Propositions \ref{prova2} and \ref{oro} only to the following  cases: $\D=\L$ and $f_\e={ V_\e}$;   $\D=S$ and $f_\e=V_\e-\z$ \rosso{(recall \eqref{def_psi})}. 

\smallskip
 In what follows we will use the following control on long filaments:
 
 \begin{Lemma}\label{lunghissimo}   Given $\tilde \o \in \O_{\rm typ}$,  $\ell>0$ and $\varphi \in C_c(\bbR^d)$,  it holds
 \be\label{lionello}
\lim_{\e\da 0}  \e^{-2} \int d\nu ^\e_{\tilde \o} (x,z) \,|\varphi(x)| \, \mathds{1}(|z|\geq \ell/\e)=0\,.
\en
\end{Lemma}
\begin{proof}
Recall the definition  \eqref{hesse1} of $f_n$. Given $n\in \bbN$ we take $\e$ small so that $\ell/\e>n$. Then  we can bound 
 \begin{equation*}\label{caffe}
\begin{split}
  \e^{-2} \int d\nu ^\e_{\tilde \o} (x,z) \,|\varphi(x)| \, \mathds{1}(|z| \geq \ell/\e)& \leq 
 \ell^{-2} \int d\nu ^\e_{\tilde \o} (x,z) |\varphi(x)|  \,  \mathds{1}(|z| \geq \ell/\e)    |z| ^{2}\\
 &\leq \ell^{-2} \int d\nu ^\e_{\tilde \o} (x,z) |\varphi(x)|  \,  \mathds{1}(|z| \geq n)   |z| ^{2}\\
 &= \ell^{-2}
   \int d\mu ^\e_{\tilde \o} (x) 
 |\varphi(x)| f_n(\t_{x/\e} \o)\,.
   \end{split}
   \end{equation*}
 As $\o \in \O_{\rm typ}\subset \cA[f_n]\cap \cA_1[|z|^2]$,
   we have (cf. (A7))
   \[
   \lim_{\e\da 0} \int d\mu ^\e_{\tilde \o} (x) 
 |\varphi(x)|  f_n  (\t_{x/\e} \o)=\int dx m |\varphi (x)| \bbE_0[f_n]<+\infty\,.
 \] 
 The last expression goes to zero as $n\to +\infty$ by dominated convergence. 
  \end{proof}

\begin{Proposition}\label{prova2}
 For $dx$--a.e. $x\in \D$, the map  $ v(x,\cdot)$ given in \eqref{totani1}  is constant $\cP_0$--a.s.
\end{Proposition}
\begin{proof} 
Recall the definition of the functional sets $\cG_1, \cH_1$ given in Section \ref{sec_tipetto}.
We  claim that $\forall \varphi \in C^1_c(\D)$ and $\forall \rosso{\xi} \in \cG_1$ it holds 
\begin{equation}\label{chiavetta}
\int _\D dx \,m \int d\cP_0 (\o) v (x,\o) \varphi (x) \rosso{\xi}(\o)=0\,.
\end{equation}
Having \eqref{chiavetta} it is standard to conclude. Indeed,
  \eqref{chiavetta} implies that,  $dx$--a.e. on $\D$,  $ \int d \cP_0 (\o) v (x,\o)\rosso{\xi}(\o)
=0$ for any $\rosso{\xi} \in \cG_1$. We conclude that,  $dx$--a.e. on $\D$, $v(x,\cdot)$ is orthogonal in $L^2(\cP_0 )$ to $\{ w \in L^2(\cP_0 )\,:\, \bbE_0[ w]=0\}$. Hence  $v(x,\o)= \bbE_0[ v(x, \cdot)]$ for $\cP_0$--a.a. $\o$.

It now remains to prove \eqref{chiavetta}. 
 We first note  that, by   \eqref{rabarbaro},  \eqref{totani1} and since  $\tilde \o \in \O_{\rm typ}$ and $\rosso{\xi} \in \cG_1\subset \cG$,
\begin{equation*}
\text{l.h.s. of }\eqref{chiavetta}= \lim_{\e\da 0} \int d \mu^\e_{\tilde \o,\D} (x) f_\e(x) \varphi (x) \rosso{\xi}( \t_{x/\e}\tilde \o) \,.
\end{equation*}
Let us take $\rosso{\xi}= g_b$  with $b \in \cH_1$ as in \eqref{lupetto}. 
 By \cite[Lemma~11.7]{Fhom3} and since $\tilde \o \in \O_{\rm typ}\subset \cA_1[b]\cap \cA_1[\tilde b]$, 
we have
\begin{equation*}
\begin{split}
\int d \mu^\e_{\tilde \o,\D} (x) f_\e(x) \varphi (x) \rosso{\xi}( \t_{x/\e}\tilde \o) & =\int d\mu^\e_{\tilde \o} (x) \bar f_\e(x) \varphi (x) \rosso{\xi}( \t_{x/\e}\tilde \o) \\
&=
- \e \int d  \nu ^\e_{\tilde \o} (x,z) \nabla_\e( \bar f_\e  \varphi)(x,z) b ( \t_{x/\e} \tilde \o, z)\,.
\end{split}
\end{equation*}
By \eqref{leibniz} we have
\begin{align*}
&  -\e \int d\nu ^\e_{\tilde \o} (x,z) \nabla_\e( \bar f_\e \varphi ) (x,z) b (\t_{x/\e} \tilde{\o}, z)= -\e C_1(\e)\rosso{-} \e C_2(\e)\,,\\
& C_1(\e):= \int d\nu ^\e_{\tilde \o} (x,z) \nabla_\e \bar f_\e (x,z) \varphi(  x) b (\t_{x/\e} \tilde{\o}, z)\,,\\
&  C_2(\e):=  \int d \nu ^\e_{\tilde \o} (x,z) \bar f_\e ( x+\e z  ) \nabla_\e \varphi( x,z) b (\t_{x/\e} \tilde{\o}, z)   \,.
\end{align*}
%
To get \eqref{chiavetta} we only need to show that $\lim_{\e \da 0} \e C_1(\e)=0$, $\lim_{\e \da 0} \e C_2(\e)=0$.

We start with $C_1(\e)$. By Schwarz inequality  and since $\tilde \o \in \O_{\rm typ}\subset   \cA _1 [ b^2  ] $
\bes
|C_1(\e)|\leq  \Big[ \int d\nu ^\e_{\tilde \o} (x,z) |\varphi(x)| \nabla _\e \bar f _\e (x,z) ^2\Big]^{1/2} 
\Big[ \int \mu^\e_{\tilde \o} (x) | \varphi(x) |  \widehat{\, b^2\, }(\t_{x/\e} \tilde\o)
\Big]^{1/2}\,.
\ens
Since $\tilde \o \in \O_{\rm typ}\subset   \cA _1 [ b^2  ] \cap \cA [  \widehat{\,b^2\,}  ]$, the last integral in the  r.h.s. converges to a finite constant as $\e \da 0$. 
It remains to prove that $ \int d\nu ^\e_{\tilde \o} (x,z) |\varphi(x)| \nabla _\e \bar f _\e (x,z) ^2 $
remains bounded from above as $\e \da 0$.  
We call $\ell $ the distance between  the support  of $\varphi$ (which is contained in $\D$ as $\varphi \in C^1_c(\D)$) and $\partial \D$.  Then, between the pairs $(x,z) $ with $x+\e z \not \in S$  contributing to the above integral, only the pairs 
$(x,z)$ such that $x\in \D$ and $|z| \geq \ell/\e$ can give a nonzero contribution.
In both cases $\D=\L$ and $\D=S$ we can estimate 
\be\label{olivia}
\begin{split}
\int d\nu ^\e_{\tilde \o} (x,z) |\varphi(x)|  \nabla _\e \bar f _\e (x,z) ^2  
& \leq \int d\nu ^\e_{\tilde \o,\D} (x,z) |\varphi(x)| \nabla _\e \bar f _\e (x,z) ^2\\
 &   +
\int d\nu ^\e_{\tilde \o} (x,z)| \varphi(x) |\nabla _\e \bar f _\e (x,z) ^2\mathds{1}(|z| \geq \ell /\e)\,.
\end{split}
\en
The first addendum in the r.h.s. of \eqref{olivia} is bounded due to \eqref{re3}. The second addendum goes to zero due to  \eqref{re1} (implying that $|\nabla_\e \bar f|\leq C/\e$ for small $\e$) and   Lemma \ref{lunghissimo}. Hence the l.h.s. of \eqref{olivia} remains bounded as $\e\da 0$.
This completes the proof that $\lim_{\e \da 0} \e C_1(\e)=0$.

 We move to $C_2(\e)$. Let $\phi$ be  as in \eqref{paradiso}. 
Using   \eqref{re1},  \eqref{paradiso} and afterwards \cite[Lemma~11.3-(i)]{Fhom3}, for some $\e$--independent constants $C$'s  (which can change from line to line), for $\e$ small   we can bound 
\be \label{crepe}
\begin{split}
 |C_2(\e)|  &\leq C
 \int d \nu ^\e_{\tilde \o} (x,z)\bigl | \nabla_\e \varphi( x,z) b (\t_{x/\e} \tilde{\o}, z)\bigr|  \\
 &\leq C \int d   \nu ^\e _{\tilde \o} (x,z) |z|\,|  b (\t_{x/\e}\tilde \o, z) | \bigl( \phi(x)+ \phi(x+\e z) \bigr) \\
 &\leq  C \int d  \nu ^\e_{\tilde \o} (x,z) 
 \phi(x) |z| (| b|+|\tilde b|)  (\t_{x/\e} \tilde{\o}, z) \\
 & \leq C \Big[ \int d  \nu ^\e_{\tilde \o} (x,z) 
 \phi(x)  |z| ^2 \Big]^{1/2} \Big[ 2 \int d  \nu ^\e_{\tilde \o} (x,z) 
 \phi(x)  (b^2 +\tilde b^2 )  (\t_{x/\e} \tilde{\o}, z)  \Big]^{1/2}\,.
 \end{split}
 \en
 The first integral in the last line of \eqref{crepe} equals $\int d \mu ^\e _{\tilde \o} (x)\phi(x)\l_2 (\t_{x/\e}\tilde \o)$. Since $\tilde \o \in \O_{\rm typ}\subset \cA_1[|z|^2] \cap \cA[\l_2]$, this integral converges to  a finite constant  as $\e\da0$.  The second  integral in the last line of  \eqref{crepe} equals 
$ \int d  \mu ^\e_{\tilde \o} (x) 
 \phi(x)  (\widehat{\, b^2\,} +\widehat{\,\tilde b^2\,} )  (\t_{x/\e} \tilde{\o})$
 as $\tilde \o \in \O_{\rm typ}\subset \cA_1[b^2 ]\cap  \cA_1[ \tilde{b}^2 ]$. 
 Since  also   $\tilde \o \in \O_{\rm typ}\subset \cA[  \widehat{\, b^2\,}   ]\cap \cA[\widehat{\,\tilde b^2\,} ] $, the last  integral converges
 to a finite constant. This implies  that $\lim_{\e \da 0} \e C_2(\e)=0$.
 \end{proof}

Due to  Proposition \ref{prova2} we can  write $v(x)$ instead of $v(x,\o)$, where $v$ is given by \eqref{totani1}. Recall Definition \ref{divido}. We extend  the notation \eqref{aria7}  for $\nabla_*$ also to functions with domain different from $\L$.

\begin{Proposition}\label{oro} Let $v$ and $w$ be as in \eqref{totani1} and \eqref{totani2}. Then it holds:
\begin{itemize}
\item[(i)]
 $v $ has weak derivatives $\partial_{\mathfrak{e_j}} v \in L^2(\D,dx)$ for  $j: 1\leq j \leq d_*$;  
\item[(ii)] $w(x,\o,z)= \nabla_* v (x) \cdot z   + v_1 (x,\o,z)$,
where $v_1\in L^2\bigl( \D, dx;L^2_{\rm pot} (\nu)\bigr )$.
\end{itemize}
\end{Proposition}

Above,  $L^2\bigl( \D, dx;L^2_{\rm pot} (\nu)\bigr )$ denotes the space of square integrable maps $f:\D \to L^2_{\rm pot} (\nu)$,
where $\D$ is endowed with the Lebesgue measure.

%
%
%

\begin{proof}
Given a square integrable form $b$, we define $\eta_b := \int d\nu (\o,z) z b(\o,z)$.
Note that $\eta_b$ is well defined  as $b\in L^2(\nu)$ and by (A7).
We observe that $\eta_b=- \eta_{\,\tilde b}$  as derived in \cite[Proof of Prop.~18.1]{Fhom3} by using (A3).
We claim that 
 for   each  solenoidal form $b  \in L^2_{\rm sol}(\nu) $ and each  function $\varphi \in C^2_c(\D)$, it holds 
 \begin{equation}\label{kokeshi}
 \int _\D dx \,  m \varphi(x) \int  d \nu (\o ,z) w(x,\o,z) b(\o, z) =  
  -\int  _\D dx\,m   v(x) \nabla \varphi(x) \cdot \eta_b\,.
\end{equation}

Having  \eqref{kokeshi} one can conclude the proof of Proposition~\ref{oro} by rather standard arguments. Indeed, having Corollary \ref{rock_bis},
it is enough to apply the same arguments presented in  \cite{Fhom3} to derive   \cite[Prop.~18.1]{Fhom3} from \cite[Eq.~(134)]{Fhom3}. The notation is similar, one has just to restrict to $x\in \D$.

\smallskip
Let us  prove \eqref{kokeshi}. Since both sides of \eqref{kokeshi} are continuous as functions of $b \in L^2_{\rm sol}(\nu)$, it is enough to prove it for $b\in \cW$. 
  Since $\tilde \o \in \O_{\rm typ}$,  along $\{\e_k\}$ it holds $\nabla_\e  f_\e \stackrel{2}{\toup} w$ as in  \eqref{totani2} and since  $b \in \cW \subset \cH$ (cf.~\eqref{yelena}) we can write 
  \be\label{bruna1}
  \begin{split}
  \text{l.h.s. of }\eqref{kokeshi}
  &= \lim _{\e \da 0} \int d \nu ^\e _{\tilde \o, \D}(x,z) \nabla_\e \bar f_\e(x,z) \varphi (x) b ( \t_{x/\e} \tilde \o, z)\,.
  \end{split}
  \en
  Since $b \in \cW\subset L^2_{\rm sol}(\nu)$ and   $\tilde \o \in \O_{\rm typ}\subset \cA_d[b] $, from    \cite[Lemma~11.7]{Fhom3} we get 
    \[
  \int d  \nu ^\e _{\tilde \o}(x,z) \nabla_\e ( \bar f_\e \varphi ) (x,z)  b ( \t_{x/\e} \tilde \o, z)=0\,.
  \]
Using the above identity and  \eqref{leibniz}, we get 
\begin{equation*}
\begin{split}
\int d  \nu ^\e _{\tilde \o}(x,z) &\nabla_\e \bar f_\e (x,z) \varphi (x) b ( \t_{x/\e} \tilde \o, z)=- \int d  \nu ^\e _{\tilde \o}(x,z)  \bar f_\e (x+ \e z) \nabla_\e \varphi (x,z) b ( \t_{x/\e} \tilde \o, z)\,.
\end{split}
\end{equation*}
As $\o \in \O_{\rm typ}\subset \cA_1[b] \cap \cA_1[\tilde b]$, 
by applying now \cite[Lemma~11.3-(ii)]{Fhom3} to the above r.h.s., we get
\begin{equation*}
\int d  \nu ^\e _{\tilde \o}(x,z) \nabla_\e \bar f_\e (x,z) \varphi (x) b ( \t_{x/\e} \tilde \o, z)=
 \int d  \nu ^\e _{\tilde \o}(x,z)\bar   f_\e (x) \nabla_\e \varphi (x,z) \tilde b ( \t_{x/\e} \tilde \o, z)
 \,.
 \end{equation*}
 By combining \eqref{bruna1}  with the above identity, we  obtain
\be\label{alba1}
  \text{l.h.s. of }\eqref{kokeshi}= \lim _{\e \da 0}  \left( -R_1(\e)+R_2(\e)
 \right)\,,
 \en 
  where 
  \begin{align*}
  & R_1(\e):=\int d \left[  \nu ^\e _{\tilde \o} -\nu ^\e _{\tilde \o, \D}\right](x,z) \nabla_\e \bar f_\e(x,z) \varphi (x) b ( \t_{x/\e} \tilde \o, z)\,,\\
& R_2(\e):=   \int d  \nu ^\e _{\tilde \o}(x,z)  \bar f_\e (x) \nabla_\e \varphi (x,z) \tilde b ( \t_{x/\e} \tilde \o, z)\,.
\end{align*}
  We claim that $\lim_{\e \da 0} R_1(\e)=0$.    
  We call $\ell $ the distance between  the support $\D_\varphi\subset \D $  of $\varphi$ and $\partial \D$.  Then in $R_1(\e)$ the contribution comes only from pairs $(x,z)$ such that $x\in \D_\varphi $ and $x+\e z \not \in S$ and therefore from pairs $(x,z)$ such that $x\in \D$ and $|z| \geq \ell/\e$:
  \[
   |R_1(\e)|\leq \int d \nu ^\e _{\tilde \o} (x,z) |\nabla_\e \bar f_\e(x,z) \varphi (x) b ( \t_{x/\e} \tilde \o, z) |\mathds{1}(x\in \D\,,\; |z| \geq \ell/\e)\,.
   \]
  By Schwarz inequality  we have therefore that 
 $
  R_1(\e)^2 \leq  I_1(\e) I_2(\e) $, where 
  \begin{align*}
  & I_1(\e):=\int d \nu ^\e _{\tilde \o} (x,z) \nabla_\e \bar  f_\e(x,z)^2 | \varphi (x)|\mathds{1}(|z| \geq \ell/\e) \,,\\
  & I_2(\e):= \int d \nu ^\e _{\tilde \o} (x,z) |\varphi(x)| b ( \t_{x/\e} \tilde \o, z)^2 =\int d\mu^\e_{\tilde \o}(x)| \varphi(x)| \widehat{\, b^2\,}(\t_{x/\e}\tilde \o) \,.
  \end{align*}
  The last identity concerning $I_2(\e)$ uses that 
  $\tilde \o
 \in \O_{\rm typ}\subset \cA_1[ b^2]$. 
 It holds $\lim_{\e\da 0} I_1(\e)=0$ due to  Lemma \ref{lunghissimo} and \eqref{re1}, while 
 $\varlimsup_{\e \da 0} I_2(\e)<+\infty$  since  $\tilde \o \in \O_{\rm typ}\subset  \cA[ \widehat{\, b^2\,}]$. This proves that $R_1(\e)\to 0$.  

\smallskip  
    We now move to $R_2(\e)$. To treat this term we will use the following fact:
    \begin{Claim}\label{mengoni} We have 
    \be\label{charaba}
 \lim _{\e \da 0} \int d   \nu ^\e _{\tilde \o}(x,z) \Big{|} \bar f_\e (x) \bigl[ \nabla_\e \varphi (x,z)-\nabla \varphi (x) \cdot z\bigr] \tilde b ( \t_{x/\e} \tilde \o, z)\Big{|}=0\,.
    \en
    \end{Claim}
     \begin{proof}[Proof of Claim \ref{mengoni}]  We fix $k \in\bbN$  and $\phi \in C_c(\bbR^d)$ with values in $[0,1]$  such that $\varphi (x)=0$ if $|x| \geq k$,  $ \phi(x)=1$ for $|x| \leq k$ and $\phi(x)=0$ for $|x| \geq k+1$. 
     Given $\ell \in \bbN$ we write the integral in \eqref{charaba} as $A_\ell (\e)+B_\ell(\e)$, where $A_\ell(\e)$ is the contribution coming from $z$ with $|z|\leq \ell$ and $B_\ell(\e)$  is the contribution coming from $z$ with $|z|> \ell$.
Due to \eqref{mirra}, \eqref{re1},  \cite[Lemma~11.3-(i)]{Fhom3} we have
\begin{equation*}
\begin{split}
A_\ell(\e) & \leq   C\ell^2  \e   \int  d  \nu ^\e _{\tilde \o}(x,z)  \bigl(  \phi (x) +\phi (x+\e z) \bigr) |\tilde b ( \t_{x/\e} \tilde \o, z)|\\
& \leq  C\ell^2  \e   \int  d  \nu ^\e _{\tilde \o}(x,z) \phi (x) (|b|+ |\tilde b|)  ( \t_{x/\e} \tilde \o, z)\,.
\end{split}
\end{equation*}
Since   $\o \in \O_{\rm typ}\subset \cA_1 [b]\cap  \cA_1 [\tilde{b}]$,  the last expression above  can be written as 
$C\ell ^2 \e   \int  d  \mu ^\e _{\tilde \o}(x) \phi (x) \bigl[ \widehat{ \, |b|\,}  +\widehat{ \, |\tilde b|\,}\bigr] ( \t_{x/\e} \tilde \o)
$.
 Since  $\o \in \O_{\rm typ}\subset \cA  [\widehat{\,| b|\, } ]\cap  \cA [\widehat{ \,  |\tilde b|\, }]$,  the above integral converges to a finite constant as $\e\da 0$. This allows to conclude that $\lim _{\e \da 0} A_\ell (\e)=0$.

It remains to prove that $\lim _{\ell \uparrow \infty} \varlimsup_{\e\da 0} B_\ell (\e)=0$. We argue  as above but now we apply  \eqref{paradiso}. 
Due to  \eqref{paradiso}, \eqref{re1} and \cite[Lemma~11.3-(i)]{Fhom3}, we can bound 
\[
B_\ell (\e) \leq C    \int  d  \nu ^\e _{\tilde \o}(x,z) \phi (x) (|b|+ |\tilde b|)  ( \t_{x/\e} \tilde \o, z) |z|\mathds{1}(|z|\geq \ell)\,.
\]
Recall \eqref{hesse1}. Then,  by Schwarz inequality $B_\ell (\e) \leq C\,  C_\ell (\e)^{1/2} D_\ell (\e)^{1/2}$, 
where (as $\tilde \o \in \O_{\rm typ}\subset \cA_1[b^2]\cap \cA_1[\tilde b ^2]$) 
\begin{align*}
 C_\ell (\e)& :=2\int  d  \nu ^\e _{\tilde \o}(x,z) \phi (x) (b^2 + \tilde b^2)  ( \t_{x/\e} \tilde \o, z)  = 2 \int  d  \mu ^\e _{\tilde \o}(x) \phi (x) \bigl( \widehat{b^2} +\widehat{ \tilde b^2}\bigr)  ( \t_{x/\e} \tilde \o)  \\
 D_\ell(\e)& := \int  d  \nu ^\e _{\tilde \o}(x,z) \phi (x)  |z|^2\mathds{1}(|z|\geq \ell)= \int  d  \mu ^\e _{\tilde \o}(x) \phi (x) f_\ell (\t_{x/\e}\tilde \o)\,.
\end{align*}   
As $\tilde \o \in \O_{\rm typ}\subset \cA[\widehat{b^2} ]\cap \cA[\widehat{ \tilde b^2}]\cap  \cA_1[|z|^2]\cap \cA[f_\ell]$,
we conclude that 
$
\varlimsup_{\e \da 0} B_\ell(\e) \leq C'  \bbE_0[\widehat{b^2} +\widehat{ \tilde b ^2}] ^{1/2}\bbE_0[ f_\ell]^{1/2}$,
and the r.h.s. goes to zero as $\ell \to \infty$ by dominated convergence and (A7).
\end{proof}
     \smallskip
      We come back to \eqref{kokeshi}.  By combining \eqref{alba1}, \eqref{charaba}  and the limit $R_1(\e)\to 0$, we conclude that 
     \be\label{alba2}
  \text{l.h.s. of }\eqref{kokeshi}= \lim _{\e \da 0} \int d   \nu ^\e _{\tilde \o}(x,z)  \bar f_\e (x) \nabla \varphi (x)\cdot z  \tilde b ( \t_{x/\e} \tilde \o, z)\,.
    \en
 Due to \eqref{alba2} and    since $\eta_{\tilde b }=- \eta_b$  as already observed, to prove \eqref{kokeshi} we only need to show that 
      \begin{equation}\label{tramontoA}
 \lim _{\e \da 0} \int d  \nu ^\e _{\tilde \o}(x,z) \bar  f_\e (x) \nabla \varphi (x) \cdot z \tilde b ( \t_{x/\e} \tilde \o, z)=
 \int  _{\rosso{\D}} dx \,m  v (x) \nabla \varphi(x) \cdot \eta_{\tilde b}\,.
 \end{equation}
 To this aim we observe that (recall \eqref{hesse2})
 \be\label{favorita1}
\int d  \nu^\e_{\tilde \o} (x,z) \bar f _\e (x) \partial _i \varphi (x) z_i \tilde{b} (\t_{x/\e} \tilde\o, z)=
\int d  \mu^\e_{\tilde \o}(x)  \bar f_\e (x) \partial _i \varphi (x) u_{\tilde b,i}( \t_{x/\e} \tilde \o)\,.
\en
We claim that 
 \be\label{favorita2}
 \lim_{\e \da 0}
 \int d   \mu^\e_{\tilde \o}(x)\bar   f_\e (x) \partial _i \varphi (x) u_{\tilde b,i}( \t_{x/\e} \tilde \o)= \int _\D dx \,m  v (x) \partial_i \varphi (x) 
\bbE_0[ u _{\tilde b,i}]
\,.
 \en
Since the r.h.s. equals   $\int _\D dx \,m  v (x) \partial_i \varphi (x) 
 (\eta_{\tilde b} \cdot e_i)$,   our target \eqref{tramontoA} then would  follow as a byproduct of \eqref{favorita1} and \eqref{favorita2}. 
 It remains therefore to prove \eqref{favorita2}. 
 Recall \eqref{hesse3}.
 Due to   Proposition \ref{prop_ergodico} (recall that $\tilde b\in \cW$ for any $b \in \cW$ and that 
 $\tilde \o\in \O_{\rm typ}\subset \cA[ u _{ b, i, n}]\cap \cA_1[\tilde b (\o,z) z_i]$ for all $b\in \cW$)
 \[ \lim _{\e \da 0}  \int d \mu^\e_{\tilde \o}(x)  | \partial_i \varphi (x)|  u _{\tilde b,i,n}  ( \t_{x/\e} \tilde \o)   =\int dx \, m | \partial_i \varphi (x)| \bbE_0[  u_{\tilde b,i,n}]\,.
 \]
 Since  $u_{\tilde b, i } \in L^1(\cP_0)$,  then the above r.h.s. goes to zero as $n\uparrow +\infty$. 
Hence, using also  \eqref{re1},  to get \eqref{favorita2} it is enough to show that 
  \be\label{favorita3}
  \begin{split}
\lim_{n\uparrow \infty}  \lim_{\e \da 0} &
 \int d  \mu^\e_{\tilde \o}(x)  \bar f_\e (x) \partial _i \varphi (x) u_{\tilde b,i }( \t_{x/\e} \tilde \o) \mathds{1}( |u_{\tilde b,i}( \t_{x/\e} \tilde \o)|\leq n)   =\int _\D dx \,m  v (x) \partial_i \varphi (x) 
\bbE[ u _{\tilde b,i}]
\,.
\end{split}
 \en
 Note that in \eqref{favorita3} we can replace
 $d  \mu^\e_{\tilde \o}(x)  \bar  f_\e (x) \partial _i \varphi (x)$ by $d  \mu^\e_{\tilde \o,\D}(x)   f_\e (x) \partial _i \varphi (x)$.
 Due to  \eqref{totani1} and since    $ u_{\tilde b,i} \mathds{1}( |u_{\tilde b,i}|\leq n) \in \cG$, by \eqref{rabarbaro} we have 
   \be\label{favorita4}
  \begin{split}
 \lim_{\e \da 0} &
 \int d \mu^\e_{\tilde \o,\D}(x)  f_\e (x) \partial _i \varphi (x) u_{\tilde b,i }( \t_{x/\e} \tilde \o) \mathds{1}( |u_{\tilde b,i}( \t_{x/\e} \tilde \o)|\leq n)  \\
 & =\int _\D dx \,  m v (x) \partial_i \varphi (x) 
\bbE[ u _{\tilde b,i}  \mathds{1}( |u_{\tilde b,i}|\leq n)   ]
\,.
\end{split}
 \en
By dominated convergence, we get \eqref{favorita3} from \eqref{favorita4}.
 \end{proof}


\section{2-scale limit  points of  $ V_\e$ and $\nabla_\e V_\e$}\label{limit_points}
We fix  $\tilde \o\in  \O_{\rm typ}$ and  we assume that $d_*={\rm dim} ({\rm Ker}(D)^\perp)\geq 1$. 
\rosso{As $0\leq V_\e \leq 1$ and $\O_{\rm typ}\subset \O_2$, the function $V_\e$ satisfies  \eqref{re1}, \eqref{re2} and \eqref{re3} with $\D=\L$ (cf.~Lemma \ref{paletta}). In particular,   
by} Lemmas  \ref{compatto1} and \ref{compatto2} along a subsequence $\e_k$ we have that  
\begin{align}
& L^2(\mu^\e_{\tilde \o, \L}) \ni V_\e \stackrel{2}{\toup} v  \in L^2(\L \times \O, m  dx \times \cP_0 ) \,,\label{totani1a}\\
& L^2(\nu^\e _{\tilde \o,\L})\ni   \nabla_\e V_\e \stackrel{2}{\toup} w\in   L^2(\L\times \O\times \bbR^d \,,\, m dx \times \nu ) \,, \label{totani2a}
\end{align}
for suitable functions  $v$ and $w$. \rosso{Moreover, $v$ and $w$ satisfy the properties stated in  Proposition \ref{prova2} and \ref{oro}. In particular, it holds $v=v(x)$}.
In the rest of this section, when considering the limit $\e\da 0$, we understand that $\e$ varies in the sequence $\{\e_k\}$. \rosso{Recall the function $\z$ defined in \eqref{def_psi}}.


\begin{Proposition}\label{diamanti} 
Let $v$  be as in  \eqref{totani1a}. Then     $v -\z_{|\L}  \in H^1_0(\L, F,d_*)$.
\end{Proposition}
\begin{proof} 
We apply the results of Section \ref{sec_bike} to the case $\D=S$ and $f_\e:= V_\e -\z$.
Since $f_\e$ is zero on $S\setminus \L$ and takes  values in $[-1,1]$ on $\L$, conditions \eqref{re1} and \eqref{re2} are satisfied.
In addition, we have $\nabla _\e f_\e (x,z)=0$ if $\{x,x+\e z\}$ does not intersect $\L$ and therefore  $\| f_\e\|_{L^2(\nu^\e_{\tilde \o, S})}= \| f_\e\|_{L^2(\nu^\e_{\tilde \o, \L})}
$.  By Lemma \ref{paletta}  we therefore conclude that also \eqref{re3} is satisfied.

At cost to refine the subsequence $\{\e_k\}$, by Lemmas \ref{compatto1} and \ref{compatto2} without loss of generality we can assume that along $\{\e_k\}$ itself  we have 
\begin{align}
&L^2(  \mu^\e_{\tilde \o, S})\ni  f_\e \stackrel{2}{\toup} \hat v  \in L^2(S \times \O,  m dx \times \cP_0 ) \,,\label{totani101}\\
& L^2(\nu^\e_{\tilde\o,S})\ni \nabla_\e{f}_\e \stackrel{2}{\toup} \hat w\in   L^2(S\times \O\times \bbR^d  \,,\, m dx \times \nu ) \,, \label{totani202}
\end{align}
for suitable functions $\hat v,\hat w$. 
By Proposition \ref{prova2} we have $\hat v= \hat v(x)$. Since $f_\e \equiv 0$ on $S\setminus \L$, it is simple to derive from the definition of $2$--scale convergence  that $\hat v(x) \equiv 0$ $dx$--a.e.  on $S\setminus \L$. 
Let us now prove that  also $\hat w(x,\cdot,\cdot)\equiv 0 $ $dx$--a.e. on $S\setminus \L$. To this aim we take $\varphi \in C_c(S\setminus \bar \L)$ and write $\ell$ for the distance between the support of $\varphi$ and $\L$. As $\|f_\e\|_\infty\leq 1$ and since $f_\e \equiv 0$ on $S\setminus \L$, we have 
\[ \bigl| \int d\nu^\e_{\tilde \o,S}  (x,z)  \varphi(x) \nabla_\e f_\e(x,z)  \bigr| \leq  \e^{-1} \int d\nu^\e_{\tilde \o,S}  (x,z) | \varphi(x)| \mathds{1}(|z|\geq \ell /\e) \,.
\]
By Lemma \ref{lunghissimo} the r.h.s. goes to zero as $\e \da 0$. By the above observation, \eqref{yelena} and \eqref{totani202}, we conclude that for all $\varphi \in C_c(S\setminus \bar \L)$  and $b$ in  $  \cH_2 $ or  $b = 
h\mathds{1}(|h|\leq n)$ for some $h\in \cW$ and $n\in \bbN$ (in both cases $b\in L^\infty(\nu)$ and $b\in \cH$),  it holds
\be\label{radio}
\int _S dx  \int d \nu (\o,z) \varphi(x)  \hat w (x,\o,z) b(\o,z) =0\,.
\en
When $b=h\mathds{1}(|h|\leq n)$, by taking the limit $n\uparrow \infty$ and using dominated convergence, we  conclude that \eqref{radio} holds also for $b=h$. As $\cH_2$ is dense in $L^2_{\rm pot}(\nu)$ and $\cW$ is dense in $L^2_{\rm sol}(\nu)$, \eqref{radio} implies that 
$\hat w(x,\cdot,\cdot)\equiv 0 $ $dx$--a.e. on $S\setminus \L$.

\smallskip
We recall that in the proof of Proposition~\ref{oro} we have in particular derived \eqref{kokeshi}:
 for   each  solenoidal form $b  \in L^2_{\rm sol}(\nu) $ and each  function $\varphi \in C^2_c(S)$, it holds 
 \begin{equation}\label{kokeshi74}
 \int _S dx \,  \varphi(x) \int  d \nu (\o ,z) \hat w(x,\o,z) b(\o, z) =  
  -\int  _S dx\,   \hat v(x) \nabla \varphi(x) \cdot \eta_b\,.
\end{equation}
As  $\hat v(x) \equiv 0$ $dx$--a.e.  on $S\setminus \L$ and  $\hat w(x,\cdot,\cdot)\equiv 0 $ $dx$--a.e. on $S\setminus \L$,  \eqref{kokeshi74} implies that 
\be\label{anatra}
\begin{split}
&  \Big |
\int  _{\L} dx\,  \hat  v(x) \nabla \varphi(x) \cdot \eta_b
\Big|
 = \Big|\int _\L dx \,  \varphi(x) \int  d \nu (\o ,z)\hat  w(x,\o,z) b(\o, z)\Big|\,.
 \end{split}
\en
By Schwarz inequality  we can bound
\begin{equation*}
\begin{split}
 C^2:&=\int _\L dx \Big[ \int  d \nu (\o ,z) \hat w(x,\o,z) b(\o, z) \Big]^2  \\
 &\leq 
  \int _\L dx  \int  d \nu (\o ,z) \hat  w(x,\o,z)  ^2 \int d \nu (\o,z) b(\o,z)^2\\
  & = \| \hat w \| ^2 _{L^2( \L\times \O\times \bbR^d, dx\times \nu) }\|b \|^2_{L^2(\nu)}<+\infty\,.
 \end{split}
 \end{equation*}
 By applying now Schwarz inequality to the r.h.s. of \eqref{anatra} we conclude that 
$
  |
\int  _{\L} dx\,   \hat v(x) \nabla \varphi(x) \cdot \eta_b
| \leq C \| \varphi\|_{L^2(\L,dx)}$.
Due to  Corollary \ref{rock_bis}  for each $i:1\leq i \leq d_*$   there exists $b\in L^2_{\rm sol}(\nu)$ such that $\eta_b=\mathfrak{e}_i$. As a byproduct, we get    $|
\int  _{\L} dx\,   \hat v(x) \partial_{\mathfrak{e}_i}  \varphi(x)
| \leq C \| \varphi\|_{L^2(\L,dx)}$.
The above bound and   Proposition \ref{monti} imply that $\hat v \in H^1_0(\L, F,d_*)$.
To conclude it remains to observe that
 $\hat v= v-\z_{|\L}$ $dx$--a.e. on $\L$. 
\end{proof}

\begin{Proposition}\label{prova1} Let $w$  be as in \eqref{totani2a}. For $dx$--a.e. $x\in \L$, the map  $(\o,z) \mapsto w(x,\o,z)$ belongs to $L^2_{\rm sol}(\nu)$.
\end{Proposition}
\begin{proof}  As common when proving similar statements,
  we take  as test function $u(x):=\e \varphi(x) g(\t_{x/\e} \tilde \o)$, where $\varphi \in C_c (\L)$ and  $g\in \cG_2$ (cf. Section \ref{sec_tipetto}). We use that $\la  \nabla_\e u,\nabla_\e V_\e \ra _{ L ^2( \nu ^\e_{\tilde\o,\L}) }=0$ (cf. Remark~\ref{sirenetta}). Due to \eqref{leibniz} \rosso{and since $\nabla_\e g(\t_{\cdot/\e} \tilde\o)(x,z)=\e^{-1}\nabla g (\t_{x/\e}\tilde\o,z)$}, this identity 
can be rewritten as
\begin{equation}\label{cioccolata}
\begin{split}
& \e \int d  \nu ^\e_{ \tilde \o,\L} (x,z)  \nabla _\e \varphi (x,z) g(\t_{z+x/\e} \tilde \o) \nabla _\e V_\e (x,z) \\
& +  \int d  \nu ^\e_{ \tilde \o,\L}  (x,z)  \varphi(x) \nabla g ( \t_{x/\e}\tilde  \o, z)\nabla_\e  V_\e (x,z)  =0\,.
\end{split}
\end{equation}
We first show that the first integral in \eqref{cioccolata} remains uniformly bounded as $\e \da 0$.
 By applying Schwarz inequality,
using that  $g$ is bounded as $g\in \cG_2$ and that  $\varlimsup_{\e\da 0}
\| \nabla_\e V_\e \| _{L^2 (\nu ^\e_{\tilde \o,\L} )}<+\infty $ due to \eqref{parasole} and since $\tilde \o \in \O_{\rm typ}\subset \O_2$, it is enough to show that  $\varlimsup_{\e\da 0}
\| \nabla_\e \varphi    \| _{L^2 (\nu ^\e_{\tilde \o,\L} )}<+\infty $.  As $\tilde \o \in \O_{\rm typ}$,  by Lemma \ref{blocco}   it remains to prove that $\varlimsup_{\e\da 0}
\|\nabla \varphi (x)\cdot z  \| _{L^2 (\nu ^\e_{\tilde \o,\L} )}<+\infty $,  which follows from the fact that
 $\tilde \o \in \O_{\rm typ}\subset \cA_1[|z|^2]\cap \cA[\l_2]$.

 Coming back to \eqref{cioccolata}, using  that  the first addendum goes to zero as $\e \da 0$  and applying the 2-scale convergence  $\nabla_\e V_\e \stackrel{2}{\toup} w$ in \eqref{totani2a}   to  treat the second addendum,  we conclude that 
$\int _\L dx \int d\nu (\o,z) \varphi(x) \nabla g (\o, z) w (x,\o,z) =0$ for all  $ g \in \cG_2$
(note that  we have used   \eqref{yelena} as  $\nabla g \in \cH_2\subset \cH$).
Since $\{\nabla g \,:\, g \in \cG_2\}$ is dense in $L^2_{\rm pot}(\nu)$,  the above identity implies that, for $dx$--a.e. $x\in \L$, the map  $(\o,z) \mapsto w(x,\o,z)$ belongs to $L^2_{\rm sol}(\nu)$. 
\end{proof}

%
%
%

\section{Proof of   Theorem \ref{teo1} and  Proposition \ref{teo2} 
\rosso{for $e_1\in {\rm Ker}(D)^\perp$}} \label{limitone}
\rosso{In this section we assume that $e_1\in {\rm Ker}(D)^\perp$}.  
 We recall that at the end of Section \ref{sec_tipetto} we have  proved that $\O_{\rm typ}$ is a measurable translation invariant subset of $\O_1$ with  $\cP(\O_{\rm typ})=1$.

\subsection{Proof of  Proposition \ref{teo2} \rosso{for $e_1\in {\rm Ker}(D)^\perp$}}\label{limitone_p1} \rosso{We fix $\tilde \o \in \O_{\rm typ}$}.
As $ \O_{\rm typ}\subset \O_2$ and due
 to Lemmas \ref{compatto1} and \ref{compatto2},
along a subsequence $\{\e_k\}$  we have
that  
$L^2(\mu^\e_{\tilde \o, \L})   \ni V_\e \stackrel{2}{\toup} v  \in L^2(\L \times \O,  m dx \times \cP_0 ) $ and $L^2(\nu^\e_{\tilde \o, \L})\ni \nabla_\e V_\e \stackrel{2}{\toup} w\in   L^2(\L\times \O\times \bbR^d\,,\, m dx \times \nu ) $ (cf. \eqref{totani1a} and \eqref{totani2a}). From Proposition~\ref{prova1} it is standard to conclude that  for  $dx$--a.e. $x\in \L$ it holds 
\begin{equation}\label{mattacchione}
\int d \nu(\o, z)  w(x,\o, z)  z= 2 D \nabla_* v (x)\,.
\end{equation}
Indeed, by Proposition~\ref{prova1} for $dx$--a.e. $x\in \L$, the map  $(\o,z) \mapsto w(x,\o,z)$ belongs to $L^2_{\rm sol}(\nu)$.
 By Proposition~\ref{oro} we know that 
$w(x,\o,z)= \nabla_*  v (x) \cdot z + v_1(x,\o,z)$, 
where $v_1\in L^2\bigl( \L,\rosso{dx;}\, L^2_{\rm pot} (\nu)\bigr )$. Hence, by \eqref{jung}, for $dx$--a.e. $x\in \L$ we have that 
$v_1 (x,\cdot,\cdot ) = v^a $, where  $ a:= \nabla_* v (x)$.
As a consequence (using also \eqref{solare789}),  for $dx$--a.e. $x\in \L$, we  can rewrite the l.h.s. of \eqref{mattacchione} as
$ \int d \nu(\o,z) z [  \nabla_* v (x)\cdot z + v^{ \nabla_*  v (x)}(\o,z)]= 2 D \nabla_* v (x)$,
thus proving \eqref{mattacchione}.

We now take a function  $\varphi \in C^2_c(\bbR^d)$ which is zero on $S\setminus \L$ (note that we are not taking $\varphi \in C^2_c (S)$). By Remark \ref{sirenetta}
 we have the identity $\la \nabla _\e \varphi, \nabla_\e V_\e\ra_{L^2(\nu^\e_{\tilde \o,\L})}=0$. The above identity, \eqref{parasole} and Lemma \ref{blocco} (use that $\tilde \o \in \O_{\rm typ}$)  imply that 
\be\label{rigogolo}
 0=\la \nabla _\e \varphi, \nabla_\e V_\e\ra_{L^2(\nu^\e_{\tilde \o,\L})}= \int d\nu^\e_{\tilde \o,\L} (x,z) \nabla \varphi(x) \cdot z \nabla_\e V_\e(x,z)+o(1)\,.
\en
For each positive $n\in \bbN$   let $A_n:= (1-1/n) \L$ and let $\phi_n\in C_c( \L)$ be a function with values in $[0,1]$ such that $\phi_n\equiv 1$ on $A_n$.
By Schwarz inequality
\begin{multline} \label{manovra}
\Big| \int d\nu^\e_{\tilde \o,\L} (x,z) (\phi_n(x)-1)\nabla \varphi(x) \cdot z \nabla_\e V_\e(x,z)
\Big|
\\
\leq  \| \nabla \varphi \|_\infty  \| \nabla_\e V_\e\|_{L^2(\nu^\e_{\tilde \o,\L})}
\Big[
\int _{\L \setminus A_n} d\mu^\e_{\tilde \o} (x) \l_2 (\t_{x/\e} \tilde \o)\Big]^{1/2}\,.
\end{multline}
Since $\tilde \o \in \O_{\rm typ} \subset \cA[\l_2]\cap \cA_1[|z|^2]$,   we get that
$\lim_{\e\da 0} \int _{\L \setminus A_n} d\mu^\e_{\tilde \o}(x) \l_2 (\t_{x/\e} \tilde \o)= \ell(\L \setminus A_n )\bbE_0[\l_2]$.  As a byproduct with  \eqref{parasole} we get $
\lim_{n\uparrow \infty} \varlimsup_{\e\da 0} \text{ l.h.s. of } \eqref{manovra}=0$.
Due to \eqref{rigogolo}  we then obtain 
\be\label{kipur65}
\lim_{n \uparrow \infty} \varlimsup_{\e \da 0} \int d\nu^\e_{\tilde \o,\L} (x,z) \phi_n(x) \nabla \varphi(x) \cdot z \nabla_\e V_\e(x,z)=0\,.
\en

On the other hand,  due to \eqref{totani2a} and since  $\tilde \o \in \O_{\rm typ}$ (recall that the form $(\o,z) \mapsto z_i $ belongs to $\cH$ and  apply \eqref{yelena}), we can rewrite \eqref{kipur65} as 
\be\label{latte32}
\lim_{n \uparrow \infty}   \int _\L dx \int d\nu (\o,z)  \phi_n(x) \nabla \varphi(x) \cdot z  w (x,\o,z)=0\,.
\en
By Schwarz inequality,  the above integral differs from the same expression with $\phi_n(x)$ replaced with $1$ by at  most  $\| w\|  \| \phi_n-1\|_{L^2(\L, dx) } \|\nabla\varphi \|_\infty \bbE[\l_2]^{1/2}$, where $\|w\|$ is the norm of $w$ in $L^2(\L\times \O\times \bbR^d \,,\, m dx \times \nu )$. Hence, due to  \eqref{latte32},  $   \int _\L dx \int d\nu (\o,z)   \nabla \varphi(x) \cdot z  w (x,\o,z)=0$.
As a byproduct with  \eqref{mattacchione}, we conclude that $0= \int_\L dx \nabla \varphi (x) \cdot D \nabla_* v(x)= \int_\L dx \nabla_* \varphi (x) \cdot D \nabla_* v(x) $
 for any 
 $\varphi \in  C^2_c(\bbR^d)$ with $\varphi\equiv 0$ on  $S\setminus \L$ (we write $\varphi \in \cC$). 
 \rosso{Let us now take $ \varphi \in C^\infty _c(\bbR^d \setminus F)$. We call $\g>0$ the distance between the support of $\varphi$ and $F$ and we  fix a function $\theta:\bbR\to\bbR$ such that $\theta(u)=0$ for $|u|\geq 1/2+\g$ and $\theta(u)=1$ for $|u|\leq 1/2$. Then $\tilde \varphi(x):=u(x_1) \varphi(x) \in \cC$ and $  
 \varphi_{|\L}=\tilde{\varphi}_{|\L}$. Therefore, by the previous observations, $0= \int_\L dx \nabla_* \varphi (x) \cdot D\nabla_* v(x) $ and by density the same holds}
   for any $\varphi \in H_0^1(\L,F,d_*)$.
 Due to Proposition  \ref{diamanti} we  also have that $v\in K$ (cf. \eqref{kafka} in Definition~\ref{vettorino}).  
 Hence, \rosso{$v$ is the unique weak solution   of the equation $ \nabla_* \cdot ( D \nabla_* v ) =0$ 
with boundary conditions  \eqref{mbc}, i.e. $v=\psi$}.

\rosso{From Lemmas \ref{compatto1} and \ref{compatto2} one easily derives the same converging result
stating there but starting from  a sequence of functions parametrized by $\e_k\da 0$ (the convergence is along a subsequence). Also dealing with sequences and subsequences, by the above results the limit point is always $\psi$. Hence} we get that $V_\e  \in L^2(\mu^\e_{\tilde \o, \L}) $ weakly 2-scale converges to  $\rosso{\psi} \in   L^2(\L \times \O, m  dx \times \cP_0 ) $ as $\e\da 0$. As $\rosso{\psi}$ does not depend from $\o$ and since $1\in \cG$, we derive from \eqref{rabarbaro} that $L^2(\mu^\e_{\tilde \o, \L})\ni V_\e \toup \rosso{\psi} \in L^2( \L, m dx)$ according to Definition \ref{debole_forte}.  \rosso{This concludes the proof of Item (i) in Proposition~\ref{teo2}. Having now identified $v$,  the discussion following \eqref{mattacchione} with $v=\psi$ leads to  Item (ii)}.

\subsection{Proof of Theorem \ref{teo1} \rosso{for $e_1\in {\rm Ker}(D)^\perp$}}\label{platano} Let us show that, given  
 $\tilde \o \in \O_{\rm typ}$, it holds
$\lim_{\e\da 0}  \frac{1}{2} \la \nabla_\e V_\e, \nabla _\e V_\e \ra 
_{L^2(\nu^\e_{\tilde \o, \L})}=m\rosso{\int_\L D \nabla_* \psi (x) \cdot \nabla_* \psi(x)dx}
$ (cf.~\eqref{mortisia}).
To this aim we  apply Remark \ref{sirenetta} to get that  $\la \nabla_\e( V_\e - \z), \nabla_\e V_\e \ra_{L^2(\nu^\e_{\tilde \o, \L})} =0$. This implies that 
\be \label{antiguaz}
\la \nabla_\e V_\e, \nabla_\e V_\e\ra_{L^2(\nu^\e_{\tilde \o, \L})}  = \la \nabla_\e \z, \nabla_\e V_\e \ra_{L^2(\nu^\e_{\tilde \o, \L})}\,.
\en
\begin{Claim} \label{marlena93}
It holds $\lim _{\e\da 0} \int d \nu^\e_{\tilde \o, \L}(x,z) |\nabla_\e \z(x,z)- z_1 |^2 =0$.
\end{Claim} 
\begin{proof} If $x, x+\e z\in \L$, then $\nabla_\e \z(x,z)=z_1$.  We have only 4 relevant alternative cases: 
(a) 
$x\in \L$, $x+\e z\in S^+$; (b) $x\in S^+$, $x+\e z\in \L$; (c) $x\in \L$, $x+\e z\in S^-$; (d) $x\in S^-$, $x+\e z\in \L$.
Below we treat only \rosso{cases (a) and (b)}, since the other cases can be treated similarly. \rosso{We first} assume (a) to hold. Then $x_1+\frac{1}{2}=\z(x) \leq \z (x+\e z) \leq x_1+\e z_1+\frac{1}{2}$ and therefore $0\leq \nabla_\e \z (x,z)  \leq z_1$. This implies that $|\nabla_\e \z(x,z)- z_1 |^2\leq z_1^2$.
Fix $\d \in (0,1/2)$ and set $\L_\d:=(-1/2+\d, 1/2-\d)^{d}$. 
Given $n\in \bbN$, for $\e$ small we can bound
\be\label{mary1}
\begin{split}
  & \int d \nu^\e_{\tilde \o, \L}(x,z) |\nabla_\e \z(x,z)- z_1 |^2 \mathds{1}( x\in \L_\d, x+\e z\in S^+
 )\\
& \leq  \int  d \nu^\e_{\tilde \o}(x,z) z_1^2 \mathds{1}(x\in \L_\d , z_1 \geq \d /\e)
\leq   \int_{\L_\d} d \mu^\e_{\tilde \o} (x)   f_n(\t_{x/\e}\tilde\o) \,\rosso{,}
\end{split}
\en
where  $f_n(\o)= \sum _{z\in \hat\o} c_{0,z}|z|^2  \mathds{1}(|z| \geq n)$ (recall \eqref{hesse1}). By Proposition \ref{prop_ergodico} and since $ \O_{\rm typ}\subset \cA[f_n]\cap \cA_1[|z|^2]$, 
as $\e \da 0$ the last integral in \eqref{mary1} converges to 
$m(1-2\d)^d \bbE_0[f_n]
$, which goes to   zero as $n\uparrow +\infty$ due to (A7). This allows to conclude that   the l.h.s. of \eqref{mary1} converges to zero as $\e \da 0$.

We can bound
\be \label{mary2}
\begin{split}
 & \int d \nu^\e_{\tilde \o, \L}(x,z) |\nabla_\e \z(x,z)- z_1 |^2 \mathds{1}( x\in\L\setminus  \L_\d, x+\e z\in S^+
 )
\\
& \leq  \int d \nu^\e_{\tilde \o}(x,z) z_1^2 \mathds{1}( x\in \L\setminus \L_\d) \leq  \int _{\L\setminus \L_\d} d \mu^\e_{\tilde \o}(x) \l_2(\t_{x/\e}\tilde \o)\,.
\end{split}
\en
By Proposition \ref{prop_ergodico} and since  $\O_{\rm typ}\subset \cA[\l_2]\cap \cA_1[|z|^2]$, as $\e\da 0$ the last integral in \eqref{mary2} converges to    $\ell(\L\setminus \L_\d) \bbE_0[\l_2]$, which goes to zero as $\d\da 0$ by (A7). 

The above results allow to conclude \rosso{that the  contribution to $\int d \nu^\e_{\tilde \o, \L}(x,z) |\nabla_\e \z(x,z)- z_1 |^2 $ of $(x,z)$ as in  case (a)  is negligible as $\e \da 0$. Let us prove the same result with $(x,z)$ as in  case (b): $x\in S^+$, $x+\e z\in \L$.  Since 
$x_1+\e z_1+\frac{1}{2}= \z (x+\e z)\leq \z(x)=1\leq x_1+\frac{1}{2}$, we have  $z_1\leq \nabla_\e \z(x,z) \leq 0$ and therefore $|\nabla_\e \z(x,z)- z_1 |^2\leq z_1^2$. On the other hand,  by  \cite[Lemma~11.3]{Fhom3} we get 
\be
\int d \nu^\e_{\tilde \o, \L}(x,z) \mathds{1}(x\in S^+) \mathds{1}(x+\e z \in \L)z_1^2=
\int d \nu^\e_{\tilde \o, \L}(x,z) \mathds{1}(x+\e z \in S^+) \mathds{1}(x \in \L)z_1^2\,,
\en
and the last term goes to zero as $\e\da 0$ due to the estimates used to treat case (a).}
\end{proof}
As a  byproduct of Claim \ref{marlena93}, Lemma \ref{paletta}  and \eqref{antiguaz}, we get 
\be\label{antigua}
\lim_{\e\da 0}\la \nabla_\e V_\e, \nabla_\e V_\e\ra_{L^2(\nu^\e_{\tilde \o, \L})}  =   \lim_{\e \da 0} \int d \nu^\e_{\tilde \o,\L} (x,z) z_1 \nabla_\e V_\e (x,z)\,.
\en
By applying Schwarz inequality as in \eqref{manovra}, we get that \be\label{antigua_bis}
\lim_{\e \da 0}  \int d \nu^\e_{\tilde \o,\L} (x,z) z_1 \nabla_\e V_\e (x,z)=
\lim_{n \uparrow \infty}
\lim_{\e \da 0} \int d \nu^\e_{\tilde \o,\L} (x,z) \phi_n(x) z_1 \nabla_\e V_\e (x,z)\,.
\en

\rosso{By Proposition~\ref{teo2}  $L^2(\mu^\e_{\tilde \o, \L})   \ni V_\e \stackrel{2}{\toup} \rosso{\psi}  \in L^2(\L \times \O,  m dx \times \cP_0 ) $ and $L^2(\nu^\e_{\tilde \o, \L})\ni   \nabla_\e V_\e \stackrel{2}{\rightharpoonup}w  \in  L^2(\L \times \O\times\bbR^d,  m dx \times \nu) $}. Since $\phi_n \in C_c(\L)$,  
as a byproduct of \eqref{antigua} and \eqref{antigua_bis} we obtain that 
\be\label{emoji}
\begin{split}
\lim_{\e\da 0}  \la \nabla_\e V_\e, \nabla_\e V_\e\ra _{L^2(\nu^\e_{\tilde \o, \L})} &  =   \lim _{n\uparrow \infty} \int _{\L} dx\,   m
\phi_n(x) \int d\nu  (\o,z) z_1   w(x,\o,z)\\
& =    \int _{\L} dx\,  m
 \int d\nu  (\o,z) z_1  w(x,\o,z)\,.
\end{split}
\en
Due to \eqref{mattacchione} \rosso{and since $v$ there equals $\psi$ due to Proposition~\ref{teo2}}, the last term  \rosso{in \eqref{emoji}} equals 
\[
\rosso{
m \int_\L 2  D \nabla_* \psi (x) \cdot e_1 dx= 
2 m \int_\L   D \nabla_* \psi (x) \cdot \nabla_* \z (x) dx}\,.
\]
\rosso{To conclude the proof of Theorem \ref{teo1}, it is enough to show that 
\be\label{finalissima}
\int_\L  D \nabla_* \psi(x) \cdot \nabla_* \z (x) dx=\int_\L  D \nabla_* \psi (x) \cdot \nabla_* \psi (x) dx\,.
\en
Since $\psi \in K$, $f:= \psi-\z$  belongs to $H^1_0(\L,F,d_*)$. As a byproduct of this observation and  Definitions \ref{fete} and \ref{rondini25}, we conclude that  $\int_\L  D \nabla_* \psi(x) \cdot \nabla_* f (x) dx=0$, which is equivalent to \eqref{finalissima}.
}

\appendix

\section{Canonical procedure for the reduction to the case $\bbG=\bbR^d$}\label{app_porto}
\rosso{In this appendix we give some more details on the procedure mentioned in Section \ref{sec_riduzione}.}
 Let us consider the context of Theorem \ref{teo3} with $\bbG=\bbZ^d$. Again, without loss of generality, we take $\t_g x=x+g$ for all $x\in \bbR^d$ and $g\in \bbZ^d$.  Similarly, we have the translation action of the group $\bbR^d$ given by  $\bar \t_g x:= x+g$ for all $x,g\in \bbR^d$.   We consider the probability space $\bar\O:=\O\times [0,1)^d$ with probability $\bar \cP:= \cP\times dx$  and  $\s$--field $\bar{\cF}:=\cF \otimes \cB([0,1)^d)$, $\cB([0,1)^d)$ being the family of  Borel sets of $[0,1)^d$. 
 Given $y\in \bbR^d$, we define the integer  part $[y]$  of $y$ as  the unique element $z\in \bbZ^d$ such that $y\in z+[0,1)^d$. We then set $\b(y):=y-[y]\in [0,1)^d$. The group $\bbR^d$ acts on $\bar \O$ by means of the maps $(\bar \t_g) _{g\in \bbR^d}$ with  $\bar \t _g (\o,a):= ( \t_{[g+a]} \o, \b (g+a))$.
 Finally, to $(\o,a)\in \bar \O$ we associate the locally finite set $\widehat{(\o,a)} :=\hat \o-a$ and consider the  conductances $\bar c_{x-a,y-a}(\o,a):= c_{x,y} (\o)$ for $x,y\in \hat \o$.
%
Then, under the assumptions of Theorem \ref{teo3}, by    the observations  collected in \cite[Section~6]{Fhom3}   the setting given by $(\bar\O, \bar \cF, \bar{\cP})$, the group $\bar \bbG:=\bbR^d$ with the above two actions, the simple point process $\widehat{\bar{\omega}}$ and  the conductance field $\bar c_{x,y}(\bar \o)$ satisfies assumptions (A1),...,(A8).
Hence, once proven  Theorem \ref{teo1} and Proposition \ref{teo2} for $\bbG=\bbR^d$,  
\rosso{one can apply them to the above case getting the desired convergences 
for all environments in a translation invariant  measurable set $\bar \O_{\rm typ}$ (for the $\bbR^d$--action)  with $\bar \cP({\bar\O}_{\rm typ})=1$.}
As discussed in  \cite[Section~6]{Fhom3},  $ \bar m  =m$ and $\bar D = D$.
  On the other hand, since $ {\bar \O}_{\rm typ}$ is translation invariant, for $(\o,a)\in \bar\O_{\rm typ}$ we have $(\o,0)= \bar \t_{-a} (\o,a)\in   {\bar \O}_{\rm typ}$. Note also that $\rosso{\bar\s}_\ell(\o,0) = \s_\ell (\o)$ as $\widehat{(\o,0)}=\hat\o$ and $\bar c_{x,y}(\o,0)=c_{x,y}(\o)$ for $x,y\in \hat \o$.  
  Hence, if $(\o,a) \in {\bar \O}_{\rm typ}$, then  \rosso{$\lim _{\ell \to +\infty} \ell^{2-d} \s_\ell( \o)= \lim _{\ell \to +\infty} \ell^{2-d} \bar\s_\ell( \o,a)$}. One can then check that Theorem \ref{teo1} is fulfilled with  $\O_{\rm typ}:=\{ \o \in \O\,: (\o,a) \in \bar \O_{\rm typ} \text{ for some } a\in [0,1)^d\}$ ($\O_{\rm typ}$ is translation invariant as $(\t_g \o, a)=\bar \t_g (\o,a)$ for all $g\in \bbZ^d$, while the issue concerning $\O_1$ can be settled using  that $\bar \l_k (\o,a)= \l_k (\o)$ as observed
   in  \cite[Section~6]{Fhom3}). By similar arguments one obtains Proposition \ref{teo2}.

\bigskip
{\bf Acknowledgments} \rosso{I thank the anonymous referee for his/her corrections and  very stimulating comments and suggestions.}  I thank  Paul Dario  and Pierre Mathieu for some  useful comments. I thank the numerous colleagues with whom I have discussed some references. \rosso{I thank Pierpaolo Gabrielli for providing some pictures}. 
 I thank my sons  Lorenzo and Pierpaolo  for their help on 1st July 2021. \rosso{This work is first of all for them}.




\end{document}